\documentclass[%
reprint,
superscriptaddress,
 amsmath,amssymb,
 aps,
 pra,
]{revtex4-2}

\usepackage{xcolor}

\usepackage{graphicx}
\usepackage{dcolumn}
\usepackage{bm}
\usepackage{comment}
\usepackage[colorlinks=true]{hyperref}

\usepackage{amsmath}
\usepackage{amssymb}
\usepackage{mathtools}
\allowdisplaybreaks

\usepackage[T1]{fontenc}

\usepackage{amsthm}
\newtheorem{theorem}{Theorem}[section]

\newtheorem{proposition}[theorem]{Proposition}

\makeatletter
\newsavebox{\@brx}
\newcommand{\llangle}[1][]{\savebox{\@brx}{\(\m@th{#1\langle}\)}%
  \mathopen{\copy\@brx\kern-0.5\wd\@brx\usebox{\@brx}}}
\newcommand{\rrangle}[1][]{\savebox{\@brx}{\(\m@th{#1\rangle}\)}%
  \mathclose{\copy\@brx\kern-0.5\wd\@brx\usebox{\@brx}}}
\makeatother

\begin{document}

\preprint{APS/123-QED}

\title{Quantum Advantage in Distributed Sensing with Noisy Quantum Networks}

\author{Allen Zang}
\email{yzang@uchicago.edu}
\affiliation{Pritzker School of Molecular Engineering, University of Chicago, Chicago, IL, USA}

\author{Alexander Kolar}
\affiliation{Pritzker School of Molecular Engineering, University of Chicago, Chicago, IL, USA}

\author{Alvin Gonzales}
\affiliation{Mathematics and Computer Science Division, Argonne National Laboratory, Lemont, IL, USA}

\author{Joaquin Chung}
\affiliation{Data Science and Learning Division, Argonne National Laboratory, Lemont, IL, USA}

\author{Stephen K. Gray}
\affiliation{Center for Nanoscale Materials, Argonne National Laboratory, Lemont, IL, USA}

\author{Rajkumar Kettimuthu}
\affiliation{Data Science and Learning Division, Argonne National Laboratory, Lemont, IL, USA}

\author{Tian Zhong}
\affiliation{Pritzker School of Molecular Engineering, University of Chicago, Chicago, IL, USA}

\author{Zain H. Saleem}
\email{zsaleem@anl.gov}
\affiliation{Mathematics and Computer Science Division, Argonne National Laboratory, Lemont, IL, USA}

\date{\today}

\begin{abstract}
    We show that quantum advantage in distributed sensing can be achieved with noisy quantum networks which only distribute noisy entangled states. We derive a closed-form expression of the quantum Fisher information (QFI) for estimating the average of local parameters using GHZ-diagonal probe states, a representative distributed sensing scenario. From the QFI we obtain the necessary condition to achieve quantum advantage over the optimal local sensing strategy, which can also serve as an optimization-free entanglement detection criterion for multipartite states. We further explore the impacts from imperfect local entanglement generation and local measurement constraint, and our results imply that the quantum advantage is more robust against quantum network imperfections than local operation errors. Notably, these implications still hold when we explicitly consider dephasing during the sensing dynamics. Our results significantly advance the understanding of the achievability of quantum advantage in noisy distributed sensing. They also offer practical guidance for real-world implementation of quantum sensor networks.
\end{abstract}

\maketitle

\section{Introduction}
Distributed quantum sensing (DQS)~\cite{proctor2017networked,proctor2018multiparameter,rubio2020quantum} is one of the most important applications of quantum networks~\cite{kimble2008quantum, wehner2018quantum}. It is expected to surpass classical sensing techniques in areas ranging from magnetometry~\cite{steinert2010high,pham2011magnetic,hall2012high,rondin2014magnetometry}, phase imaging~\cite{humphreys2013quantum}, precision clocks~\cite{komar2014quantum}, energy applications~\cite{crawford2021quantum}, all the way to the exploration of fundamental physics~\cite{barontini2022measuring,ye2024essay}, including search for dark matter~\cite{brady2022entangled} and measuring stability of fundamental constants~\cite{roberts2020search}.
Over the past decade, DQS has attracted great theoretical interests~\cite{yue2014quantum,zhang2014quantum,berry2015quantum,gessner2018sensitivity,albarelli2019evaluating,roy2019fundamental,albarelli2022probe,yang2024quantum}. Various DQS protocols have been proposed under ideal conditions~\cite{pezze2017optimal,ge2018distributed,qian2019heisenberg,qian2021optimal,bringewatt2021protocols,gorecki2022multiple,ehrenberg2023minimum,hu2025optimal}. On the other hand, along with the experimental demonstration of concept in small-scale matter systems~\cite{malia2022distributed}, the analysis of realistic DQS has also emerged recently, for instance, DQS with noise in signals~\cite{sekatski2020optimal,wolk2020noisy,hamann2022approximate,hamann2024optimal}, atomic ensemble partition for DQS~\cite{fadel2023multiparameter}, and DQS with probabilistic quantum network operations~\cite{van2024utilizing}.
However, the state-of-the-art progress in quantum networking~\cite{knaut2024entanglement,liu2024creation,stolk2024metropolitan} demonstrates that remote entanglement has much higher infidelity than local operations, so the process of initial state preparation will be the major error source for DQS in the near term. Despite its critical importance and necessity, the systematic analysis of DQS with imperfect state preparation surprisingly remains largely unexplored. In this work, we focus on studying the possibility of realizing quantum advantage in DQS with noisy quantum networks which can only distribute noisy entangled states. 

\begin{figure}[t]
    \centering
    \includegraphics[width=0.5\linewidth]{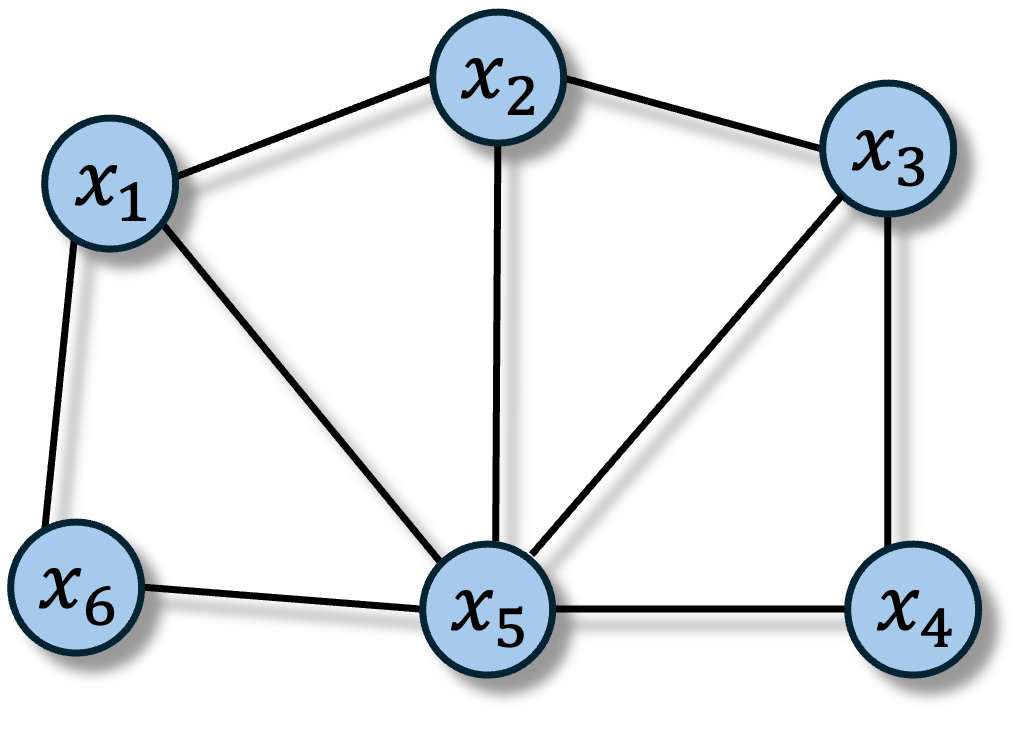}
    \caption{A general quantum sensor network where a parameter $x_i$ is encoded locally for sensor node $i$. The solid edges between sensor nodes represent physical connections such as optical fibers.}
    \label{fig:sensor_network}
\end{figure}

We first specify the initial probe state, the parameter encoding process, and the estimation objective as follows. In a general DQS task, the objective is to estimate global function(s) of local parameters, where a generally different parameter $x_i$ is encoded on sensor node $i$, as shown in Fig.~\ref{fig:sensor_network}. In this work, we consider the canonical simple example of estimating the average of local parameters~\cite{gross2020one} that are distributed at distant network nodes. This task can be decomposed into three steps: (i) preparation of the initial probe state, (ii) encoding the parameters on the probe state, and (iii) performing measurement on the encoded state to extract information about the parameters. We assume that when there is no decoherence the parameter encoding has the standard phase accumulation form: $U(x) = \exp\left[-i\left(\sum_{i=0}^{d-1}x_iH_i\right)\right]$, where $x=(x_1,\dots,x_d)^T\in\mathbb{R}^d$ is an array of $d$ parameters located at $d$ sensor nodes in the network, and $H_i=\frac{1}{2}\sum_{k=0}^{n-1}\sigma^{(i,k)}_z$ is the z-component of collective spin for the $n$ local qubit sensors on node $i$. 

Under the aforementioned unitary encoding channel, the optimal probe state for this problem has been shown to be the global GHZ state ~\cite{proctor2018multiparameter,proctor2017networked}, i.e., a GHZ state $|\mathrm{GHZ}\rangle=(|0\dots 0\rangle+|1\dots 1\rangle)/\sqrt{2}$ involving all the sensors on each sensor node. Therefore, we choose this global GHZ state as the target initial probe state to prepare in the quantum network of sensors. The preparation of the probe state over a quantum network can be decomposed into two steps. Firstly, the quantum network will distribute a $d$-qubit GHZ state across all $d$ sensor nodes, with each node having one qubit. Each node can then perform local entanglement generation~\cite{monz201114,song2019generation,omran2019generation,ho2019ultrafast,pogorelov2021compact,mooney2021generation,zhao2021creation,comparin2022multipartite,zhang2024fast,cao2024multi,yin2024fast} to extend the size of the global GHZ state: By entangling $(n-1)$ additional quantum sensors per node, the global GHZ state will have $nd$ qubits in total. Note that only $(d-1)$ Bell pairs are needed to distribute the $d$-qubit GHZ state over the $d$ nodes, see App.~\ref{app:architecture}. Therefore, such global GHZ states are arguably the easiest to generate by future quantum networks which primarily distribute Bell pairs~\cite{kimble2008quantum, wehner2018quantum}, which further justifies why we choose them as the probe states.

Given noisy quantum networks, the prepared initial probe state has to be a mixed state. In this work, we assume that the initial probe state is GHZ-diagonal: $\rho_0=\sum_a\lambda_a|\psi_{a0}\rangle\langle\psi_{a0}|$ with all $\lambda_a$ being non-zero, where $a$ is the index of the pure GHZ basis state $|\psi_{a0}\rangle$. The GHZ basis states we consider are superpositions of computational basis states, e.g. $(|0\dots00\rangle\pm|1\dots,11\rangle)/\sqrt{2},~|0\dots01\rangle\pm|1\dots,10\rangle)/\sqrt{2}$, etc. Such states can be understood as the the result of a pure GHZ state undergoing Pauli channels. The model of Pauli errors~\cite{shettell2020graph} can be justified as follows: In principle one can perform Pauli twirling~\cite{dur2005standard,emerson2007symmetrized,dankert2009exact} to transform error channels into Pauli channels, or stabilizer twirling~\cite{toth2007efficient} to cancel off-diagonal density matrix elements under a stabilizer basis. 
Many real-world noise processes during state preparation already tend to destroy the coherence between different GHZ basis states, and twirling just makes this aspect more explicit and systematic.
This assumption of GHZ-diagonal states also facilitates our analytical studies, for which we analytically derive the quantum Fisher information (QFI)~\cite{helstrom1969quantum,czekaj2015quantum} which characterizes the lower bound of parameter estimation variance, the quantum Cram\'er-Rao bound (QCRB). Further details about the justification of why we consider GHZ-diagonal initial probe states can be found in App.~\ref{app:twirling}.

Next we derive the condition for quantum advantage over the optimal local sensing protocol from the QFI and demonstrate how the condition serves as an entanglement detection criterion for multipartite states in Sec.~\ref{sec:q_adv_cond}.
We demonstrate quantum advantage in distributed sensing when ignoring decoherence during the encoding dynamics in Sec.~\ref{sec:noiseless_encoding}, where we consider scenarios of increasing practicality, namely with noiseless local entanglement generation, imperfect local entanglement generation when scaling up the number of local sensors, and local measurement constraint, respectively. We then explicitly include the influence of individual sensors' dephasing during parameter encoding dynamics in Sec.~\ref{sec:dephasing_dynamics}. Finally, we summarize the results and provide further discussion in Sec.~\ref{sec:conclusion}.

\section{Condition for quantum advantage}\label{sec:q_adv_cond}
The condition for quantum advantage is formulated in terms of the QFI for the estimation of the average of all local parameters. With the QFI, we can further obtain an optimization-free criterion for detecting entanglement in multipartite states, that has a clear operational interpretation.

\subsection{Derivation of the QFI}
To derive the QFI for the average of all local parameters, we start with the QFI matrix (QFIM)~\cite{paris2009quantum,petz2011introduction,toth2014quantum,liu2020quantum} for the $d$ local parameters. According to the assumed parametrization unitary, for a general GHZ-diagonal $nd$-qubit initial probe state $\rho_0$ where each of the $d$ sensor nodes holds $n$ qubit sensors, the QFIM for parameters $x$ is $\mathcal{F}(x)=(1-C)n^2J_d$, where $J_d$ is the $d$-by-$d$ matrix of ones, and $C$ captures the quality of the initial probe state, and it is analytically calculated as
\begin{align}\label{Eq:C_coeff}
    C = \sum_{(a,b)\in\mathcal{S}}\frac{4\lambda_{a0}\lambda_{b0}}{\lambda_{a0}+\lambda_{b0}},
\end{align}
where $\lambda_{a0}$ is the eigenvalue corresponding to GHZ basis state $|\psi_{a0}\rangle$, and $\mathcal{S}$ denotes the set of index pairs $(a,b)$ \textit{without double counting}, such that $|\psi_{a0}\rangle$ and $|\psi_{b0}\rangle$ are GHZ states expressed as superpositions of the same pair of computational basis states but with opposite relative phase. For instance, for 3-qubit GHZ states $(|000\rangle\pm|111\rangle)/\sqrt{2}$ is such a pair. Note that $C$ is not necessarily a constant and in general depends on $n$, and its dependence is determined by error models. 

Our parameter to estimate is $\theta_1 = v_1^Tx$ where $v_1\propto(1,\dots,1)^T$. For concreteness, we follow the convention in~\cite{proctor2018multiparameter,proctor2017networked} and choose $v_1=(1,\dots,1)^T/\sqrt{d}$ which is normalized under the 2-norm. Then we are able to transform the QFIM for the new global parameter $\theta_1$, where the transformation could be operationally realized by constructing additional $(d-1)$ normalized vectors $v_2,\dots,v_d$, s.t. $v_i^Tv_j=\delta_{ij}$~\cite{paris2009quantum}. Consequently, we have
\begin{align}
    \mathcal{F}(\theta) = d(1-C)n^2
    \begin{pmatrix}
        1 & 0 & \dots & 0 \\
        0 & 0 & \dots & 0 \\
        \vdots & \vdots & \ddots & \vdots \\
        0 & 0 & \dots & 0 \\
    \end{pmatrix},
\end{align}
where $\theta=(v_1^Tx,\dots, v_d^Tx)$, and the only non-zero entry is the QFI for $\theta_1$ of our interest: $\mathcal{F}(\theta_1) = d(1-C)n^2$, which incorporates effects from both the network ($d$) and local resources ($n^2$). The details can be found in App.~\ref{app:qfi_ghz}.
Given the analytical formula of calculating the $C$ term for the QFI, we have a lower bound of QFI for all possible GHZ-diagonal states with a fixed fidelity, whose proof is also in App.~\ref{app:qfi_ghz}.
\begin{proposition}\label{prop:worst_case_qfi}
    For $d$-qubit GHZ-diagonal states with a fixed fidelity $F\in(0,1)$, the lowest QFI for estimating the average of $d$ local parameters is $d(2F-1)^2$.
\end{proposition}
The worst-case scenario corresponds to that all errors in the initial probe state are indistinguishable from the signal. On the other hand, the QFI for noisy GHZ state can be equal to that for noiseless GHZ state, which can be interpreted as that when probe state errors can be distinguished from the signal, it is still possible to extract the information encoded in the ``noisy'' components of the probe state.

\subsection{Comparison with the optimal local strategy}
To demonstrate quantum advantage, the comparison baseline should be the optimal local sensing strategy where each sensor node can estimate the local parameter to the best accuracy possible under the resource constraints, and we then use the estimated values to approximate their average. 

It is well known that the GHZ state is the optimal probe state for local phase estimation~\cite{bollinger1996optimal}, and indeed using separate GHZ states on each sensor node is the best local strategy (without any quantum communication between sensor nodes) for estimating the average of local parameters~\cite{proctor2018multiparameter,proctor2017networked}. The variance of this indirect estimation can be calculated through propagation of error~\cite{toth2014quantum}: $\mathrm{Var}_\mathrm{local}(\hat{\theta}_1) = \sum_{l=1}^d\left(\frac{\partial\theta_1}{\partial x_l}\right)^2\mathrm{Var}_\mathrm{local}(\hat{x}_l)$, where we use the QCRB of local GHZ state with $n$ qubits so that $\mathrm{Var}_\mathrm{local}(\hat{x}_l)=1/(n^2\mu)$, where $\mu$ is the number of samples. Noticing that $\partial\theta_1/\partial x_l=1/\sqrt{d}$, we immediately have $\mathrm{Var}_\mathrm{local}(\hat{\theta}_1) = 1/(n^2\mu)$. Therefore, on top of the scaling advantage from increasing local resources, the additional \textit{relative} quantum advantage of entangling sensor nodes with quantum networks is achievable when
\begin{align}\label{Eq:adv_cond}
    \eta \equiv d(1-C)>1.
\end{align}
Since entanglement in the initial probe state across sensor nodes is necessary for DQS quantum advantage, the above condition can also be interpreted as an operationally meaningful, optimization-free entanglement detection criterion for arbitrary $d$-qubit states. See details in App.~\ref{app:qfi_ghz}.
Next, we will evaluate $\eta$ for $nd$-qubit GHZ state created from the initial noisy $d$-qubit GHZ state with noiseless and noisy local entanglement generation, respectively.

\subsection{Entanglement detection criterion}
For entanglement detection purpose, we can show that any fully separable initial probe state cannot result in better DQS performance than the optimal local sensing strategy as follows.
\begin{proposition}\label{prop:ent_detect}
    Suppose a $d$-qubit state is subject to a unitary parameter encoding of $d$ parameters $x=(x_1,\dots,x_d)^T$, $U(x) = \exp\left[-i\left(\sum_{i=1}^d x_i\frac{\sigma^{(i)}_z}{2}\right)\right]$. Then the estimation variance of $\theta_1=v_1^Tx$ for $v_1=(1,\dots,1)^T/\sqrt{d}$ using fully separable $d$-qubit state is always no smaller than $1/\mu$, where $\mu$ is the amount of state copies.
\end{proposition}
The proof can be found in App.~\ref{app:qfi_ghz}. The above Proposition means that entanglement in the initial probe state across sensors is necessary to achieve quantum advantage in the DQS task of estimating $\theta_1$ over the optimal local sensing strategy. In other words, if we find that a certain input state can demonstrate DQS quantum advantage, there must be some entanglement in it with respect to the partition of $d$ sensors. Although the calculation of $C$ in the quantum advantage condition $d(1-C)>1$ is based on the GHZ-diagonal assumption, arbitrary state can be converted into GHZ-diagonal from through GHZ stabilizer twirling which is local and thus will not generate any additional entanglement. From this perspective, the quantum advantage condition $d(1-C)>1$ is a sufficient condition for entanglement detection~\cite{horodecki2009quantum,guhne2009entanglement} in $d$-partite state where each party has local dimension $2^n$, with a clear operational interpretation. Importantly, the only calculation needed is the value of $C$, and the calculation is \textit{optimization-free}. Moreover, even though there are admittedly exponentially many diagonal elements in the multipartite density matrix, in many cases we do not have to go through all of them, as we can stop the summation when we already have $d\sum_{(a,b)\in\mathcal{S}'}(\lambda_{a0}-\lambda_{b0})^2/(\lambda_{a0}+\lambda_{b0}) > 1$ for $\mathcal{S}'\subset\mathcal{S}$.

Then we demonstrate one example of $d$-qubit state which is not fully separable and can be detected by the DQS quantum advantage condition. The example state is the equal mixture of GHZ states which can be expressed as superposition of two computational basis states with zero relative phase. For instance, for 3 qubits such GHZ states include $(|000\rangle+|111\rangle)/\sqrt{2}, (|001\rangle+|110\rangle)/\sqrt{2}, (|010\rangle+|101\rangle)/\sqrt{2}, (|011\rangle+|100\rangle)/\sqrt{2}$. For $d$ qubits there are $2^{d-1}$ such GHZ states $|\mathrm{GHZ}_d^{(i)}\rangle$. Then the example state can be written as $\rho_d = 2^{1-d}\sum_{i=0}^{2^{d-1}-1}|\mathrm{GHZ}_d^{(i)}\rangle\langle\mathrm{GHZ}_d^{(i)}|$. We have $1 - C(\rho_d) = 1$, which means that $d[1 - C(\rho_d)]=d>1$, and thus $\rho_d$ must be entangled. 

The DQS quantum advantage condition easily determines the entanglement in $\rho$, but for other entanglement criteria it could be much less straightforward, if possible, to successfully detect the entanglement. For instance, according to Laskowski-{\.Z}ukowski (LZ) criterion~\cite{laskowski2005detection} for $k$-separability which is equivalent to Mermin-type inequalities~\cite{mermin1990extreme,nagata2002configuration,seevinck2008partial}, a $d$-qubit state that is $d$-separable (fully separable) must satisfy $\max_j\left\vert\rho_{j,\overline{j}}\right\vert \leq 2^{-d}$, where $\rho_{j,\overline{j}}$ is the off-diagonal element for the density matrix with $j=1,\dots,2^d$ being the row index and $\overline{j}=2^d-j+1$. This criterion is basis independent, so if a state under a certain basis violates the necessary condition for full separability, it is entangled. However, to obtain the maximum of density matrix off-diagonal element over all possible bases requires additional optimization. It is obvious that the example state $\rho_d$ does not violate the $d$-separability necessary condition in two common bases, namely the computational basis and the GHZ basis. In addition, after D{\"u}r-Cirac depolarization~\cite{dur2000classification} the example state $\rho_d$ becomes separable under arbitrary bipartition, so the D{\"u}r-Cirac separability criterion cannot be used to detect entanglement in $\rho_d$. Moreover, it can be easily verified for $\rho_3$ that it is positive after partial transpose (PPT) under any bipartition, which means that the DQS quantum advantage condition can detect entanglement when the PPT criterion~\cite{peres1996separability,horodecki2001separability} fails.

\section{Quantum advantage under noiseless encoding dynamics}\label{sec:noiseless_encoding}
The time scale for quantum network operations is in general significantly longer than local quantum operations, so the local parameter encoding process can be much less noisy than the initial probe state preparation over the quantum network. By assuming a noiseless encoding process, it is also easier to study the impact of network imperfections and distinguish it from the effect of the encoding process on the DQS performance. Therefore, we first focus on the noiseless unitary encoding.

\subsection{Noiseless local entanglement generation}
We first consider noiseless local entanglement generation, which can be interpreted as applying perfect CNOT gates between the qubit sensor initially entangled with other sensor nodes and other local qubit sensors prepared in $|0\rangle$ state to entangle, with the former being the control. This assumption results in that the $nd$-qubit GHZ state has the same fidelity as the $d$-qubit GHZ state. We are then allowed to decouple the network imperfections and local errors, and thus to understand the limit on the amount of network imperfections beyond which there cannot be any quantum advantage for DQS. 

While the analysis of $\eta$ is applicable to arbitrary GHZ-diagonal states, here for concreteness we assume a canonical full-rank error model for the initial $d$-qubit GHZ state, a depolarized GHZ state as a mixture of pure GHZ state and maximally mixed state, $\rho_\mathrm{dp}(F) = \frac{2^dF-1}{2^d-1}|\mathrm{GHZ}_d\rangle\langle\mathrm{GHZ}_d| + \frac{1-F}{2^d-1}I$, where $F$ is the fidelity to the pure GHZ state. Then we have the closed form expression of the constant $C$ in Eqn.~\ref{Eq:C_coeff}
\begin{align}
    C_\mathrm{dp} = \frac{(1-F)(4^dF+2^d-2)}{[(2^d-2)F+1](2^d-1)}.
\end{align}
One can verify that for fixed $d$ the maximum value of $C_\mathrm{dp}$ is taken at $F=2^{-d}$, corresponding to a maximally mixed state. The maximum value is $\left.C_\mathrm{dp}\right\vert_{F=2^{-d}} = 1$, because the maximally mixed state will not be able to carry any information of the parameter to estimate.

We can then substitute the $C$ for $d$-qubit depolarized GHZ state in Eqn.~\ref{Eq:adv_cond} and let $\eta_\mathrm{dp}=1$ to derive a closed-form fidelity threshold, as a requirement on the quantum network performance for quantum advantage in DQS
\begin{align}\label{Eq:fid_th_noiseless_opt}
    F_\mathrm{th,dp} = 2^{-d} + \frac{(2^d-1)\left(2^d-2+\sqrt{(2^d-2)^2+2^{d+3}d}\right)}{2^{2d+1}d}.
\end{align}
We see that in the asymptotic regime of large $d$, the above fidelity threshold reduces to $1/d$. 

As we have already argued the quantum advantage in DQS has a close relationship with quantum entanglement. The entanglement properties of the depolarized GHZ state have been well studied. $d$-qubit depolarized GHZ states are not completely separable if $F>3/(2^d+2)$~\cite{dur2000classification}, and indeed we have $F_\mathrm{th,dp}>3/(2^d+2)$ for $d\geq 2$ as expected from Prop.~\ref{prop:ent_detect}. Moreover, depolarized GHZ states are genuine multipartite entangled (GME) if $F>1/2$~\cite{guhne2010separability}. However, except for the special case of $\left.F_\mathrm{th,dp}\right\vert_{d=3}\approx 0.50963>1/2$, we have $F_\mathrm{th,dp}<1/2$ for $d>3$, which suggests that GME is generally unnecessary to demonstrate quantum advantage in DQS. 

We can also consider the worst-case scenario where the quantum network distributes a $d$-qubit rank-2 dephased GHZ state $F|\mathrm{GHZ}_d\rangle\langle\mathrm{GHZ}_d| + (1-F)\sigma_z|\mathrm{GHZ}_d\rangle\langle\mathrm{GHZ}_d|\sigma_z$. The fidelity threshold for this type of noisy GHZ state to be advantageous over the optimal sensing strategy is
\begin{align}
    F_{\mathrm{th},r=2} = \frac{1+\sqrt{d}}{2\sqrt{d}} > \frac{1}{2}.
\end{align}
From GME criterion in~\cite{guhne2010separability}, any rank-2 GHZ-diagonal state is GME, but it is not metrologically useful for our task in general. This suggests that GME is not sufficient for demonstrating quantum advantage in DQS either. Consequently, we have rigorously demonstrated that GME is neither sufficient nor necessary for quantum advantage in DQS with concrete examples.

\subsection{Imperfect local entanglement generation}

\begin{figure}[t]
    \centering
    \includegraphics[width=\linewidth]{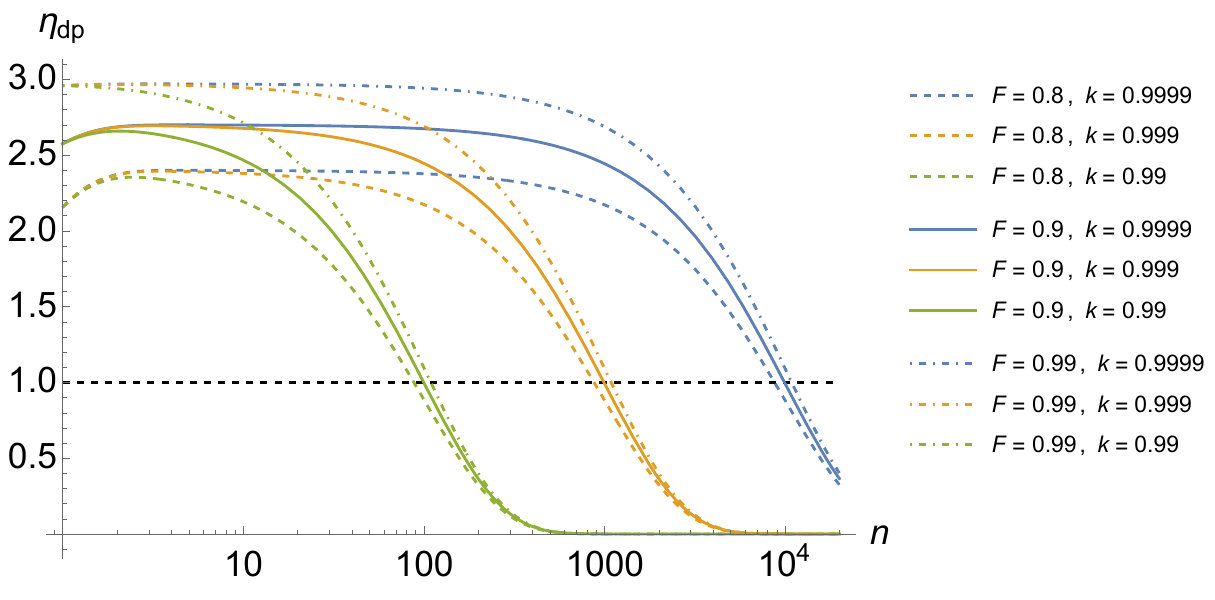}
    \caption{Visualization of $\eta_\mathrm{dp}$ as function of the number of local qubits $n$, for different initial fidelity $F$ and local entanglement generation quality $k$, with fixed $d=3$.}
    \label{fig:n_threshold}
\end{figure}

We now start to include imperfections in local entanglement generation. Specifically, we consider the following phenomenological model to reflect the fidelity decrease when adding more qubits to the GHZ state: The state after local entanglement generation is an $nd$-qubit depolarized GHZ state, while the fidelity is modified according to the number of local qubits as $F(n) = k^{n-1}F$, where $F$ is the fidelity of the initial $d$-qubit GHZ state, and $k\in(0,1)$ is a parameter which describes the quality of local entanglement generation and thus the larger the better. The intuition behind this phenomenological model is the usage of noisy CNOT gates to generate GHZ states, and one can think $k$ increases as gate fidelity increases. The justification of this phenomenological model when the error bias is not too large can be found in App.~\ref{app:noisy_eg}. We can again evaluate $C$ in Eqn.~\ref{Eq:C_coeff} as a function of the four system parameters
\begin{align}
    C_\mathrm{dp}(F,d,n,k) = \frac{(1-k^{n-1}F)(4^{nd}k^{n-1}F+2^{nd}-2)}{[(2^{nd}-2)k^{n-1}F+1](2^{nd}-1)}.
\end{align}

It can be shown that, as intuitively expected, $C_\mathrm{dp}$ decreases monotonically as initial fidelity $F$, number of local parameters $d$, and local entanglement generation quality $k$, increase for all $n\geq 1$, if $F,k>2^{-d}$. On the other hand, the dependence of $C_\mathrm{dp}$ on number of local quantum sensors $n$ is more complicated. We visualize $\eta_\mathrm{dp}=d(1-C_\mathrm{dp})$ with varying $n$ under different parameter choices in Fig.~\ref{fig:n_threshold}. From the plot we can verify that the relative advantage $\eta_\mathrm{dp}$ is larger when $F,k$ increase, and $\eta_\mathrm{dp}$ vanishes for large $n$ as the extra errors introduced by having more local sensors through noisy local entanglement generation overwhelm the gain of information. We can also see non-monotonic behavior of $\eta_\mathrm{dp}$ when $k$ is sufficiently high in comparison to $F$, which indicates an increase in relative advantage for small local sensor numbers. 

The point at which $\eta_\mathrm{dp}=1$ determines the maximum number of local quantum sensors $n_\mathrm{max}$ for quantum advantage over the optimal local sensing strategy to be potentially demonstrated. We observe that $n_\mathrm{max}$ does not change much when varying $F$ given fixed $k$, but it changes significantly when fixing $F$ and varying $k$, which suggests that the initial fidelity $F$ has less impact on $n_\mathrm{max}$ than the local entanglement generation quality $k$. This can be understood by considering the fidelity threshold for the $nd$-qubit depolarized GHZ state to be advantageous over the optimal local strategy, which also scales as $1/d$ and is almost independent of $n$. Imperfect local entanglement generation will decrease the GHZ fidelity to the threshold at $n_\mathrm{max}$. Therefore, we expect $k^{n_\mathrm{max}-1}F\approx 1/d$, which gives $n_\mathrm{max}\approx-\ln(dF)/\ln(k)$. We can explicitly evaluate the sensitivities of $n_\mathrm{max}$ to changes in $F$ and $k$ as the partial derivatives with respect to $F$ and $k$: $S_F = \partial n_\mathrm{max}/\partial F \approx - 1/[F\ln(k)]$ and $S_k = \partial n_\mathrm{max}/\partial k \approx \ln(dF)/[k\ln^2(k)]$.
We can then compare the sensitivities $S_F$ and $S_k$ by taking their ratio, from which one can show that in general $S_k/S_F\gg 1$ for high $k$ which is needed for meaningful local entanglement generation. These results imply robustness of the quantum advantage in DQS against the imperfection of probe state generation over realistic quantum networks. See more details in App.~\ref{app:noisy_eg}.

\subsection{Local measurement constraint}
It is possible that the optimal measurement to saturate the QCRB is entangling between the sensor nodes. In the DQS setup, entangling measurement needs additional remote entanglement as a resource to implement. Therefore, we want to use only local operations and classical communication (LOCC)~\cite{chitambar2014everything} to extract information in DQS problems. However, in general LOCC is not guaranteed to saturate the QCRB in DQS~\cite{zhou2020saturating}. Here we explicitly consider the local measurement constraint. 

It is known that $M=\sigma_x^{\otimes nd}$ is the optimal observable for $nd$-qubit pure GHZ probe state under unitary z-direction phase accumulation~\cite{bollinger1996optimal}, and in principle each sensor node only needs to perform measurement of the local observable $\sigma_x^{\otimes n}$. However, if we use this measurement for noisy GHZ states, for instance depolarized GHZ state, the variance of parameter estimation diverges in the limit of small local parameters. This conclusion is also implied by other numerical results~\cite{cao2023distributed}. In App.~\ref{app:loc_meas} we analytically demonstrate this property for the depolarized GHZ state, while it holds generally for any GHZ-diagonal state with non-unit fidelity.

Motivated by our problem formulation that the parameter to estimate is encoded through z-direction phase accumulation, we further explore the optimization over a subset of local measurements, i.e. the tensor product of single-qubit measurements along a direction on the equator of the Bloch sphere: $M(\alpha)=\left[O(\alpha)\right]^{\otimes nd}$ where $O(\alpha) = |\psi^+(\alpha)\rangle\langle\psi^+(\alpha)| - |\psi^-(\alpha)\rangle\langle\psi^-(\alpha)|$ with $|\psi^\pm(\alpha)\rangle = (|0\rangle\pm e^{i\alpha}|1\rangle)$ (thus $O(\alpha) = e^{i\alpha}|1\rangle\langle 0| + e^{-i\alpha}|0\rangle\langle 1|$), characterized by the azimuthal angle $\alpha$. For an $nd$-qubit depolarized GHZ state, the optimal azimuthal angle is $\alpha_\mathrm{opt} = \frac{2l+1}{2nd}\pi,\ l\in\mathbb{Z}$. Note that the optimal azimuthal angle depends on number of local quantum sensors $n$ and sensor node number $d$. Moreover, the estimation variance diverges quickly when $\alpha$ deviates from $\alpha_\mathrm{opt}$. This implies that the accuracy of local operation is extremely important, and potentially more important than the quality of the entangled states distributed by quantum networks. More details on the azimuthal measurement optimization can be found in App.~\ref{app:loc_meas}.

\begin{figure}[t]
    \centering
    \includegraphics[width=0.8\linewidth]{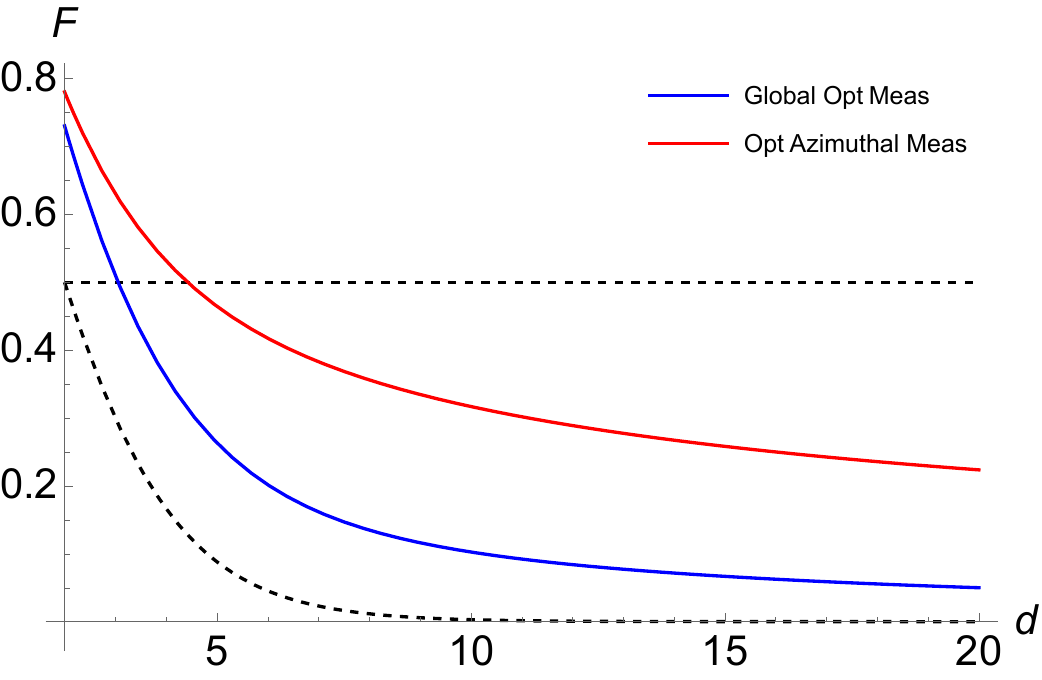}
    \caption{Visualization of fidelity thresholds. The blue curve denotes the threshold given by the QFI, and the red curve denotes the threshold given by optimized local azimuthal measurement. The horizontal and curved black dashed lines represent the thresholds for $d$-qubit depolarized GHZ state to be GME and non-fully separable, respectively.}
    \label{fig:fid_th_comparison}
\end{figure}

The fidelity threshold for an $nd$-qubit depolarized GHZ state to be advantageous over the optimal local strategy when using the optimized azimuthal measurement is
\begin{align}\label{Eq:fid_th_azimuthal}
    F_{\mathrm{th},M(\alpha_\mathrm{opt})}(n) = \frac{2^{nd} + \sqrt{d} - 1}{2^{nd}\sqrt{d}}.
\end{align}
The case of $n=1$ for Eqn.~\ref{Eq:fid_th_azimuthal} corresponds to noiseless local entanglement generation, similar to Eqn.~\ref{Eq:fid_th_noiseless_opt}. We thus visualize and compare both fidelity thresholds in Fig.~\ref{fig:fid_th_comparison}. The fidelity threshold in Eqn.~\ref{Eq:fid_th_azimuthal} is always higher than that given by Eqn.~\ref{Eq:fid_th_noiseless_opt}, suggesting that the optimized local azimuthal measurement does not saturate the QCRB. Note that from symmetry and the assumed signal orientation, the optimized azimuthal measurement should have been the best local measurement when $n=1$. Moreover, although for small problem sizes $d<5$ the initial GHZ state needs to be GME to demonstrate quantum advantage if we use the optimized azimuthal measurement, when the problem size grows the requirement of initial fidelity drops, and in general the initial $d$-qubit GHZ state does not have to be GME. 

On the other hand, we can also consider the local azimuthal measurement for $n=1$, rank-2 dephased GHZ state $\rho=F|\mathrm{GHZ}_d\rangle\langle\mathrm{GHZ}_d| + (1-F)Z|\mathrm{GHZ}_d\rangle\langle\mathrm{GHZ}_d|Z$. Through standard propagation of error analysis, we can show that the minimal variance is $\left.\mathrm{Var}_{M(\alpha_\mathrm{opt}),r=2}(\hat{\theta}_1)\right\vert_{\theta_1=0} = 1/[d(2F-1)^2]$, which is identical to the worst case value in Prop.~\ref{prop:worst_case_qfi}. This means that the optimal azimuthal measurement is able to saturate the QCRB for $n=1$ rank-2 dephased GHZ state, thus an optimal measurement, and also implies an error model-dependent locality of the optimal measurement for DQS.

\section{Effect of dephasing during parameter encoding}\label{sec:dephasing_dynamics}
Previously we have focused on noiseless encoding dynamics. In practice, the decoherence effect during the parameter encoding dynamics can also be very important~\cite{huelga1997improvement,chaves2013noisy,zheng2022preparation,saleem2023optimal}. Therefore, here we explicitly include decoherence during the parameter encoding process. We specifically focus on the effect of individual qubit dephasing which corresponds to Pauli Z error. The dephasing error is an important example of Markovian noise. Moreover, in our problem setup where the Hamiltonian generates the signal in z-direction, the improvement from using advanced error suppression techniques such as quantum error correction is greatly limited~\cite{dur2014improved,kessler2014quantum,arrad2014increasing,sekatski2017quantum,demkowicz2017adaptive,zhou2018achieving}.

\subsection{Effective initial state}
Now the independent parameter encoding dynamics for the $k$-th qubit on the $i$-th node is described by the single-qubit Lindblad master equation
\begin{align}\label{eqn:encode_dephase}
    \frac{d}{dt}\rho = -i\frac{\omega^{(i)}}{2}\left[\sigma_z^{(i,k)},\rho\right] + \frac{\gamma}{2}\left(\sigma_z^{(i,k)}\rho \sigma_z^{(i,k)} - \rho\right),
\end{align}
where $\omega^{(i)}$ is the phase accumulation speed of each qubit sensor on node $i$, and $\gamma$ is the single-qubit dephasing rate. The quantum channel described by the above master equation can be decomposed into commuting parts $\mathcal{E}^{(i,k)}(t) = \mathcal{E}_\mathrm{enc}^{(i,k)}(t)\circ\mathcal{E}_\mathrm{dp}^{(i,k)}(t) = \mathcal{E}_\mathrm{dp}^{(i,k)}(t)\circ\mathcal{E}_\mathrm{enc}^{(i,k)}(t)$, where $\mathcal{E}_\mathrm{enc}^{(i,k)}(t)[\rho] = e^{-i\sigma_z^{(i,k)}\omega^{(i)} t/2}\rho e^{i\sigma_z^{(i,k)}\omega^{(i)} t/2}$ is the unitary encoding channel, where $x_i=\omega^{(i)} t$ is the accumulated phase for each qubit on node $i$. Also, $\mathcal{E}_\mathrm{dp}^{(i,k)}(t)[\rho] = (1 + e^{-\gamma t})\rho / 2 + (1 - e^{-\gamma t})\sigma_z^{(i,k)} \rho \sigma_z^{(i,k)} / 2$ is the single-qubit pure dephasing/phase flip channel. Then according to the commutation of the pure Hamiltonian parameter encoding channel and the pure dephasing channel on each qubit, the result of noisy quantum sensing dynamics of time $t$ is equal to the result of a new noisy initial state undergoing the noiseless encoding. The effective new initial state is
\begin{align}
    \rho_0' = \left[\bigotimes_{i,k}\mathcal{E}_\mathrm{dp}^{(i,k)}(t)\right][\rho_0],
\end{align}
which is the original initial state undergoing extra dephasing. From the structures of the stabilizer of the GHZ states and the extra Pauli error channel, we can derive the explicit form of the effective initial state as shown in App.~\ref{app:dephasing}. Then we have the QFI for $\theta_1$ from encoding under dephasing
\begin{align}\label{eqn:qfi_dephasing_phase}
    \mathcal{F}_\mathrm{net}(\theta_1) = d\frac{\left[F(2^dk)^n - k\right]^2 (2q-1)^{2nd}}{k(2^{nd} - 1)\left[Fk^n(2^{nd}-2) + k\right]} n^2,
\end{align}
where $q$ is the probability for a qubit to not undergo phase flip during the encoding dynamics, which does not depend on evolution time $t$ explicitly at this moment. The subscript emphasizes that the QFI emphasizes the use of quantum networks.

\subsection{Estimation of the average frequency}
In the above we consider the estimation of the average accumulated phase over an evolution of $t$ duration $\theta_1 = (\sum_{i=1}^d\omega^{(i)}t)/\sqrt{d}$. In practice, we may instead be interested in the angular frequency, thus now we consider the estimation of another global parameter $\tilde{\theta}_1 = (\sum_{i=1}^d\omega^{(i)})/\sqrt{d} = \theta_1 / t$. Treating evolution time $t$ as a fixed constant for a specific sensing cycle, the estimation variance of $\tilde{\theta}_1$ satisfies $\mathrm{Var}(\tilde{\theta}_1) = \mathrm{Var}(\theta_1)/t^2$. Thus we have the QFI for $\tilde{\theta}_1$ with sensing cycle of $t$ duration $\mathcal{F}_\mathrm{net}^{(\tilde{\theta}_1)}(t) = t^2\mathcal{F}_\mathrm{net}^{(\theta_1)}(t)$, where the superscript denotes the parameter with respect to which the QFI is defined. 

When we consider the best case scenario of local parameter estimation as the comparison baseline, we use $n$-qubit \textit{noiseless GHZ states} at each of the $d$ sensor nodes to estimate the local parameters through \textit{noiseless dynamics}, before combining the local sensing results to infer the global parameter. We know that the QFI dynamics for noiseless optimal local estimation strategy is $\mathcal{F}_\mathrm{local}(t) = n^2t^2$. Therefore, when we fix the evolution time $t$, the global estimation strategy with noises can potentially demonstrate quantum advantage when $\mathcal{F}_\mathrm{net}(t) > \mathcal{F}_\mathrm{local}(t)$. As shown in App.~\ref{app:dephasing} the above condition leads to the requirement on the evolution time 
\begin{align}\label{eqn:adv_cond_fix_t}
    t \lesssim \ln(Fk^{n-1}d)/2nd\gamma.
\end{align}
From the approximation we can evaluate the system performance by taking its partial derivatives with respect to $n,d,f,k$. As shown in App.~\ref{app:dephasing}, the partial derivative with respect to $k$ is generally the largest, as its inverse scales linearly in terms of $n$ and $d$, while other partial derivatives' inverses scale at least quadratically. This again suggests the relative robustness of distributed quantum sensing against network imperfections which only affect the initial state fidelity $F$, resonating with our previous results about noiseless encoding.

In practice, it is common to repeat the sensing cycle sequentially within a given total amount of time $T$ as resource, which is known as the sequential scheme of quantum sensing~\cite{huelga1997improvement}. Under such a setup, our objective then becomes to maximize the gained information within a fixed amount of time to minimize the parameter estimation variance. Therefore, we are interested in the quantity $\mathcal{F}(t)/t$, \textit{QFI per unit (evolution) time} which characterizes the average information accumulation speed during an evolution with duration $t$~\cite{saleem2023optimal,zang2025enhancing}. Now the noiseless local sensing baseline becomes $\mathcal{F}_\mathrm{local}(t)/t = n^2t$, which monotonically increases as the time for each single sensing cycle $t$ increases. In the noisy network case, we would like to optimize $t$ so that $\mathcal{F}_\mathrm{net}(t)/t$ is maximized. According to the error model, dephasing and state preparation errors are decoupled. Therefore, the optimization of the duration of each sensing cycle can be conveniently obtained as $t^* = 1 / (2 n d \gamma)$. If we further increase the sensing cycle time for the noiseless local baseline, the QFI per unit time will keep increasing. Therefore, in order for the noisy network sensing to demonstrate advantage, we must consider the limitation on the sensing cycle time, which can be described by the threshold time $t_\mathrm{th}$, s.t. $t\leq t_\mathrm{th}$. When $t_\mathrm{th}\geq t^*$, we have the necessary condition for potential quantum advantage as a constraint on $t$
\begin{align}
    t_\mathrm{th} \lesssim \frac{F k^{n-1}}{2 n e \gamma},
\end{align}
where $e$ is the base of natural logarithm. One can see that quantum advantage is impossible in a two-node, i.e. $d=2$, setup when $t_\mathrm{th}\geq t^*$, because $F k^{n-1} < e$, while the condition is still possible to be fulfilled in larger systems. In addition, we may consider that the sensing cycle time has even stricter limitation. When $t_\mathrm{th}<t^*$, the QFI per unit time for the noisy networked case monotonically increases as $t$ increases. Therefore, for a specific $t_\mathrm{th}<t^*$ we simply compare $\mathcal{F}_\mathrm{local}(t_\mathrm{th})/t_\mathrm{th}$ and $\mathcal{F}(t_\mathrm{th})/t_\mathrm{th}$, which leads to the necessary condition of quantum advantage that is in the same form as Eqn.~\ref{eqn:adv_cond_fix_t}.

\subsection{Local azimuthal measurement}
In the above we focus on QFI conditions for potential quantum advantage. Similar to previous discussion without dephasing, we need to consider the constraint of local measurement in the network setup. Not aiming to optimize over all possible LOCC measurement setup, we still focus on the family of local azimuthal measurement as previously used: $M(\alpha)=\left[O(\alpha)\right]^{\otimes nd}$ where $O(\alpha) = |\psi^+(\alpha)\rangle\langle\psi^+(\alpha)| - |\psi^-(\alpha)\rangle\langle\psi^-(\alpha)|$ with $|\psi^\pm(\alpha)\rangle = (|0\rangle\pm e^{i\alpha}|1\rangle)$. We can perform the same measurement angle optimization as in the noiseless case, and discover that the optimal condition is the same as for the noiseless encoding dynamics case studied previously. This means that the optimal local azimuthal measurement is independent of the sensing cycle time, which avoids further complication in system control. We can show that the optimal local azimuthal measurement only saturates the QCRB when there is no initial state preparation error before the encoding dynamics under dephasing, which corroborates with our previous result, as we have shown that when the initial state only has dephasing error, the optimal local azimuthal measurement can indeed saturate the QCRB. 

We can evaluate the sensitivity of the QCRB saturation by taking partial derivatives of the estimation variance using the optimal azimuthal measurement with respect to $n,d,F,k$. From the detailed expressions shown in App.~\ref{app:dephasing}, we can see that the extent to which the QCRB is saturated is almost not affected by the number of sensor nodes (number of local parameters) in the network, as the partial derivative with respect to $d$ vanishes for increasing $nd$. On the other hand, the relative difference between the QCRB and the estimation variance using the optimal local azimuthal measurement generally increases as the number of local sensors $n$ increases, due to accumulation of depolarizing error. Also, when $F\sim k$ the local entanglement generation quality $k$ has increasing impact on the saturation than the quantum network entanglement distribution fidelity $F$ as $n$ increases. These results again demonstrate that the performance of distributed quantum sensing is less sensitive to the quantum network properties, including number of nodes $d$ and entanglement distribution fidelity $F$. In addition, even though the QCRB is not saturated by the local azimuthal measurement in general, quantum advantage over the noiseless local sensing strategy is still possible, but with relatively stricter conditions.

\section{Conclusion and discussion}\label{sec:conclusion}
In summary, we uncover the impacts of noisy quantum networks on DQS. Our results offer new insights into realistic DQS, reveal the relation between entanglement and DQS quantum advantage, and can serve as practical guidance to real-world implementation of useful DQS. Notably, the insight about DQS with unitary encoding, that the DQS performance is more robust against network imperfections than local imperfections, applies to noisy encoding with dephasing as well. Moreover, our work may lead to numerous future work: The details of real DQS system deserve further evaluation from architectural perspectives; other probe states such as squeezed states~\cite{gessner2020multiparameter,pezze2025advances} may demonstrate interesting tradeoff between the difficulty to distribute over quantum networks and the robustness to noise.

For real-world implementation of DQS, continuous entanglement distribution~\cite{chakraborty2019distributed,kolar2022adaptive,inesta2023performance,ghaderibaneh2022pre,zang2024analytical,zhan2025design} together with buffer qubit architecture~\cite{wu2023qucomm,liu2025hardware} might be necessary to reduce latency of probe state preparation over quantum networks, while vacuum beam guide~\cite{huang2024vacuum} can potentially offer additional benefits. Optimization of Bell-state-based graph state distribution~\cite{meignant2019distributing,de2020protocols,fischer2021distributing,avis2023analysis,bugalho2023distributing,ghaderibaneh2023generation,fan2024optimized,shimizu2024simple,negrin2024efficient,huang2025peer} is important for larger scale DQS as well. In real quantum networks entanglement generation will fail probabilistically, and we provide the discussion on the estimation strategy under entanglement generation failure in App.~\ref{app:est_w_fail}. We have also simulated the GHZ state distribution process over a 3-node quantum network with SeQUeNCe~\cite{wu2021sequence}, demonstrating the possibility of DQS quantum advantage with realistic quantum network stacks as shown in App.~\ref{app:qn_sim}. The new features we develop in SeQUeNCe to reflect imperfections in entanglement distribution are completely open-source~\cite{sequence-github}, and can thus serve as valuable resources for future quantum network research. We plan to investigate more complex quantum network stacks in the future.

In addition, privacy and security~\cite{shettell2022private,bugalho2024private,hassani2024privacy} is a potentially important aspect of realistic quantum sensor network as well. We acknowledge previous exploration of noisy probe state's effect in non-distributed sensing~\cite{datta2011quantum,ouyang2021robust}. Finally, although we mainly assume finite-dimensional matter-based quantum sensors, it is noteworthy that photonic systems~\cite{zhang2021distributed,xia2020demonstration,guo2020distributed,zhao2021field,liu2021distributed,kim2024distributed} are also playing an important role in DQS.

\begin{acknowledgments}
    We thank Tian-Xing Zheng and Boxuan Zhou for helpful discussions.
    This material is based upon work supported by the U.S. Department of Energy Office of Science National Quantum Information Science Research Centers. 
    Work performed at the Center for Nanoscale Materials, a U.S. Department of Energy Office of Science User Facility, was supported by the U.S. DOE, Office of Basic Energy Sciences, under Contract No. DE-AC02-06CH11357.
    A.Z. and T.Z. are also supported by the NSF Quantum Leap Challenge Institute for Hybrid Quantum Architectures and Networks (NSF Grant No. 2016136), and the Marshall and Arlene Bennett Family Research Program.
\end{acknowledgments}

\appendix

\onecolumngrid

\section{Probe state assembly from bipartite entanglement for DQS}\label{app:architecture}
Here we describe two methods to prepare multipartite entangled probe state for DQS with bipartite entanglement distributed by the quantum network and LOCC, namely gate teleportation and GHZ merging. In practice the quantum networks might be hybrid, in that different functions are realized by different physical systems. For instance, communication qubits which generate and distribute entanglement might be different from the quantum sensors used for DQS, thus the development of quantum interconnects~\cite{awschalom2021development,awschalom2022roadmap} is necessary.

\subsection{Gate teleportation}
A Bell pair can be used to perform CNOT teleportation, shown in Fig.~\ref{fig:assembly}(a). Then we can directly run the standard GHZ generation circuit, assuming the availability of two-qubit gates between communication qubits and sensor qubits. The standard GHZ generation circuit is shown in the first part of Fig.~\ref{fig:assembly}(b). 

We estimate the resources needed for this approach as follows. All generated Bell pairs are consumed. To generate probe state of size $N$, we need $N$ sensor qubits (one of them is the center) when performing CNOT teleportation. In total $(N-1)$ CNOTs are needed, thus $(N-1)$ Bell pairs. During each gate teleportation after two single qubit measurements, (in ideal case) there are 1/4 probability that no local unitary correction is needed, 1/2 probability that one local unitary correction is needed, and 1/4 probability that two local unitary corrections are needed. Moreover, the two single qubit measurements are in different basis. Since measurements are usually native in one basis, we may need an additional single-qubit unitary to transform measurement basis. Altogether we need: $(3N-2)$ qubits, $(2N-2)$ of which for $(N-1)$ Bell pairs; $(2N-2)$ single-qubit measurements; $(2N-2)$ local CNOTs; and on average $(N-1)$ (or $(2N-2)$) single-qubit gates. The estimation is summarized in Table~\ref{tab:assembly_estimation}.

\subsection{GHZ merging}
Merging of two GHZ states can be achieved with local CNOT and measurement feedforward, as shown explicitly in the following where we consider the merging of two GHZ states with $N$ and $M$ qubits, respectively:
\begin{align}
    &|\mathrm{GHZ}_N\rangle|\mathrm{GHZ}_M\rangle \propto (|\underbrace{0\dots 0}_{N\times~0}\rangle + |\underbrace{1\dots 1}_{N\times~1}\rangle)(|\underbrace{0\dots 0}_{M\times~0}\rangle + |\underbrace{1\dots 1}_{M\times~1})\rangle\\
    \xrightarrow[]{\mathrm{CNOT}_{N,N+1}}& (|\underbrace{0\dots 0}_{(N-1)\times~0}00\underbrace{0\dots 0}_{(M-1)\times~0}\rangle + |\underbrace{0\dots 0}_{(N-1)\times~0}01\underbrace{1\dots 1}_{(M-1)\times~1}\rangle + |\underbrace{1\dots 1}_{(N-1)\times~1}11\underbrace{0\dots 0}_{(M-1)\times~0}\rangle + |\underbrace{1\dots 1}_{(N-1)\times~1}10\underbrace{1\dots 1}_{(M-1)\times~1}\rangle).\nonumber
\end{align}
If we measure the target qubit of CNOT (with index $N+1$ when all qubits are indexed from 1 through $N+M$) in computational basis, when the measurement outcome is 0, the post-measurement state is exactly a GHZ state with $N+M-1$ qubits, and when the outcome is 1, the post-measurement state can be transformed into a GHZ state with $N+M-1$ qubits by applying $\min\{N,M-1\}$ local X gates to flip the $|0\rangle$ and $|1\rangle$, as shown in Fig.~\ref{fig:assembly}(b).

We estimate the resources needed for this approach as follows. Merging does not consume all qubits in Bell pairs. We need to perform $(N-1)$ mergings, and for each merging there is 1/2 probability that one local unitary correction is needed, and 1/2 probability that no local unitary correction is needed. Altogether we need: $(2N-2)$ qubits, all for $(N-1)$ Bell pairs; $(N-1)$ single-qubit measurements; $(N-1)$ local CNOTs; and on average $(N-1)/2$ single-qubit gates. The estimation is summarized in Table~\ref{tab:assembly_estimation}.

However, if we consider the separation between communication qubits and sensor qubits, we need to first generate the GHZ states across communication qubits from merging and then perform local SWAP gates to swap the GHZ state from communication qubits to the sensor qubits. In this case, we need $N$ sensor qubits, and $N$ local SWAP gates between communication and sensor qubits (which can be decomposed into 3 CNOTs).

\begin{figure}[t]
    \centering
    \includegraphics[width=0.6\linewidth]{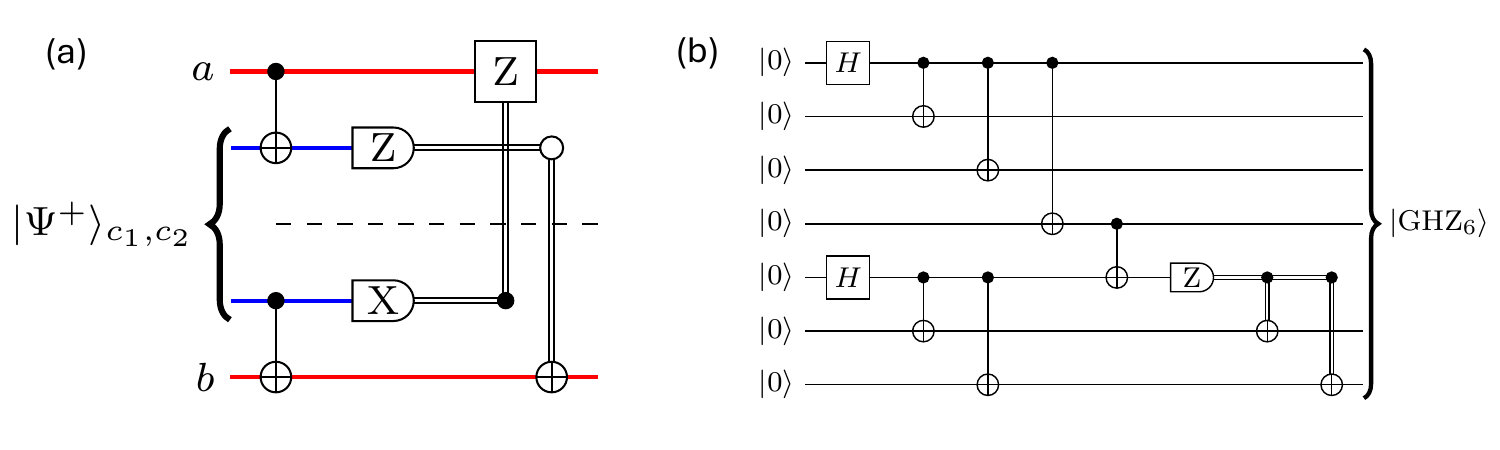}
    \caption{Circuits of (a) CNOT teleportation and (b) GHZ merging. The standard GHZ generation circuit is included in the GHZ merging circuit as the first part.}
    \label{fig:assembly}
\end{figure}

\begin{table}[t]
    \caption{\label{tab:assembly_estimation}Resource estimation for $N$-qubit GHZ state assembly from $(N-1)$ Bell pairs by CNOT teleportation and GHZ merging.}
    \begin{ruledtabular}
    \begin{tabular}{cccccc}
    & qubits\footnote{Including $(2N-2)$ qubits of the Bell pairs.} & 1-qubit measurements & 2-qubit gates & 1-qubit gates\footnote{Not including measurement basis transformation.} \\
    \colrule
    CNOT teleportation & $3N-2$ & $2N-2$ & $2N-2$\footnote{All are CNOTs.} & $N-1$\\
    GHZ merging & $2N-2$\footnote{All are from the distributed Bell pairs between communication qubits. $N$ sensor qubits are still needed.} & $N-1$ & $N-1$\footnote{Need $N$ SWAPs between sensor qubits and communication qubits.} & $(N-1)/2$\\
    \end{tabular}
    \end{ruledtabular}
\end{table}

\section{Justification of GHZ-diagonal initial probe state}\label{app:twirling}
Here we elaborate on the mechanisms that allow us to have GHZ-diagonal states as the initial probe for DQS.

\subsection{Pauli twirling}
Pauli twirling~\cite{dur2005standard,emerson2007symmetrized,dankert2009exact} can be utilized in two steps during the GHZ state distribution by the quantum network to make the final state in a GHZ-diagonal form. First of all, we distribute two-qubit Bell states between quantum sensor nodes. For each Bell state we can perform bilocal Pauli twirling: $\rho\to\sum_{i=1}^4(P_i\otimes P_i)\rho(P_i\otimes P_i)/4$ where $P_i$ are 1-qubit Pauli operators and the summation is over all 4 Pauli operators. It is easy to check that the resulting state is Bell-diagonal. Then, we need to assemble the Bell states together to create the global GHZ state (as reviewed in App.~\ref{app:architecture}). The assembly process involves local Clifford operations. Remarkably, the noise introduced by quantum gates, and in general quantum circuits that can be compiled into layers of quantum gates, can be effectively transformed into incoherent errors with Pauli channel an example, using techniques such as randomized compiling~\cite{wallman2016noise,hashim2021randomized}. As a result, since the assembly process only applies Clifford operations and introduces Pauli errors to Bell-diagonal states, the final state is effectively a pure GHZ state undergoing a Pauli channel, which is in GHZ-diagonal form. The twirling operations are also all local, so in principle applicable to quantum network scenarios.

\subsection{GHZ stabilizer twirling}
On the other hand, we can use GHZ stabilizer twirling~\cite{toth2007efficient} at the end of GHZ state distribution, to turn any GHZ state into GHZ-diagonal form without changing the diagonal elements in the density matrix. Explicitly, the GHZ state stabilizer twirling is a process of randomly applying the elements of the stabilizer group $\mathcal{G}=\{g_i\}$ for the GHZ state $\rho\to\sum_ig_i\rho g_i/\vert\mathcal{G}\vert$, where we have considered that stabilizer group elements are all Pauli strings so we have dropped the Hermitian conjugate. For completeness, we give a simple proof of the mechanism of GHZ stabilizer twirling.
\begin{proposition}
    GHZ stabilizer twirling turns any state into GHZ-diagonal form with diagonal elements unchanged.
\end{proposition}
\begin{proof}
    For an arbitrary $n$-qubit state $\rho$, we can always represent its density matrix in the GHZ basis. Now we explicitly apply the GHZ stabilizer twirling on $\rho$. 
    
    \textit{(1) Diagonal elements of the density matrix in GHZ basis.} For the element corresponding to the standard GHZ state this follows directly from the definition of stabilizer. Other GHZ basis states $|\psi\rangle$ can be obtained through applying Pauli strings to the standard GHZ state $|\psi\rangle=P|\mathrm{GHZ}\rangle$, where $P$ is a certain $n$-qubit Pauli string. Then the application of GHZ stabilizer elements on $|\psi\rangle$ gives 
    \begin{align}
        g_i|\psi\rangle = g_iP|\mathrm{GHZ}\rangle = \left(Pg_i - [P,g_i]\right)|\mathrm{GHZ}\rangle = (-1)^{f([P,g_i])}P|\mathrm{GHZ}\rangle = (-1)^{f([P,g_i])}|\psi\rangle,
    \end{align}
    where we define $f(\cdot)$ to output zero when the input is zero and output one otherwise, and we have used the fact that Pauli strings either commute or anti-commute. In other words, the application of GHZ stabilizers on a GHZ basis state gives either the basis state itself or with an additional $(-1)$ factor. For GHZ basis states $|\psi\rangle$ we have $g_i|\psi\rangle\langle\psi|g_i = \left[(-1)^{f([P,g_i])}\right]^2|\psi\rangle\langle\psi| = |\psi\rangle\langle\psi|$. Therefore, any diagonal element of the density matrix in GHZ basis will not be changed under GHZ stabilizer twirling.

    \textit{(2) off-diagonal elements.} Recall the the stabilizer group is Abelian and the $n$-qubit GHZ stabilizer group is generated by $n$ commuting $n$-qubit Pauli strings. Therefore, the additional factor, i.e. $(+1)$ or $(-1)$, coming from applying GHZ stabilizer element to GHZ basis states can be determined by the additional factor obtained from applying the GHZ stabilizer generators. For each GHZ basis state $|\psi_j\rangle$ where $j=1,\dots,2^n$ is the index, we denote the set of generators which give plus one factor $\mathcal{P}_j$ and the set of generators which give minus one factor $\mathcal{M}_j$, such that $\vert\mathcal{P}_j\vert+\vert\mathcal{M}_j\vert=n$. Note that for different $j$ it is impossible to have identical $\mathcal{P}_j$ and $\mathcal{M}_j$, because if so the two states are identical and should have the same index. 

    We have that the application of half of the stabilizer elements will result in an additional $(-1)$ factor, while the other half will result in $(+1)$ factor. This can be seen by considering $\mathcal{P}_j$ and $\mathcal{M}_j$, and let $\vert\mathcal{P}_j\vert=p_j$ and $\vert\mathcal{M}_j\vert=m_j$. The number of stabilizer elements that result in an additional $(-1)$ factor can be calculated as
    \begin{align}
        \left.2^{p_j}\sum_{i=0}^{\frac{m-1}{2}}\binom{m_j}{2i+1}\right\vert_{\mathrm{odd}~m_j} = \left.2^{p_j}\sum_{i=0}^{\frac{m-2}{2}}\binom{m_j}{2i+1}\right\vert_{\mathrm{even}~m_j} = 2^{m_j+p_j-1} = 2^{n-1} = \frac{\vert\mathcal{G}\vert}{2}.
    \end{align}
    This directly implies that $|\psi_j\rangle\langle\mathrm{GHZ}|$ and $|\mathrm{GHZ}\rangle\langle\psi_j|$ will vanish after GHZ stabilizer twirling due to the cancellation of $(-1)$ and $(+1)$ factor terms. For the application of GHZ stabilizer twirling to general $|\psi_j\rangle\langle\psi_k|$, we consider that arbitrary two GHZ basis states can be transformed into each other via application of a Pauli string. Therefore, we have
    \begin{align}
        g_i|\psi_j\rangle\langle\psi_k|g_i = g_i|\psi_j\rangle\langle\psi_j|Pg_i = (-1)^{f([P,g_i])}g_i|\psi_j\rangle\langle\psi_j|g_iP = (-1)^{f([P,g_i])}|\psi_j\rangle\langle\psi_j|P = (-1)^{f([P,g_i])}|\psi_j\rangle\langle\psi_k|.
    \end{align}
    Consequently, the above argument of evaluating the additional factor introduced from the application of GHZ stabilizer element still applies. We thus know that again half of the stabilizer elements give $f([P,g_i])=0$ while the other half gives $f([P,g_i])=1$. In the end, under GHZ stabilizer twirling any $|\psi_j\rangle\langle\psi_k|$ vanishes as long as $j\neq k$.
\end{proof}
The proof applies to arbitrary graph states, and thus any state can be twirled to be diagonal in arbitrary graph state basis. Moreover, the stabilizer elements are simply Pauli strings, so the stabilizer twirling can in principle be implemented over quantum networks.

\section{QFIM under unitary encoding}\label{app:qfi_ghz}
Suppose the encoded state has spectral decomposition $\rho=\sum_a\lambda_a|\psi_a\rangle\langle\psi_a|$ with all $\lambda_a$ being non-zero. Then the QFIM can be expressed as~\cite{liu2020quantum}
\begin{align}
    \mathcal{F}_{ij} =& \sum_a\frac{(\partial_i\lambda_a)(\partial_j\lambda_a)}{\lambda_a} + \sum_a 4\lambda_a\mathrm{Re}(\langle\partial_i\psi_a|\partial_j\psi_a\rangle) - \sum_{a,b}\frac{8\lambda_a\lambda_b}{\lambda_a+\lambda_b}\mathrm{Re}(\langle\partial_i\psi_a|\psi_b\rangle\langle\psi_b|\partial_j\psi_a\rangle).
\end{align}
We note that there are different equivalent ways of calculating QFIM, as documented in some other extensive reviews, for instance~\cite{paris2009quantum,petz2011introduction,toth2014quantum}.

Then we consider that the parameter dependence of $\rho$ comes from unitary encoding $U(x)$ of an $x$-independent initial state $\rho_0=\sum_a\lambda_{a0}|\psi_{a0}\rangle\langle\psi_{a0}|$, i.e. $\rho = U(x)\rho_0U(x)^\dagger$. Now the QFIM can be re-expressed in the convention of~\cite{liu2020quantum} as:
\begin{align}\label{Eq:QFIM}
    \mathcal{F}_{ij} = \sum_a 4\lambda_{a0}\mathrm{cov}_{|\psi_{a0}\rangle}(G_i,G_j) - \sum_{a\neq b}\frac{8\lambda_{a0}\lambda_{b0}}{\lambda_{a0}+\lambda_{b0}}\mathrm{Re}(\langle\psi_{a0}|G_i|\psi_{b0}\rangle\langle\psi_{b0}|G_j|\psi_{a0}\rangle),
\end{align}    
where $\mathrm{cov}_{|\psi\rangle}(A,B)$ denotes the covariance of observables $A$ and $B$ under a pure state $|\psi\rangle$:
\begin{align}
    \mathrm{cov}_{|\psi\rangle}(A,B) = \frac{1}{2}\langle\psi|\{A,B\}|\psi\rangle - \langle\psi|A|\psi\rangle\langle\psi|B|\psi\rangle,
\end{align}
and $G_i$ is the generator of parameter $x_i$:
\begin{align}
    G_i = i(\partial_iU^\dagger)U = -iU^\dagger(\partial_iU).
\end{align}

\subsection{Derivation of QFIM for GHZ-diagonal states}
Consider unitary parameter encoding $U(x) = \exp\left[-i\left(\sum_{i=1}^{d}x_iH_i\right)\right]$, where $d$ denotes the total number of sensor nodes in the network, $H_i=\frac{1}{2}\sum_{k=0}^{n-1}\sigma^{(i,k)}_z$ is the local collective spin of $n$ qubit sensors on node $i$, and the parameters $x_i$ physically correspond to accumulated phases through Hamiltonian evolution. The generators according to the above definition are $G_i = -H_i$, where the negative sign does not matter as generators appear in pairs in Eqn.~\ref{Eq:QFIM}. Then for $\rho_0$ which is a GHZ-diagonal state, we have that:
\begin{align}
    \mathrm{cov}_{|\psi_{m0}\rangle}(G_i,G_j) =& \mathrm{cov}_{|\psi_{m0}\rangle}(H_i,H_j) = \frac{1}{2}\langle\psi_{m0}|\{H_i,H_j\}|\psi_{m0}\rangle,\ \forall m,
\end{align}
because weight-1 Pauli strings do not stabilize the standard GHZ states. Meanwhile, we have that:
\begin{align}
    \langle\psi|\sigma_z^{(a)}\sigma_z^{(b)}|\psi\rangle = 1,\ \forall a,b,
\end{align}
where $a,b$ are indices for the qubits in $|\psi\rangle$ which is a GHZ-basis state. This is because weight-2 Pauli strings with only Pauli Z operators stabilize the standard GHZ states. Therefore, for any GHZ-diagonal $\rho_0$, the first single-index sum in Eqn.~\ref{Eq:QFIM} is always a constant:
\begin{align}
    \sum_a 4\lambda_{a0}\mathrm{cov}_{|\psi_{a0}\rangle}(G_i,G_j) =& \sum_a 4\lambda_{a0}\mathrm{cov}_{|\psi_{a0}\rangle}\left(\frac{1}{2}\sum_{k=0}^{n-1}\sigma^{(i,k)}_z,\frac{1}{2}\sum_{l=0}^{n-1}\sigma^{(j,l)}_z\right)\nonumber\\
    =& \sum_a \lambda_{a0}\sum_{k=0}^{n-1}\sum_{l=0}^{n-1}\mathrm{cov}_{|\psi_{a0}\rangle}\left(\sigma^{(i,k)}_z,\sigma^{(j,l)}_z\right) = n^2,\ \forall i,j.
\end{align}    
Then the derivation of the QFIM reduces to the second term in Eqn.~\ref{Eq:QFIM}.

There are $2^{nd}$ orthonormal $nd$-qubit GHZ states across $d$ sensor nodes where each node holds $n$ qubits, which can be labeled by the binary string from 0 through $2^{nd-1}-1$: These states are superpositions of two computational basis states which correspond to binary strings of $b\in\{0,1,\dots,2^{nd-1}-1\}$ and $2^{nd}-b-1$, respectively. The additional $2^{nd-1}$ states come from adding an additional $\pi$ relative phase between the two computational basis states. For any GHZ state $|\psi\rangle$, the application of a single Pauli Z will lead to a $\pi$ change in the relative phase between the two computational basis states in superposition. Therefore, $\langle\psi_{a0}|G_i|\psi_{b0}\rangle$ is either zero due to orthogonality between $|\psi_{a0}\rangle$ and $G_i|\psi_{b0}\rangle$, or one, and it takes unit value only when $|\psi_{a0}\rangle$ and $|\psi_{b0}\rangle$ correspond to the same length-$nd$ binary string and have relative phases which differ by $\pi$.

For a general $nd$-qubit GHZ-diagonal state $\rho_0=\sum_a\lambda_{a0}|\psi_{a0}\rangle\langle\psi_{a0}|$ where each node has $n$ qubits, we use $\mathcal{S}$ to denote the set of index pairs $(a,b)$ such that $|\psi_{a0}\rangle$ and $|\psi_{b0}\rangle$ are GHZ states as superposition of the same pair of computational basis states but with opposite relative phase. Note that for such pairs, only one of $(a,b)$ and $(b,a)$ is included in set $\mathcal{S}$. Then we have:
\begin{align}\label{Eq:noise_constant}
    &\sum_{a\neq b}\frac{8\lambda_{a0}\lambda_{b0}}{\lambda_{a0}+\lambda_{b0}}\mathrm{Re}(\langle\psi_{a0}|G_i|\psi_{b0}\rangle\langle\psi_{b0}|G_j|\psi_{a0}\rangle) = n^2\sum_{(a,b)\in\mathcal{S}}\frac{4\lambda_{a0}\lambda_{b0}}{\lambda_{a0}+\lambda_{b0}} = Cn^2,\ \forall i,j,
\end{align}
which again does not depend on $i,j$, while in practice can be dependent on $n$. 

The form of the series in Eqn.~\ref{Eq:noise_constant} implies that only Pauli Z errors affect the QFIM, as result of the assumed encoding channel, i.e. a z-axis coupling. Additionally, the maximum value of $C$ is 1 for GHZ-diagonal states, and it is achieved if and only if every pair of states in set $\mathcal{S}$ have identical eigenvalues. This can be seen as follows:
\begin{align}
    C = \sum_{(a,b)\in\mathcal{S}}\frac{4\lambda_{a0}\lambda_{b0}}{\lambda_{a0}+\lambda_{b0}} = \sum_{(a,b)\in\mathcal{S}}\frac{(\lambda_{a0}+\lambda_{b0})^2-(\lambda_{a0}-\lambda_{b0})^2}{\lambda_{a0}+\lambda_{b0}} = 1 - \sum_{(a,b)\in\mathcal{S}}\frac{(\lambda_{a0}-\lambda_{b0})^2}{\lambda_{a0}+\lambda_{b0}},
\end{align}
where the subtracted term is non-negative, and it becomes zero if and only if $\lambda_{a0}=\lambda_{b0}$, $\forall(a,b)\in\mathcal{S}$.

Combining the above results, we have the QFIM with respect to local parameters $x=(x_1,\dots,x_d)^T$
\begin{align}
    \mathcal{F}(x) = (1-C)n^2
    \begin{pmatrix}
        1 & 1 & \dots & 1 \\
        1 & 1 & \dots & 1 \\
        \vdots & \vdots & \ddots & \vdots \\
        1 & 1 & \dots & 1 \\
    \end{pmatrix}.
\end{align}
According to the parameter estimation problem of our interest, we have that $v_1\propto (1,\dots,1)^T\in\mathbb{R}^d$. For concreteness, we may choose $v_1=(1,\dots,1)^T/\sqrt{d}$, which is a normalized vector under 2-norm. As now we only focus on estimating one parameter $v_1^Tx$, we can construct an orthonormal matrix $M=(v_1,\dots,v_d)^T$ such that $v_i^Tv_j=\delta_{ij}$. 
When we consider $d$ new derived parameters $\theta = (f_1(x),\dots,f_d(x))^T\in\mathbb{R}^d$ from the original parameters $x$, the QFIM of $\theta$ can be expressed in terms of the QFIM of $x$ of as~\cite{paris2009quantum}:
\begin{align}\label{Eq:QFIM_transform}
    \mathcal{F}(\theta) = J^T\mathcal{F}(x)J,
\end{align}
where $J$ is the Jacobian matrix whose matrix elements are $J_{ij} = \partial x_i/\partial\theta_j$.
Then according to Eqn.~\ref{Eq:QFIM_transform} we have the QFIM with respect to new parameters $\theta=(v_1^Tx,\dots, v_d^Tx)$:
\begin{align}
    \mathcal{F}(\theta) =& \sqrt{d}(1-C)n^2
    \begin{pmatrix}
        v_1^T \\
        v_2^T \\
        \vdots \\
        v_d^T \\
    \end{pmatrix}
    \begin{pmatrix}
        v_1 & v_1 & \dots & v_1
    \end{pmatrix}
    \begin{pmatrix}
        v_1 & v_2 & \dots & v_d
    \end{pmatrix}\nonumber\\
    =& \sqrt{d}(1-C)n^2
    \begin{pmatrix}
        1 & 1 & \dots & 1 \\
        0 & 0 & \dots & 0 \\
        \vdots & \vdots & \ddots & \vdots \\
        0 & 0 & \dots & 0 \\
    \end{pmatrix}
    \begin{pmatrix}
        v_1 & v_2 & \dots & v_d
    \end{pmatrix}\nonumber\\
    =& d(1-C)n^2
    \begin{pmatrix}
        v_1^T \\
        0 \\
        \vdots \\
        0 \\
    \end{pmatrix}
    \begin{pmatrix}
        v_1 & v_2 & \dots & v_d
    \end{pmatrix} = d(1-C)n^2
    \begin{pmatrix}
        1 & 0 & \dots & 0 \\
        0 & 0 & \dots & 0 \\
        \vdots & \vdots & \ddots & \vdots \\
        0 & 0 & \dots & 0 \\
    \end{pmatrix}.
\end{align}

\subsection{Proof of Proposition~\ref{prop:worst_case_qfi}}
\begin{proof}
    We label the the eigenvalues of the density matrix in the following way. GHZ states can be expressed as a superposition of two computational basis states corresponding to two binary strings. Further, one bit string corresponds to a smaller integer $n\in\{0,1,\dots,2^{d-1}-1\}$ while the other corresponds to a larger integer $m=2^d-n-1$. For a GHZ state corresponding to integers $n$ and $m$ with $n<m$, we let its index be $2n$ if the relative phase between two computational basis states is $0$, and $2n+1$ otherwise. For instance, the standard GHZ state $|\mathrm{GHZ}_d\rangle$ has index 0, and $Z|\mathrm{GHZ}_d\rangle$ has index 1. Then the eigenvalue corresponding to $|\mathrm{GHZ}_d\rangle$ is $\lambda_{00}=F$. For the rest $2^d-1$ eigenvalues with $i\in\{1,2,\dots,2^d-1\}$, we can express them as $\lambda_{i0}=(1-F)p_i$, s.t. $p_i\in[0,1]$ and $\sum_{i=1}^{2^d-1}p_i = 1$.

    The objective then becomes finding the combination of $(p_1,\dots,p_{2^d-1})$ which gives the highest $C$ under the above constraints. The constraints clearly define a closed and compact region $\mathcal{R}$. Then according to the extreme value theorem, there exist a maximum value and a minimum value for any continuous function of $(p_1,\dots,p_{2^d-1})$ on $\mathcal{R}$, and the extreme values must be taken either on the boundary of $\mathcal{R}$, or at critical points inside $\mathcal{R}$.
    
    Now $C$ can be re-written as:
    \begin{align}
        C = \frac{4F(1-F)p_1}{F + (1-F)p_1} + \sum_{i=1}^{2^{d-1}-1}\frac{4(1-F)p_{2i}p_{2i+1}}{p_{2i} + p_{2i+1}},
    \end{align}
    which is obviously continuous on $\mathcal{R}$. We first take the partial derivatives with respect to $p_i$:
    \begin{align}
        &\frac{\partial}{\partial p_1}C = \frac{4(1-F)F^2}{\left[F + (1-F)p_1\right]^2} > 0,\\
        &\frac{\partial}{\partial p_i}C = \frac{4(1-F)p_{i+(-1)^{i\,\mathrm{mod}\,2}}^2}{\left(p_i + p_{i+(-1)^{i\,\mathrm{mod}\,2}}\right)^2} \geq 0,~i>1,
    \end{align}
    which means that there is no critical point inside $\mathcal{R}$, so the maximum value can only be on the boundary of $\mathcal{R}$. 
    
    The boundary of $\mathcal{R}$ can be divided into $2^d-1$ parts $\{B_i\}$, each determined by $p_i=0$, $i\in\{1,2,\dots,2^d-1\}$. Under the equality constraint $\sum_{i=1}^{2^d-1}p_i = 1$, the above partial derivatives suggest that there is still no critical point inside any $B_i$. Then the extreme values on the boundary should be on the boundary of boundary $B_i$, i.e, $\{B_{ij}\}$, where the subscript means the $j$-th part of the boundary of $B_i$. The conclusion of no inside critical point holds until we have reduced the boundary into zero-dimension points, characterized by $p_i=1$.

    Finally, we can compare the $C$ values for all the $2^d-1$ choices of $(p_1,\dots,p_{2^d-1})$. If $p_i=1$ and $i>1$, $C=0$, and if $p_1=1$, $C=4F(1-F)>0$. Therefore, the maximum value of $C$ on $\mathcal{R}$ is $4F(1-F)$ corresponding to $\rho=F|\mathrm{GHZ}_d\rangle\langle\mathrm{GHZ}_d| + (1-F)Z|\mathrm{GHZ}_d\rangle\langle\mathrm{GHZ}_d|Z$, which gives the lowest QFI $\mathcal{F}=d(1-C)=d(2F-1)^2$.
\end{proof}

\subsection{Proof of Proposition~\ref{prop:ent_detect}}
\begin{proof}
    Firstly, for a pure state $|\varphi\rangle$ undergoing unitary parameter encoding whose all generators commute with each other as in our case, the QFIM can be calculated as 
    \begin{align}
        \mathcal{F}_{ij} = 4\left[\frac{1}{2}\langle\varphi|\{H_i,H_j\}|\varphi\rangle - \langle\varphi|H_i|\varphi\rangle\langle\varphi|H_j|\varphi\rangle\right].
    \end{align}
    Now we consider a fully separable $d$-qubit pure state $|\varphi\rangle_\mathrm{sep} = \bigotimes_{i=1}^d|\varphi_i\rangle$. Its QFIM under the aforementioned unitary parameter encoding channel is thus
    \begin{align}
        \mathcal{F}_{ij} = \frac{1}{2}\langle\varphi|_\mathrm{sep}\{H_i,H_j\}|\varphi\rangle_\mathrm{sep} - \langle\varphi|_\mathrm{sep}H_i|\varphi\rangle_\mathrm{sep}\langle\varphi|_\mathrm{sep}H_j|\varphi\rangle_\mathrm{sep} = \left[1 - \left(\langle\varphi_i|\sigma_z|\varphi_i\rangle\right)^2\right]\delta_{ij} \leq \delta_{ij}.
    \end{align}
    Therefore, for any fully separable $d$-qubit pure state $\mathcal{F}\leq I$, i.e. $I-\mathcal{F}$ is positive semidefinite. 
    
    Then we consider arbitrary fully separable $d$-qubit state that can be expressed as a convex combination of tensor products of subsystem quantum states $\rho_\mathrm{sep}=\sum_ip_i\rho_i^{(1)}\otimes\rho_i^{(2)}\otimes\dots\otimes\rho_i^{(d)}$, where $p_i\geq 0,~\sum_ip_i=1$, and in our case $\rho_i^{(j)}$ is an arbitrary single-qubit density matrix. Notice that every $\rho_i^{(j)}$ can be decomposed as a convex combination of single-qubit pure state, so any fully separable state can be decomposed as a convex combination of $d$-qubit fully separable pure states $\rho_\mathrm{sep}=\sum_jq_j|\varphi_j^{(1)}\rangle\langle\varphi_j^{(1)}|\otimes\dots\otimes|\varphi_j^{(d)}\rangle\langle\varphi_j^{(d)}| = \sum_jq_j|\varphi_j\rangle\langle\varphi_j|_\mathrm{sep}$, where $q_j\geq 0$ and $\sum_jq_j=1$. 
    
    Recall that the QFIM is convex~\cite{liu2020quantum}, so we have
    \begin{align}
        \mathcal{F}(\rho_\mathrm{sep}) = \mathcal{F}\left(\sum_jq_j|\varphi_j\rangle\langle\varphi_j|_\mathrm{sep}\right) \leq \sum_jq_j\mathcal{F}(|\varphi_j\rangle\langle\varphi_j|_\mathrm{sep}) \leq I.
    \end{align}
    Also, since QFIM is real symmetric and positive semidefinite, we have the orthogonal decomposition $\mathcal{F}(\rho_\mathrm{sep}) = Q\Lambda Q^T$, where $QQ^T=Q^TQ=I$ and $\Lambda$ is diagonal s.t. $0\leq\Lambda\leq I$. We then transform from the current basis of ``natural'' parameters to another orthonormal basis including $\theta_1=v_1^Tx$, through an orthogonal transform $M$ of $\mathcal{F}(\rho_\mathrm{sep})$
    \begin{align}
        \tilde{\mathcal{F}}(\rho_\mathrm{sep}) = M\mathcal{F}(\rho_\mathrm{sep})M^T = MQ\Lambda Q^TM^T = O\Lambda O^T,
    \end{align}
    where we use $\tilde{\mathcal{F}}$ to denote the transformed QFIM for the new parameters, and $O=MQ$ is another orthogonal transform s.t. $OO^T=O^TO=I$.

    For estimation of a single parameter in a multiparameter scenario, we have the following bound~\cite{proctor2017networked}
    \begin{align}
        \mathrm{Var}(\hat{\theta}_1) \geq \frac{\left(\tilde{\mathcal{F}}^{-1}(\rho_\mathrm{sep})\right)_{11}}{\mu} \geq \frac{1}{\mu\left(\tilde{\mathcal{F}}(\rho_\mathrm{sep})\right)_{11}},
    \end{align}
    where the first inequality is always saturable when we only focus on $\theta_1$, while the second, though simpler, is not always saturable. Then we explicitly evaluate $\left(\tilde{\mathcal{F}}(\rho_\mathrm{sep})\right)_{11}$ to bound it from the above
    \begin{align}
        \left(\tilde{\mathcal{F}}(\rho_\mathrm{sep})\right)_{11} = \sum_{i=1}^d\lambda_i\left(w_1^{(i)}\right)^2 \leq \sum_{i=1}^d\left(w_1^{(i)}\right)^2 = w_1^Tw_1= 1,
    \end{align}
    where we have written $O=(w_1,\dots,w_d)^T$ with $w_i^Tw_j=\delta_{ij}$, while $w_i^{(i)}$ is the $j$-th element of $w_i$, and $\lambda_i$ are the diagonal elements of $\Lambda$. From the upper bound of $\left(\tilde{\mathcal{F}}(\rho_\mathrm{sep})\right)_{11}$ the desired lower bound of $\theta_1$ estimation variance is then obvious.
\end{proof}
Notice that $1/\mu$ is exactly the estimation variance of the optimal local sensing strategy. The proof can be straightforwardly extended to the scenario where each sensor contains $n$ qubits, to show that any fully separable state with respect to the partition of $d$ sensors cannot achieve $\theta_1$ estimation variance lower than $1/(n^2\mu)$.

\section{Details for unitary encoding with noisy local entanglement generation}\label{app:noisy_eg}
We first justify the phenomenological model for noisy local entanglement generation which we use in this work when the errors are not too biased. Then we give additional details about the maximum number of local sensors. We also include noise in local entanglement generation for the optimal local sensing strategy, where each node aims to generate a GHZ state to estimate the local parameter.

\subsection{Justification of the phenomenological model for noisy local entanglement generation}
We consider noisy CNOT models where there is $F_\mathrm{CNOT}$ probability of realizing the noiseless CNOT and $(1-F_\mathrm{CNOT})$ probability of applying Pauli noise channel on the involved qubits. When $(d-1)(n-1)$ CNOT gates are applied, there will be $F_\mathrm{CNOT}^{(d-1)(n-1)}\sim k^{n-1}$ probability to obtain the ideal output state, and $(1-F_\mathrm{CNOT}^{(d-1)(n-1)})\sim(1-k^{n-1})$ probability to have additional Pauli errors applied to the final state. Our initial state is a GHZ-diagonal state with GHZ fidelity $F$. Thus, we have $Fk^{n-1}$ probability of directly obtaining the ideal GHZ state, $(1-F)k^{n-1}$ probability of obtaining mixture of other orthogonal GHZ states, $F(1-k^{n-1})$ probability of applying Pauli errors on the ideal GHZ state, and $(1-F)(1-k^{n-1})$ probability of applying Pauli errors on the mixture of other orthogonal GHZ states. 

For each GHZ basis state there are $2^{nd}$ out of $4^{nd}$ Pauli strings which will transform it into the objective GHZ state. However, if the errors are not too biased, each error operator has similar probability to be applied. The errors which stabilize the objective GHZ state will only contribute roughly $(2^{nd}/4^{nd})(1-k^{n-1}) = (1-k^{n-1})/2^{nd}\sim (1/2)^{nd}$ to the fidelity. Similarly, the errors that transform other orthogonal GHZ states to the objective GHZ state will contribute roughly $(2^{nd}/4^{nd})(1-F)(1-k^{n-1})\sim (1/2)^{nd}$. In practice $k\sim F_\mathrm{CNOT}^{d-1}$ is significantly large than $1/2$ for good local entanglement generation and moderate $d$, which means that $(1/2)^{nd}\ll Fk^{n-1}$. The above thus justifies that the fidelity of $nd$-qubit GHZ state after noisy local entanglement generation can be approximated by $Fk^{n-1}$. 

Then we consider error models which are not too biased and close to uniform distribution. Recall the explicit calculation of the parameter $C$ for QFI: 
\begin{align}
    C = 1 - \sum_{(a,b)\in\mathcal{S}}\frac{(\lambda_{a0}-\lambda_{b0})^2}{\lambda_{a0}+\lambda_{b0}}.
\end{align}
For GHZ state that can potentially demonstrate quantum advantage, we know that its fidelity $F>1/d\gg 2^{-nd}$. For error distribution close to uniform distribution, we have that $\lambda_{a0}+\lambda_{b0}\sim (1-F)2^{1-nd}$, and $|\lambda_{a0}-\lambda_{b0}| \sim f(\lambda_{a0}+\lambda_{b0})\sim f(1-F)2^{1-nd}$, where $f$ characterizes the fluctuation in the error distribution. Then the summation for calculating $C$ can be estimated as
\begin{align}
    \sum_{(a,b)\in\mathcal{S}}\frac{(\lambda_{a0}-\lambda_{b0})^2}{\lambda_{a0}+\lambda_{b0}} \sim \frac{(F-\lambda)^2}{F+\lambda} + (2^{nd-1}-1)\frac{\left[f(1-F)2^{1-nd}\right]^2}{(1-F)2^{1-nd}} \sim F + f^2(1-F),
\end{align}
where $\lambda\sim 2^{-nd}$ denotes the fidelity of $(|0\dots0\rangle-|1\dots 1\rangle)/\sqrt{2}$. For unbiased error distributions, we have $f\ll 1$, so that $f^2(1-F)\ll F$. We can visualize the the ratio between the estimated contribution to the summation for $C$ from $(|0\dots0\rangle+|1\dots 1\rangle)/\sqrt{2}$ and other errors terms, i.e. $f^2(1-F)/F$, in Fig.~\ref{fig:error_ratio}. It is clear that the contribution from error components in the GHZ-diagonal state is negligible when the fluctuation in the error distribution is small, i.e. when the error model is not too biased. 

\begin{figure}[t]
    \centering
    \includegraphics[width=0.45\linewidth]{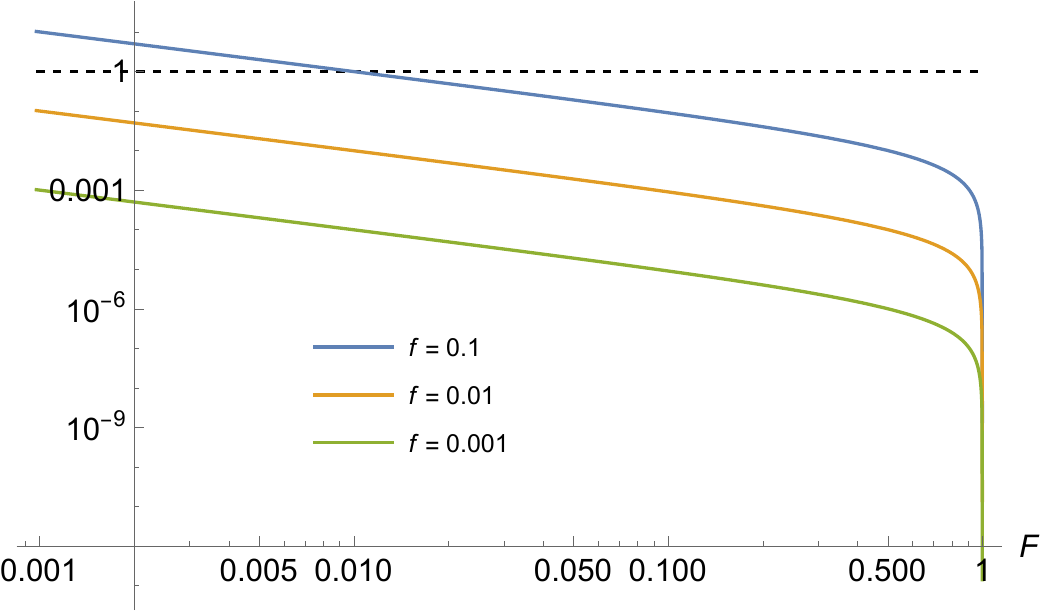}
    \caption{Visualization of the ratio between the estimated contribution from the standard GHZ component and other error terms, $f^2(1-F)/F$ as function of GHZ fidelity $F$ for error distribution fluctuation factors $f=0.1,0.01,0.001$.}
    \label{fig:error_ratio}
\end{figure}

\subsection{Additional details about the maximum number of local sensors}

\begin{figure}[t]
    \centering
    \includegraphics[width=\linewidth]{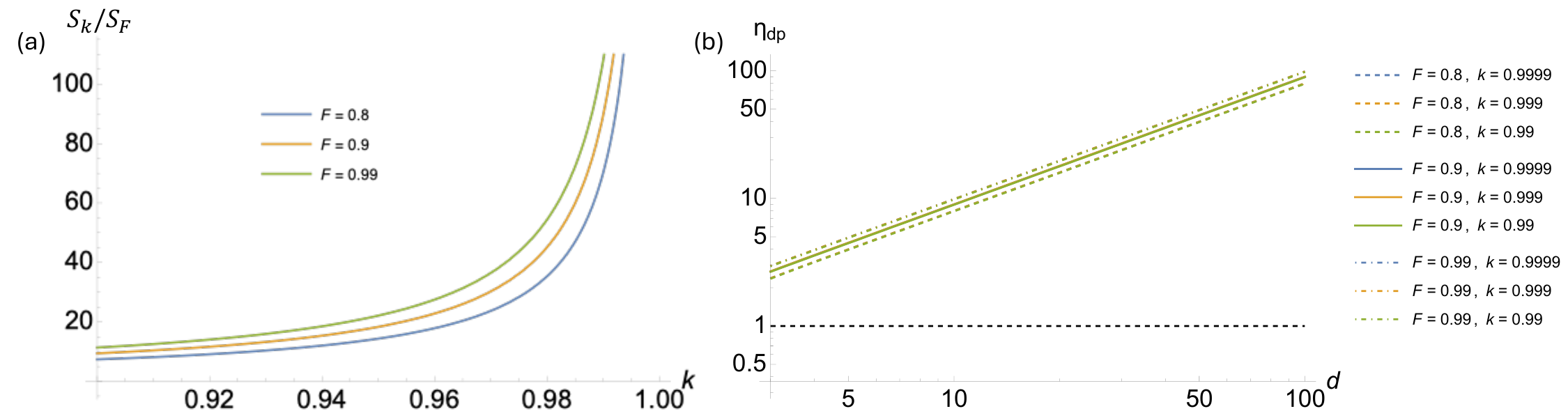}
    \caption{Visualization for $n_\mathrm{max}$. (a) Ratio between $n_\mathrm{max}$'s sensitivity to $k$ and sensitivity to $F$, i.e. $S_k/S_F$, where $d=3$ is fixed. (b) $\eta_\mathrm{dp}$ as function of the number of sensor nodes $d$, for some choices of initial $d$-qubit GHZ state fidelity $F$ and local entanglement generation quality $k$. We fix the number of local sensors $n=2$ while varying $F$ and $k$. Dashed curves correspond to $F=0.8$; solid curves correspond to $F=0.9$; dot-dashed curves correspond to $F=0.99$. Blue color represents $k=0.9999$; yellow color represents $k=0.999$; green color represents $k=0.99$.}
    \label{fig:nmax_details}
\end{figure}

We visualize the behavior of the sensitivity ratio $S_k/S_F$ when $d=3$ in Fig.~\ref{fig:nmax_details} (a). In addition, we consider the scenario where the number of local quantum sensors is fixed. The focus is thus on the impact from local entanglement generation quality $k$ and $d$-qubit GHZ state fidelity $F$, respectively. To this end, we examine the partial derivatives of relative advantage $\eta_\mathrm{dp}(F,d,n,k)=d[1-C_\mathrm{dp}(F,d,n,k)]$ w.r.t. $k$ and $F$, respectively.
\begin{align}
    & \frac{\partial}{\partial k}\eta_\mathrm{dp} = \frac{Fk^{n-2}d\left(2^{nd}k^nF-k\right)\left[\left(3\times 2^{nd} - 2\right)k + \left(2^{nd} - 2\right)2^{nd}k^nF\right](n-1)}{\left(2^{nd} - 1\right)\left[\left(2^{nd} - 2\right)k^nF-k\right]^2},\\
    & \frac{\partial}{\partial F}\eta_\mathrm{dp} =  \frac{k^{n-1}d\left(2^{nd}k^nF-k\right)\left[\left(3\times 2^{nd} - 2\right)k + \left(2^{nd} - 2\right)2^{nd}k^nF\right]}{\left(2^{nd} - 1\right)\left[\left(2^{nd} - 2\right)k^nF-k\right]^2}.
\end{align}
The partial derivatives are non-negative when $\left(2^{nd}k^nF-k\right)\geq 0$, which is generally satisfied in realistic parameter regimes. Moreover, despite the apparently complicated form of each individual expression, we have a simple expression for their ratio
\begin{align}
    \frac{\partial\eta_\mathrm{dp}/\partial k}{\partial\eta_\mathrm{dp}/\partial F} = (n-1)\frac{F}{k},
\end{align}
which demonstrate that when fixed $n$ is small and when $F$ is in general smaller than $k$, the quantum network entanglement distribution quality could indeed have higher impact on the relative advantage. We demonstrate such effect in Fig.~\ref{fig:nmax_details} (b). It is shown that curves with identical line style but different colors all overlap, but curves with different line styles are still separable from each other. This reveals that the impact from different GHZ state fidelity $F$ is greater than the local entanglement generation quality $k$, because $n=2$ is small and for the considered parameter values we have $k>F$.

\subsection{Comparison with local imperfect GHZ state}
We consider a specific local sensing strategy, where each node utilizes $n$-qubit imperfect GHZ state to estimate the local parameter. In such single-parameter estimation scenario for each node, the QFI be calculated as a special case of the multiparameter formula. Similar to the previous twirling argument, without loss of generality, we consider local initial probe state to be in the form of GHZ-diagonal state. Then through analysis that is the same as the multiparameter QFIM case, we have the analytical formula of QFI for initial GHZ-state diagonal state $\mathcal{F} = (1-C)n^2$, which corresponds to $d=1$ special case of the multiparameter formula.

We focus on the impact of imperfect initial probe state preparation, and thus assume that the optimal measurement on each local node can be performed, so the QCRB for each sensor node's local estimation can be achieved. Then according to the propagation of error, we have the estimation error for the average of all local parameters with the local strategy:
\begin{align}
    \mathrm{Var}_\mathrm{local}(\hat{\theta}_1) = \sum_{l=1}^d\left(\frac{\partial\theta_1}{\partial x_l}\right)^2\mathrm{Var}(\hat{x}_l) = \frac{1}{n^2\mu}\sum_{l=1}^d\frac{1}{d(1-C_l)},
\end{align}
for $\mu$ repetitions of measurements, $d$ sensor nodes, and $n$ qubits per node.

We consider the following model of imperfect local entanglement generation: The generated local probe state at each node is an identical depolarized GHZ state with fidelity $F(n) = \Tilde{k}^{n-1}$ as a function of the number of local sensors $n$, where $\Tilde{k}\in(0,1)$ is a constant representing the quality of local entanglement generation and the higher the better. To make the comparison with the global strategy, we consider the following correspondence: $\Tilde{k} = \sqrt[d]{k}$, where $k$ is the same constant as we used for the fidelity of the $nd$-qubit global probe state with imperfect local entanglement generation. This correspondence is motivated by the fact that when $n$ increases one, $d$ qubits are added to the global probe state, while only one qubit is added to each node's local probe state. Given this model, we can further express $\mathrm{Var}_\mathrm{local}(\hat{\theta}_1)$ as:
\begin{align}
    \mathrm{Var}_\mathrm{local}(\hat{\theta}_1) = \frac{1}{(1-C_\mathrm{local})n^2N},
\end{align}
where
\begin{align}
    &C_\mathrm{local}(d,n,k) = \frac{\left(1 - k^{(n-1)/d}\right) \left[(2^n-2)k^{1/d} + 4^nk^{n/d}\right]}{(2^n-1)\left[k^{1/d} + (2^n-2)k^{n/d}\right]}.
\end{align}

Then we can make the explicit comparison between $\eta_\mathrm{global}=d[1-C_\mathrm{dp}(F,d,n,k)]$ for the global strategy and $\eta_\mathrm{local}=(1-C_\mathrm{local})$ for the local strategy. As an example, we visualize the ratio $r=\eta_\mathrm{global}/\eta_\mathrm{local}$ when $F=0.9$, $d=3$, and $k=0.9999$, in Fig.~\ref{fig:global-vs-local-imperfect}. It can be observed that the threshold of local sensor number for advantage over imperfect local strategy is higher than the threshold for advantage over the optimal local strategy. This is intuitive as the baseline in the former scenario is worse. However, it is important to re-emphasize that the global strategy performance will become worse than the specific local strategy, even if imperfect local entanglement generation is included in the local strategy. This reinforces the fundamental limit in the advantage of global strategy when local entanglement generation is imperfect. In fact, this can be seen analytically through the asymptotic analysis of $\eta_\mathrm{local}$ and $\eta_\mathrm{global}$ (under the reasonable assumption of $1>k>1/2$):
\begin{align}
    \eta_\mathrm{global} \sim& (dF)k^{n-1},~\eta_\mathrm{local} \sim \left(k^{1/d}\right)^{n-1},
\end{align}
which means that as $n$ increases $\eta_\mathrm{global}$ will always drop below $\eta_\mathrm{local}$.
\begin{figure}[t]
    \centering
    \includegraphics[width=0.45\linewidth]{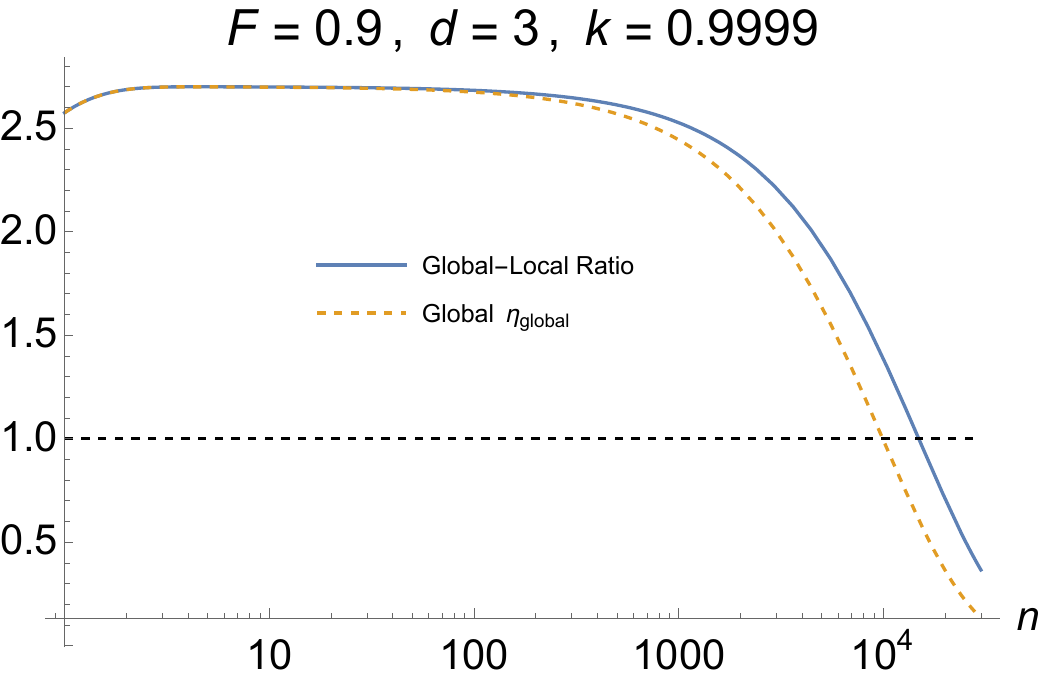}
    \caption{Visualization of the comparison between global and local sensing strategies with imperfect local entanglement generation, when $F=0.9$, $d=3$, and $k=0.9999$. The blue solid line denotes $r=\eta_\mathrm{global}/\eta_\mathrm{local}$, and the yellow dashed line denotes $\eta_\mathrm{global}$, which is identical to $\eta_\mathrm{dp}$ in the main text. The black dashed line marks the baseline value of 1: If $\eta_\mathrm{global}>1$ there is quantum advantage over the optimal local strategy, and if $r>1$ there is advantage of using global strategy over local strategy when local entanglement generation is imperfect for both.}
    \label{fig:global-vs-local-imperfect}
\end{figure}

\section{Details for unitary encoding with local measurement constraint}\label{app:loc_meas}
According to propagation of error, we can obtain the variance of estimating a certain parameter encoded in a quantum state from the measurement results of an observable $M$~\cite{toth2014quantum}:
\begin{align}\label{Eq:error_propagation}
    \mathrm{Var}_M(\hat{\theta}_1) = \frac{\langle M^2\rangle - \langle M\rangle^2}{\left\vert\frac{\partial}{\partial\theta_1}\langle M\rangle\right\vert^2},
\end{align}
where the expectation value $\langle M\rangle$ is taken under the encoded state $\rho_x$ and thus is a function of local parameters $x_i$. In quantum sensing the values of parameters to estimate are usually small so we may take the limit of $\theta\rightarrow 0$. Given a specific $d$-dimensional orthonormal basis, we will be able to transform the $x$-dependence of $\langle M\rangle$ into $\theta$-dependence. Then the partial derivative with respect to $\theta_1$ should be straightforward.

\subsection{Optimal measurement for noiseless case is useless for noisy case}
Here we consider the observable $M=\sigma_x^{\otimes nd}$. Since $M^2 = I^{\otimes nd}$ we always have $\langle M^2\rangle=1$. Consider GHZ states corresponding to the same binary strings but with opposite relative phase $|\mathrm{GHZ}_{nd}^\pm(b)\rangle = (|b\rangle\pm|2^{nd}-b-1\rangle)/\sqrt{2}$, where $|x\rangle$ denote the computational basis state corresponding to the binary representation of integer $x=0,1,\dots,2^{nd}-1$ and $b=0,1,\dots,2^{nd-1}-1$. We have
\begin{align}
    \langle\mathrm{GHZ}_{nd}^+(b)|U^\dagger(x)MU(x)|\mathrm{GHZ}_{nd}^+(b)\rangle =& \frac{1}{2}\left(\langle b| + e^{-i\phi(x)}\langle 2^{nd}-b-1|\right)\sigma_x^{\otimes nd}\left(|b\rangle+ e^{i\phi(x)}|2^{nd}-b-1\rangle\right) \nonumber\\
    =& \frac{1}{2}\left(e^{i\phi(x)} + e^{-i\phi(x)}\right) = \cos\left(\phi(x)\right) \nonumber\\
    =& -\langle\mathrm{GHZ}_{nd}^-(b)|U^\dagger(x)MU(x)|\mathrm{GHZ}_{nd}^-(b)\rangle
\end{align}    

Given the above properties of GHZ-diagonal state under the encoding channel and observable of our interest, the analysis of depolarized GHZ state becomes significantly simplified. Only $(|0\dots 0\rangle\pm|1\dots 1\rangle)/\sqrt{2}$ contribute to the expectation value $\langle M\rangle$, because other GHZ states corresponding to the same binary string have identical weights and thus their contributions cancel each other. Thus for depolarization error model we have:
\begin{align}
    \langle M\rangle_n(\theta_1) = \left(F - \frac{1-F}{2^{nd}-1}\right)\cos\left(n\sqrt{d}\theta_1\right),
\end{align}
where $n$ denotes the number of sensors per node, and noiseless local entanglement generation corresponds to $n=1$. Then we can substitute the above into Eqn.~\ref{Eq:error_propagation}:
\begin{align}
    \mathrm{Var}_M(\hat{\theta}_1) = \frac{1 - \left(F - \frac{1-F}{2^{nd}-1}\right)^2\cos^2\left(n\sqrt{d}\theta_1\right)}{dn^2\left(F - \frac{1-F}{2^{nd}-1}\right)^2 \sin^2\left(n\sqrt{d}\theta_1\right)}
\end{align}
When taking the above function to the limit of $\theta_1\rightarrow 0$ it goes to infinity if $F<1$, which suggests that this specific measurement scheme is useless to estimate small values with high accuracy. The divergence of estimation variance in small local parameter regime is general for any GHZ-diagonal state with fidelity below 1, because the denominator in error propagation will approach zero, while the numerator will stay non-zero when the fidelity is not equal to 1.

\subsection{Optimization of restricted local measurement}
\begin{figure*}[t]
    \centering
    \includegraphics[width=0.45\linewidth]{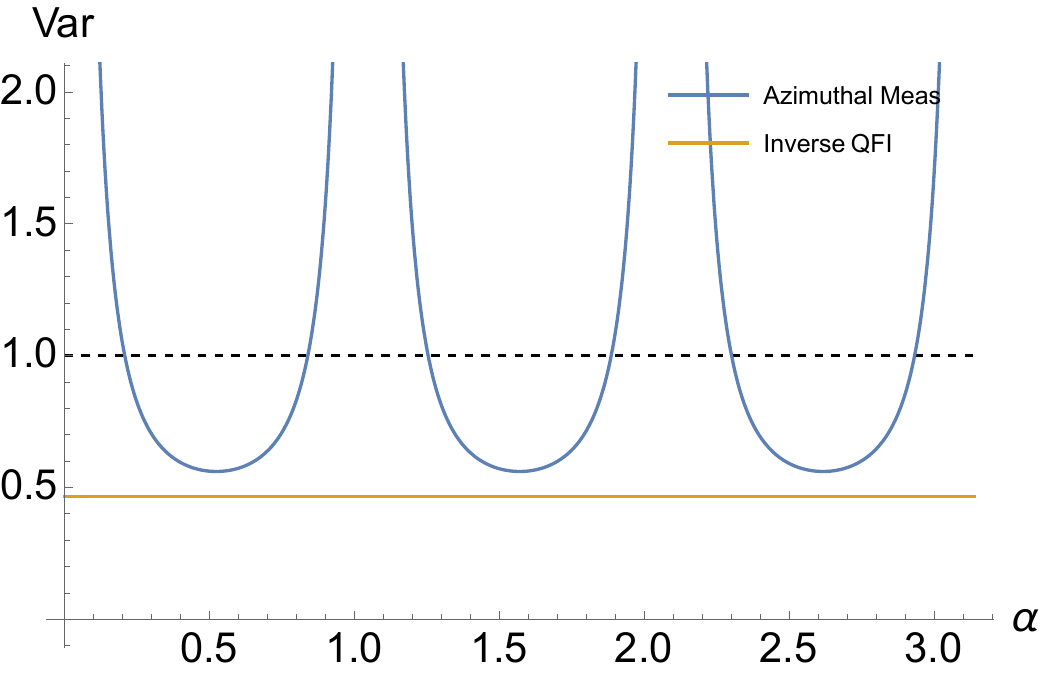}
    \caption{The estimation variance as a function of azimuthal angle $\alpha$ is shown as the blue solid curve, for $F=0.8,d=3,n=1$. The black dashed line denotes the baseline which can be achieved by the best local strategy. The yellow solid line demonstrates the inverse of the QFI, representing the ultimate estimation variance which requires some other measurement.}
    \label{fig:azimuthal_meas}
\end{figure*}
Recall the family of measurements as considered in the main text $M(\alpha)=\left[O(\alpha)\right]^{\otimes nd}$ where $O(\alpha) = |\psi^+(\alpha)\rangle\langle\psi^+(\alpha)| - |\psi^-(\alpha)\rangle\langle\psi^-(\alpha)|$ with $|\psi^\pm(\alpha)\rangle = (|0\rangle\pm e^{i\alpha}|1\rangle)$, i.e., $O(\alpha) = e^{i\alpha}|1\rangle\langle 0| + e^{-i\alpha}|0\rangle\langle 1|$. Notice that $\left(O(\alpha)\right)^2=I$ and $O(\alpha)$ is always off-diagonal in computational basis. Thus we can still simplify its expectation values under GHZ-diagonal states, especially depolarized GHZ states with $nd$ qubits:
\begin{align}
    \mathrm{Var}_{M(\alpha)}(\hat{\theta}_1) = \frac{1 - \left(F - \frac{1-F}{2^{nd}-1}\right)^2\cos^2\left[n\sqrt{d}(\theta_1+\sqrt{d}\alpha)\right]}{dn^2\left(F - \frac{1-F}{2^{nd}-1}\right)^2 \sin^2\left[n\sqrt{d}(\theta_1+\sqrt{d}\alpha)\right]}.
\end{align}
It is clear that the non-zero azimuthal angle $\alpha$ has the effect of modifying the undesired zero denominator when $\theta_1=0$. Then we can safely take $\theta_1=0$ and optimize $\alpha$:
\begin{align}\label{Eq:var_alpha}
    \left.\mathrm{Var}_{M(\alpha)}(\hat{\theta}_1)\right\vert_{\theta_1=0} = \frac{1 - \left(F - \frac{1-F}{2^{nd}-1}\right)^2\cos^2\left(nd\alpha\right)}{dn^2\left(F - \frac{1-F}{2^{nd}-1}\right)^2 \sin^2\left(nd\alpha\right)}.
\end{align}
The partial derivative with respect to $\alpha$ gives rise to:
\begin{align}
    \frac{\partial}{\partial\alpha}\left.\mathrm{Var}_{M(\alpha)}(\hat{\theta}_1)\right\vert_{\theta_1=0} = -\frac{2^{nd+1}(1-F)\left[2^{nd}(1+F) - 2\right]}{n\left(2^{nd}F - 1\right)^2}\frac{\cos(nd\alpha)}{\sin^3(nd\alpha)}.
\end{align}
It is thus obvious that the variance takes the minimum values when $\cos(nd\alpha_\mathrm{opt})=0$, i.e. $\alpha_\mathrm{opt} = \frac{2l+1}{2nd}\pi$ with $l\in\mathbb{Z}$. The minimum value is:
\begin{align}
    \left.\mathrm{Var}_{M(\alpha_\mathrm{opt})}(\hat{\theta}_1)\right\vert_{\theta_1=0} = \frac{1}{d\left(F - \frac{1-F}{2^{nd}-1}\right)^2n^2} = \frac{1}{\eta_{M(\alpha_\mathrm{opt})}n^2}.
\end{align}
Such periodicity of the optimal azimuthal angle echoes with previous numerical results~\cite{cao2023distributed}. We can visualize the parameter estimation variance with different $\alpha$ for $n=1$ in Fig.~\ref{fig:azimuthal_meas}. The ratio between $\eta_{M(\alpha_\mathrm{opt})}$ and $\eta_\mathrm{dp}(k=1)$ is:
\begin{align}
    \frac{\eta_{M(\alpha_\mathrm{opt})}}{\eta_\mathrm{dp}(k=1)} = F + \frac{1 - F}{2^{nd} - 1},
\end{align}
which quickly converges to $F$ for larger $n$ and $d$.

The fidelity threshold for $nd$-qubit depolarized GHZ state to be advantageous over the optimal local strategy when using the optimized azimuthal measurement is given by $\eta_{M(\alpha_\mathrm{opt})} = 1$:
\begin{align}
    F_{\mathrm{th},M(\alpha_\mathrm{opt})}(n) = \frac{2^{nd} + \sqrt{d} - 1}{2^{nd}\sqrt{d}}.
\end{align}
For $n=1$, we have the fidelity threshold for the initial depolarized GHZ state to demonstrate advantage when using the optimized azimuthal measurement, i.e. $\eta_{M(\alpha_\mathrm{opt})}>1$:
\begin{align}
    F_{\mathrm{th},M(\alpha_\mathrm{opt})}(1) = \frac{2^d + \sqrt{d} - 1}{2^d\sqrt{d}},
\end{align}
which can be easily shown to be monotonically decreasing for the distributed regime $d\geq 2$ of our interest. 

We can further characterize the increase of estimation variance when the azimuthal angle of the measurement deviates from the optimal value. Firstly, we can explicitly characterize the behavior of estimation variance of the local azimuthal measurement when the azimuthal angle has small deviation from the optimal angle: $\alpha = \alpha_\mathrm{opt} + \delta$
\begin{align}
    \left.\mathrm{Var}_{M(\alpha_\mathrm{opt} + \delta)}(\hat{\theta}_1)\right\vert_{\theta_1=0} = \frac{1}{d\left(F - \frac{1-F}{2^{nd}-1}\right)^2n^2} + \frac{d(1-F)2^{nd}\left[2^{nd}(F+1)-2\right]}{\left(2^{nd}F-1\right)^2}\delta^2 + O(\delta^4).
\end{align}
Then we examine the range of azimuthal angle within which we can achieve better estimation variance than the optimal local sensing strategy. Specifically, we let the estimation variance of azimuthal local measurement in Eqn.~\ref{Eq:var_alpha} equal to the estimation variance of the optimal local sensing strategy, i.e.
\begin{align}
    \frac{1 - \left(F - \frac{1-F}{2^{nd}-1}\right)^2\cos^2\left(nd\alpha\right)}{dn^2\left(F - \frac{1-F}{2^{nd}-1}\right)^2 \sin^2\left(nd\alpha\right)} = \frac{1}{n^2} \implies \frac{1 - \left(F - \frac{1-F}{2^{nd}-1}\right)^2\cos^2\left(nd\alpha\right)}{d\left(F - \frac{1-F}{2^{nd}-1}\right)^2 \sin^2\left(nd\alpha\right)} = 1,
\end{align}
which can be simplified as
\begin{align}
    \sin^2\left(nd\alpha\right) = \frac{1 - \left(F-\frac{1-F}{2^{nd}-1}\right)^2}{(d-1)\left(F-\frac{1-F}{2^{nd}-1}\right)^2}.
\end{align}
The solution to the above equation determines the azimuthal angles at which the estimation variance of local azimuthal measurement equals the estimation variance of the optimal local sensing strategy. As the estimation variance of local azimuthal measurement has a period of $T=\pi/nd$, the range of the azimuthal angle for DQS quantum advantage is
\begin{align}
    W_\alpha = \frac{\pi}{nd} - \frac{2}{nd}\arcsin\sqrt{\frac{1 - \left(F-\frac{1-F}{2^{nd}-1}\right)^2}{(d-1)\left(F-\frac{1-F}{2^{nd}-1}\right)^2}}.
\end{align}
This range of allowed azimuthal angle can be considered as an indicator of the robustness of DQS quantum advantage to local quantum control for azimuthal measurement. To characterize the properties of $W_\alpha$, we further consider the ratio between the range which allows quantum advantage $W_\alpha$ and the period $T$, $R_\alpha=W_\alpha/T$,
\begin{align}
    R_\alpha(F,n,d) = 1 - \frac{2}{\pi}\arcsin\sqrt{\frac{1 - \left(F-\frac{1-F}{2^{nd}-1}\right)^2}{(d-1)\left(F-\frac{1-F}{2^{nd}-1}\right)^2}}.
\end{align}
We have the following basic properties of $R_\alpha(F,n,d)$.
\begin{proposition}
    $R_\alpha(F,n,d)$ increases monotonically as the fidelity $F$, the number of local quantum sensors $n$, and the number of sensor nodes $d$ increase.
\end{proposition}
\begin{proof}
    The properties can be proved through explicit evaluation of the partial derivatives. We first take the partial derivative w.r.t. the fidelity $F$
    \begin{align}
        \frac{\partial}{\partial F}R_\alpha = \frac{2^{\frac{nd}{2}+1}\left(2^{nd}-1\right)^2}{\pi\left(2^{nd}F-1\right)\sqrt{(1-F)\left[2^{nd}(1+F)-2\right]\left[d\left(2^{nd}F-1\right)^2-\left(2^{nd}-1\right)^2\right]}} \geq 0.
    \end{align}
    Then we take the partial derivative w.r.t. the number of local quantum sensors $n$
    \begin{align}
        \frac{\partial}{\partial n}R_\alpha = \frac{\ln(2)(1-F)2^{\frac{nd}{2}+1}\left(2^{nd}-1\right)d}{\pi\left(2^{nd}F-1\right)\sqrt{(1-F)\left[2^{nd}(1+F)-2\right]\left[d\left(2^{nd}F-1\right)^2-\left(2^{nd}-1\right)^2\right]}} \geq 0.
    \end{align}
    Lastly we take the partial derivative w.r.t. the number of sensor nodes $d$
    \begin{align}
        \frac{\partial}{\partial n}R_\alpha = \frac{2^{\frac{nd}{2}} \left[4^{nd}F(1-F^2) - 2^{nd}(1-F)(1+3F+(d-1)n\ln(4)) - (1-F)((d-1)n\ln(4)-2)\right]}{\pi(d-1)\left(2^{nd}F-1\right)\sqrt{(1-F)\left[2^{nd}(1+F)-2\right]}\left[d\left(2^{nd}F-1\right)^2-\left(2^{nd}-1\right)^2\right]} \geq 0.
    \end{align}
    For the above inequality we consider the last two terms in the numerator
    \begin{align}
         & - 2^{nd}(1-F)(1+3F+(d-1)n\ln(4)) - (1-F)((d-1)n\ln(4)-2)\nonumber\\
         =& -(1-F)\left[\left(2^{nd}(1+3F)-2\right) - \left(2^{nd}-1\right)(d-1)n\ln(4)\right]\nonumber\\
         \geq& -(1-F)\left[4\times 2^{nd} - \left(2^{nd}-1\right)\right] = -(1-F)\left(3\times 2^{nd}+1\right),
    \end{align}
    where for the inequality we have used $F\leq 1$, $d\geq 2$ in DQS setup, and $n\geq 1$. Then we evaluate the lower bound of the entire numerator
    \begin{align}
        & 4^{nd}F(1-F^2) - 2^{nd}(1-F)(1+3F+(d-1)n\ln(4)) - (1-F)((d-1)n\ln(4)-2)\nonumber\\
        \geq & 4^{nd}F(1-F^2) - (1-F)\left(3\times 2^{nd}+1\right)\nonumber\\
        =& (1-F)\left[4^{nd}F(1+F) - \left(3\times 2^{nd}+1\right)\right]\nonumber\\
        \geq& (1-F)\left[4^{nd}\frac{1}{\sqrt{d}}\left(1+\frac{1}{\sqrt{d}}\right) - 4\times 2^{nd}\right] \geq 0,
    \end{align}
    where for the second inequality we have used the fidelity threshold for DQS advantage using the optimized azimuthal measurement: We know that if fidelity is below the threshold we have $R_\alpha=0$.
\end{proof}
Although $R_\alpha(F,n,d)$ is monotonically increasing, it quickly converges to a value independent of $n$:
\begin{align}
    R_\alpha(F,n,d) \to 1 - \frac{2}{\pi}\arcsin\sqrt{\frac{1 - F^2}{(d-1)F^2}}.
\end{align}
Then we are clear about the behavior of $W_\alpha$ as $n$ and $d$ increases for larger relative quantum advantage in DQS 
\begin{align}
    W_\alpha \to \frac{\pi}{nd}\left[1 - \frac{2}{\pi}\arcsin\sqrt{\frac{1 - F^2}{(d-1)F^2}}\right].
\end{align}
Such scaling clearly demonstrates that while when $n$ and $d$ increase we can in principle achieve higher quantum advantage, the requirement on control accuracy also becomes stricter.

We re-emphasize that our optimization of the azimuthal measurement is not a global optimization, as it is locally separable when there are multiple sensor qubits per node. In principle we could utilize entangling measurement locally. However, the fact that when $n=1$ the optimized azimuthal measurement could not saturate the QCRB implies that the QCRB for this distributed quantum sensing problem is not achievable with local measurement. It is known that the QCRB can be achieved by projective measurement, and when $n=1$ the local projective measurement should be a tensor product of single-qubit projective measurements. Moreover, the assumed depolarization noise on the GHZ state guarantees symmetry among all qubits, so it suffices to consider identical projective measurement per qubit. Note that under our problem setup, z-direction projection is unable to extract the information of the parameter to estimate. Therefore, our optimization of azimuthal projective measurement should have covered the optimal local measurement, while as seen in the results they could not achieve the QCRB. On the other hand, suppose we insist on applying the global entangling measurement, it will need distribution of additional entangled state by the quantum network as resource to implement the measurement. Consequently, the performance of parameter estimation will be further limited by the fidelity of resource state for performing entangling measurement, and thus the fidelity of distributed entangled state by the network must be high. However, as we have seen that when the network is able to distributed high fidelity entangled state, performing local optimized azimuthal measurement can already achieve fairly low estimation variance, which is only $\sim 1/F$ times the ultimate achievable variance by the globally optimal measurement.

\subsection{Bell state fidelity requirement estimation}
Moreover, we can estimate the fidelity requirement of bipartite entanglement (Bell pair) distribution network to achieve quantum advantage in the local parameter average estimation task, by taking the $(d-1)$-th root of the fidelity threshold. This estimation comes from the assumption that we need $(d-1)$ bipartite entangled states between the sensor nodes to assemble the desired $d$-qubit GHZ state, and the approximation that the final GHZ state has fidelity equal to the product of all Bell states' fidelities. Specifically, for $n=1$ we have:
\begin{align}
    &F_{\mathrm{th},M(\alpha_\mathrm{opt})}^{\mathrm{Bell}} = \left(\frac{2^d - 1 + \sqrt{d}}{2^d\sqrt{d}}\right)^{\frac{1}{d-1}},\\
    &F_\mathrm{th,opt}^{\mathrm{Bell}} = \left(2^{-d} + \frac{(2^d-1)\left(2^d-2+\sqrt{(2^d-2)^2+2^{d+3}d}\right)}{2^{2d+1}d}\right)^{\frac{1}{d-1}},
\end{align}
where we use the superscript ``Bell'' to emphasize that the above fidelity thresholds are for Bell states distributed by the quantum network. The two thresholds are visualized in Fig.~\ref{fig:Bell_thresholds}(a). It can be seen that the fidelity threshold for the global optimal measurement slightly decreases when the number of sensor nodes $d$ is small. More specifically, we have:
\begin{align}
    &F_\mathrm{th,opt}^{\mathrm{Bell}}(2) \approx 0.730,~ F_\mathrm{th,opt}^{\mathrm{Bell}}(3) \approx 0.714,~ F_\mathrm{th,opt}^{\mathrm{Bell}}(4) \approx 0.711,\nonumber\\
    &F_\mathrm{th,opt}^{\mathrm{Bell}}(5) \approx 0.716,~ F_\mathrm{th,opt}^{\mathrm{Bell}}(6) \approx 0.726,~ F_\mathrm{th,opt}^{\mathrm{Bell}}(7) \approx 0.738.
\end{align}
Nevertheless, in general both the thresholds increase monotonically as the number of sensor nodes increases. Moreover, we can straightforwardly evaluate the asymptotic scaling of both thresholds:
\begin{align}
    F_{\mathrm{th},M(\alpha_\mathrm{opt})}^{\mathrm{Bell}} \sim d^{\frac{1}{2(1-d)}},~ F_\mathrm{th,opt}^{\mathrm{Bell}} \sim d^{\frac{1}{1-d}}.
\end{align}
Very interestingly, the Bell state fidelity threshold for the global optimal measurement is the square of the threshold for the optimized azimuthal measurement in the asymptotic limit of $d\rightarrow\infty$, i.e. $F_\mathrm{th,opt}^{\mathrm{Bell}}\sim\left(F_{\mathrm{th},M(\alpha_\mathrm{opt})}^{\mathrm{Bell}}\right)^2$. It is then easily seen that both fidelity thresholds converge to one in the large $d$ limit. To visualize the asymptotic scaling, we further plot the thresholds of infidelity $\epsilon=1-F$ in log-log coordinate in Fig.~\ref{fig:Bell_thresholds}(b). It is clear that $\epsilon_\mathrm{th,opt}^{\mathrm{Bell}}/2$ is almost equal to $\epsilon_{\mathrm{th},M(\alpha_\mathrm{opt})}^{\mathrm{Bell}}$ when $d$ becomes large, which verifies the quadratic relation of the asymptotic scalings.

\begin{figure*}[t]
    \centering
    \includegraphics[width=0.8\linewidth]{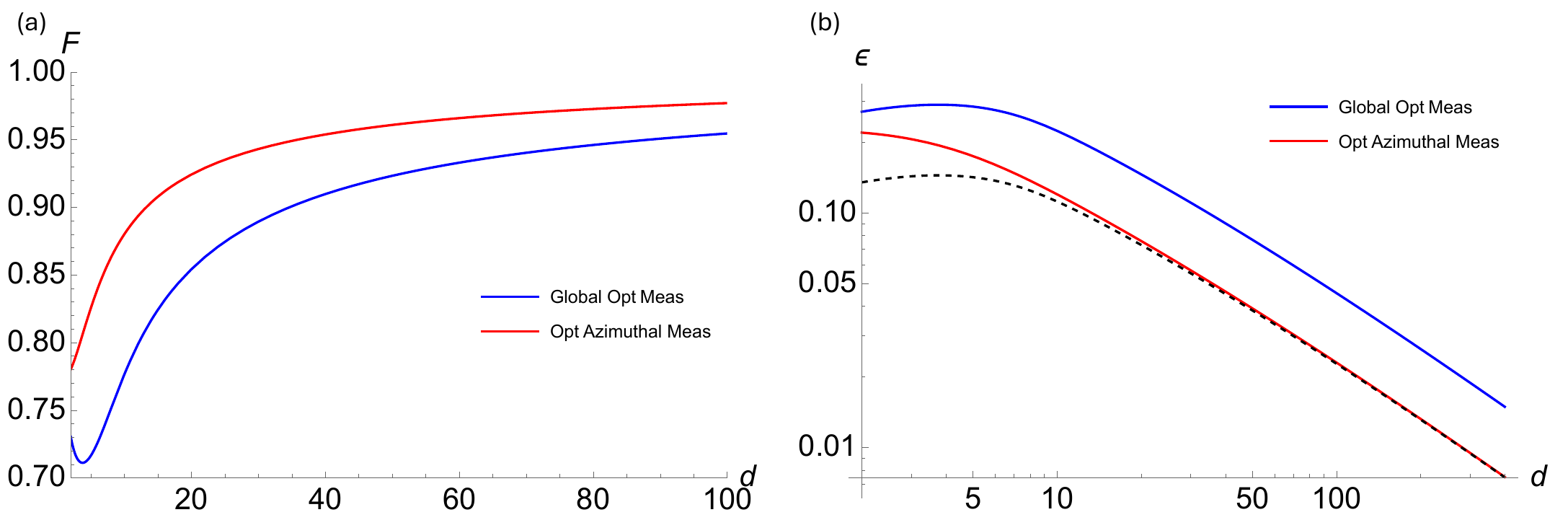}
    \caption{Estimation of Bell state fidelity thresholds for demonstrating quantum advantage in distributed quantum sensing, where the initial $d$-qubit GHZ state is assembled from $(d-1)$ Bell states distributed by the quantum network. (a) The Bell state fidelity threshold when the global optimal measurement is allowed is depicted by the blue curve, while the threshold when only the optimized azimuthal measurement can be performed is illustrated by the red curve. (b) The infidelity $\epsilon=1-F$ thresholds were plotted in the log-log coordinate for additional visualization insight. The blue curve again corresponds to the global optimal measurement and the red curve to the optimized azimuthal measurement. For reference, we plot the half of the infidelity threshold for the global optimal measurement in black dashed line, to better demonstrate the relationship between the asymptotic scalings of both thresholds.}
    \label{fig:Bell_thresholds}
\end{figure*}

\section{Details for encoding under dephasing}\label{app:dephasing}
We first add the detailed derivation of the effective initial state. Then we discuss the estimation of the frequency instead of the accumulated phase. Last but not least we also consider the specific local azimuthal measurement and its optimization.

\subsection{Derivation of the effective initial state}
In the $nd$-qubit extra Pauli channel for the effective initial state the probability for each Pauli string to be applied to the initial state is determined by the number of Pauli Z's in it $\#Z$. More explicitly we have $p_{\#Z} =q^{nd-\#Z}(1-q)^{\#Z}$. For our ideal objective state, the standard $nd$-qubit GHZ state $|\mathrm{GHZ}_{nd}\rangle$, it is known that its stabilizer is generated by the following $nd$ Pauli strings 
\begin{align}
    S =& \{\sigma_x^{(1)}\otimes \sigma_x^{(2)}\otimes \dots\otimes \sigma_x^{(nd)}, \sigma_z^{(1)}\otimes \sigma_z^{(2)}, \sigma_z^{(2)}\otimes \sigma_z^{(3)}, \dots, \sigma_z^{(nd-1)}\otimes \sigma_z^{(nd)}\},
\end{align}
from which one can observe that any Pauli string that contains an even number of Pauli Z's will stabilize $|\mathrm{GHZ}_{nd}\rangle$. On the other hand, the remaining Pauli strings in the $nd$-qubit Pauli channel from dephasing still only have identity operators and the Pauli Z's, so they will only result in a phase flip in $|\mathrm{GHZ}_{nd}\rangle$.

Now recall that we consider $nd$-qubit depolarized GHZ state distributed over the quantum sensor network 
\begin{align}
    \rho_0 = V(nd)|\mathrm{GHZ}_{nd}\rangle\langle\mathrm{GHZ}_{nd}| + [1-V(nd)]I_{nd}/2^{nd},
\end{align}
where $V(nd)$ again denotes visibility, and according to our state preparation error model we have GHZ fidelity $F(nd) = V(nd) + [1-V(nd)]/2^{nd} = Fk^{n-1}$. The above explicit decomposition into a pure GHZ projector and a maximally mixed state is convenient in that the maximally mixed state is invariant under the Pauli channel according to the unitality. For the GHZ projector, the probability for it to remain unchanged is thus the total probability of its stabilizing Pauli strings in the noise channel~\cite{zang2025enhancing}
\begin{align}
    \tilde{q} = \sum_{\mathrm{even}~\#Z}\binom{nd}{\#Z}p_{\#Z} = \sum_{i=0}^{\lfloor nd/2\rfloor}\binom{nd}{2i}q^{nd-2i}(1-q)^{2i} = \frac{1 + (2q-1)^{nd}}{2},
\end{align}
for positive integer $n$. The resulting state is then
\begin{align}\label{eqn:eff_probe_dephasing}
    \rho_0' = \tilde{q}V\mathrm{Proj}_{\mathrm{GHZ}^+} + (1-\tilde{q})V\mathrm{Proj}_{\mathrm{GHZ}^-} + (1-V)I_{nd}/2^{nd},
\end{align}
where $\mathrm{Proj}_{\mathrm{GHZ}^+} = |\mathrm{GHZ}_{nd}\rangle\langle\mathrm{GHZ}_{nd}|$ and $\mathrm{Proj}_{\mathrm{GHZ}^-} = |\mathrm{GHZ}_{nd}^-\rangle\langle\mathrm{GHZ}_{nd}^-|$ with $|\mathrm{GHZ}_{nd}^-\rangle=(|00\dots 0\rangle-|11\dots 1\rangle)/\sqrt{2}$ is the phase-flipped $|\mathrm{GHZ}_{nd}\rangle$.

Then according to the derivation of the QFIM, we only need to care about the eigenvalues of $\rho_0'$ which correspond to $|\mathrm{GHZ}_{nd}\rangle$ and $|\mathrm{GHZ}_{nd}^-\rangle$, because the effective dephasing channel does not change other eigenvalues which are all equal according to the depolarizing error assumption. Therefore, we have the QFI for $\theta_1$ as
\begin{align}
    \mathcal{F}(\theta_1) =& d \frac{\left[\left(\tilde{q}V + \frac{1 - V}{2^{nd}}\right) - \left((1-\tilde{q})V + \frac{1 - V}{2^{nd}}\right)\right]^2}{\left(\tilde{q}V + \frac{1 - V}{2^{nd}}\right) + \left((1-\tilde{q})V + \frac{1 - V}{2^{nd}}\right)}n^2 = d\frac{\left[F(2^dk)^n - k\right]^2 (2q-1)^{2nd}}{k(2^{nd} - 1)\left[Fk^n(2^{nd}-2) + k\right]} n^2,
\end{align}
which reduces to the QFI without dephasing if $q=1$. Note that in the above the parameter $q$ is time-dependent, and in the following we will consider this fact explicitly.

\subsection{Estimation of the average frequency}
From Eqn.~\ref{eqn:qfi_dephasing_phase}, we have the explicit expression of $\mathcal{F}^{(\tilde{\theta}_1)}(t)$
\begin{align}
    \mathcal{F}^{(\tilde{\theta}_1)}(t) = \frac{d\left[F(2^dk)^n - k\right]^2 e^{-2nd\gamma t}}{k(2^{nd} - 1)\left[Fk^n(2^{nd}-2) + k\right]} n^2 t^2 .
\end{align}
In the following we may omit the superscript when there is no abuse of notation. The requirement for potential quantum advantage at fixed encoding duration $\mathcal{F}_\mathrm{net}(t) > \mathcal{F}_\mathrm{local}(t)$ gives
\begin{align}\label{eqn:qfi_dephase_vs_noiseless_local}
    &\frac{d\left[F(2^dk)^n - k\right]^2}{k(2^{nd} - 1)\left[Fk^n(2^{nd}-2) + k\right]} \approx Fk^{n-1}d > e^{2nd\gamma t}\nonumber\\
    \Rightarrow& \gamma t < \frac{1}{2nd}\ln\left(\frac{d\left[F(2^dk)^n - k\right]^2}{k(2^{nd} - 1)\left[Fk^n(2^{nd}-2) + k\right]}\right) \approx \frac{1}{2nd}\ln\left(Fk^{n-1}d\right),
\end{align}
where the approximation comes from assumptions $k\lesssim 1$, $2^d\gg 1$ and $2^{nd}\gg 2$. The partial derivatives with respect to $n,d,k,F$ are
\begin{align}
    \frac{\partial}{\partial n}\left[\frac{1}{2nd}\ln\left(Fk^{n-1}d\right)\right] =& \frac{n\ln k - \ln\left(Fk^{n-1}d\right)}{2dn^2},\\
    \frac{\partial}{\partial d}\left[\frac{1}{2nd}\ln\left(Fk^{n-1}d\right)\right] =& \frac{1 - \ln\left(Fk^{n-1}d\right)}{2d^2n},\\
    \frac{\partial}{\partial k}\left[\frac{1}{2nd}\ln\left(Fk^{n-1}d\right)\right] =& \frac{n - 1}{2dkn},\\
    \frac{\partial}{\partial F}\left[\frac{1}{2nd}\ln\left(Fk^{n-1}d\right)\right] =& \frac{1}{2dFn}.
\end{align}

For the sequential scheme of quantum sensing, we want to maximize
\begin{align}
    \frac{\mathcal{F}_\mathrm{net}(t)}{t} = \frac{d\left[F(2^dk)^n - k\right]^2 e^{-2nd\gamma t}}{k(2^{nd} - 1)\left[Fk^n(2^{nd}-2) + k\right]} n^2 t.
\end{align}
The optimal time for the networked sensing under noise $t^*=1/(2nd\gamma)$ leads to
\begin{align}
    \frac{\mathcal{F}_\mathrm{net}(t^*)}{t^*} = \frac{\left[F(2^dk)^n - k\right]^2n}{2 e \gamma k(2^{nd} - 1)\left[Fk^n(2^{nd}-2) + k\right]},
\end{align}
where $e$ is the base of the natural logarithm. For the noiseless local baseline we simply have
\begin{align}
    \frac{\mathcal{F}_\mathrm{local}(t_\mathrm{th})}{t_\mathrm{th}} = n^2t_\mathrm{th}.
\end{align}
When $t_\mathrm{th}\geq t^*$, we have the necessary condition for potential quantum advantage
\begin{align}
    \frac{\mathcal{F}_\mathrm{net}(t^*)}{t^*} > \frac{\mathcal{F}_\mathrm{local}(t_\mathrm{th})}{t_\mathrm{th}}  \Rightarrow \frac{\left[F(2^dk)^n - k\right]^2}{2 n e \gamma k(2^{nd} - 1)\left[Fk^n(2^{nd}-2) + k\right]t_\mathrm{th}} \approx \frac{F k^{n-1}}{2 n e \gamma t_\mathrm{th}} > 1,
\end{align}
where we have made the same approximation as above. This condition is effectively given a constraint on the time threshold $t_\mathrm{th}$, as 
\begin{align}
    \frac{1}{2 n d \gamma} = t^* < t_\mathrm{th} \lesssim \frac{F k^{n-1}}{2 n e \gamma},
\end{align}
which requires
\begin{align}
    \frac{1}{2 n d \gamma} < \frac{F k^{n-1}}{2 n e \gamma}.
\end{align}
When the above inequality is satisfied, the noisy quantum sensor network may possibly demonstrate advantage over the noiseless baseline for long sensing cycle time threshold $t_\mathrm{th}\geq t^*$. 

When $t_\mathrm{th}<t^*$, the necessary condition of quantum advantage is
\begin{align}
    \frac{d\left[F(2^dk)^n - k\right]^2}{k(2^{nd} - 1)\left[Fk^n(2^{nd}-2) + k\right]} > e^{2nd\gamma t_\mathrm{th}},
\end{align}
which is in the same form as Eqn.~\ref{eqn:qfi_dephase_vs_noiseless_local}. Therefore, if the threshold of sensing cycle time satisfies
\begin{align}
    \gamma t_\mathrm{th} < \frac{1}{2nd}\ln\left(\frac{d\left[F(2^dk)^n - k\right]^2}{k(2^{nd} - 1)\left[Fk^n(2^{nd}-2) + k\right]}\right) \approx \frac{1}{2nd}\ln\left(Fk^{n-1}d\right),
\end{align}
the noisy networked sensing can potentially demonstrate quantum advantage over the noiseless local sensing strategy.

\subsection{Local azimuthal measurement}
From error propagation, we have the estimation variance of the average phase $\tilde{\theta}_1$ using $M(\alpha)$
\begin{align}
    \mathrm{Var}_{M(\alpha)}(\hat{\tilde{\theta}}_1) =& \frac{1 - \left[\left(\tilde{q}V + \frac{1 - V}{2^{nd}}\right) - \left((1-\tilde{q})V + \frac{1 - V}{2^{nd}}\right)\right]^2\cos^2\left[n\sqrt{d}(\theta_1+\sqrt{d}\alpha)\right]}{dn^2\left[\left(\tilde{q}V + \frac{1 - V}{2^{nd}}\right) - \left((1-\tilde{q})V + \frac{1 - V}{2^{nd}}\right)\right]^2 \sin^2\left[n\sqrt{d}(\theta_1+\sqrt{d}\alpha)\right] t^2}\nonumber\\
    =& \frac{1 - \left[(2\tilde{q} - 1)V\right]^2\cos^2\left[n\sqrt{d}(\theta_1+\sqrt{d}\alpha)\right]}{dn^2\left[(2\tilde{q} - 1)V\right]^2 \sin^2\left[n\sqrt{d}(\theta_1+\sqrt{d}\alpha)\right] t^2}\nonumber\\
    =& \frac{1 - \left(F(nd) - \frac{1-F(nd)}{2^{nd}-1}\right)^2\cos^2\left[n\sqrt{d}(\theta_1+\sqrt{d}\alpha)\right]e^{-2nd\gamma t}}{dn^2 \left(F(nd) - \frac{1-F(nd)}{2^{nd}-1}\right)^2 \sin^2\left[n\sqrt{d}(\theta_1+\sqrt{d}\alpha)\right] t^2 e^{-2nd\gamma t}},
\end{align}
where $\tilde{\theta}_1 = \theta_1 / t$ and $F(nd) = Fk^{n-1}$. Then we take $\tilde{\theta}_1 = 0$ and optimize $\alpha$
\begin{align}
    \left.\mathrm{Var}_{M(\alpha)}(\hat{\tilde{\theta}}_1)\right\vert_{\tilde{\theta}_1 = 0} = \frac{1 - \left(F(nd) - \frac{1-F(nd)}{2^{nd}-1}\right)^2\cos^2(n d \alpha)e^{-2nd\gamma t}}{dn^2 \left(F(nd) - \frac{1-F(nd)}{2^{nd}-1}\right)^2 \sin^2(n d \alpha) t^2 e^{-2nd\gamma t}}.
\end{align}
The partial derivative with respect to $\alpha$ leads to
\begin{align}
    \frac{\partial}{\partial\alpha}\left.\mathrm{Var}_{M(\alpha)}(\hat{\tilde{\theta}}_1)\right\vert_{\tilde{\theta}_1=0} = \frac{2\left[\left(2^{nd}F(nd) - 1\right)^2 - \left(2^{nd} - 1\right)^2e^{2nd\gamma t}\right]}{n\left(2^{nd}F(nd) - 1\right)^2t^2}\frac{\cos(nd\alpha)}{\sin^3(nd\alpha)}.
\end{align}
The variance takes the minimum values when $\cos(nd\alpha_\mathrm{opt})=0$, i.e. $\alpha_\mathrm{opt} = \frac{2l+1}{2nd}\pi$ with $l\in\mathbb{Z}$. The minimum value is:
\begin{align}
    \left.\mathrm{Var}_{M(\alpha_\mathrm{opt})}(\hat{\tilde{\theta}}_1)\right\vert_{\tilde{\theta}_1=0} = \frac{k^2\left(2^{nd} - 1\right)^2}{d \left[F(2^dk)^n - k\right]^2 e^{-2nd\gamma t}} \frac{1}{n^2t^2}.
\end{align}
Then we can compare with the QFI that we derive in the above
\begin{align}
    \frac{\left.\mathrm{Var}_{M(\alpha_\mathrm{opt})}(\hat{\tilde{\theta}}_1)\right\vert_{\tilde{\theta}_1=0}}{1 / \mathcal{F}(t)} = \frac{d\left[F(2^dk)^n - k\right]^2 e^{-2nd\gamma t} n^2 t^2}{k(2^{nd} - 1)\left[Fk^n(2^{nd}-2) + k\right]} \frac{k^2\left(2^{nd} - 1\right)^2}{d \left[F(2^dk)^n - k\right]^2 e^{-2nd\gamma t} n^2 t^2} = \frac{2^{nd} - 1}{Fk^{n-1}(2^{nd}-2) + 1} \geq 1,
\end{align}
where the equality is taken if and only if $Fk^{n-1} = 1$, i.e. when initial probe state preparation is noiseless. The partial derivatives with respect to $n,d,F,k$ are
\begin{align}
    \frac{\partial}{\partial n}\left[\frac{2^{nd} - 1}{Fk^{n-1}(2^{nd}-2) + 1}\right] =& \frac{k\left[2^{nd}d(k - Fk^n) \ln 2 - \left(4^{nd} - 3\times 2^{nd} + 2\right)Fk^n\ln k\right]}{\left[Fk^n\left(2^{nd} - 2\right) + k\right]^2} ,\\
    \frac{\partial}{\partial d}\left[\frac{2^{nd} - 1}{Fk^{n-1}(2^{nd}-2) + 1}\right] =& \frac{2^{nd}nk(k - Fk^n) \ln 2}{\left[Fk^n\left(2^{nd} - 2\right) + k\right]^2} \approx \frac{n}{(k^2 2^d)^n} \frac{k(k - Fk^n) \ln 2}{F^2} ,\\
    \frac{\partial}{\partial k}\left[\frac{2^{nd} - 1}{Fk^{n-1}(2^{nd}-2) + 1}\right] =& - \frac{\left(4^{nd} - 3\times 2^{nd} + 2\right)k^n}{\left[Fk^n\left(2^{nd} - 2\right) + k\right]^2} (n-1)F,\\
    \frac{\partial}{\partial F}\left[\frac{2^{nd} - 1}{Fk^{n-1}(2^{nd}-2) + 1}\right] =& - \frac{\left(4^{nd} - 3\times 2^{nd} + 2\right)k^n}{\left[Fk^n\left(2^{nd} - 2\right) + k\right]^2} k .
\end{align}
From the above one can see that the saturation is almost not affected by the number of sensor nodes (number of local parameters) in the network, as the partial derivative with respect to $d$ vanishes for increasing $nd$. On the other hand, the relative difference between the QCRB and the estimation variance using the optimal local azimuthal measurement generally increases as the number of local sensors $n$ increases, due to accumulation of depolarizing error. Also, when $F\sim k$ the local entanglement generation quality $k$ has increasing impact on the saturation than the quantum network entanglement distribution fidelity $F$ as $n$ increases.

\section{Parameter estimation under entanglement distribution failure}\label{app:est_w_fail}
The probe state is generated through assembling bipartite entangled states, and in real quantum networks it is possible that some bipartite entangled states are not successfully generated within certain attempts. We consider a specific way of bipartite entangled state assembly to generate the GHZ state: A center node will share bipartite entanglement with other nodes. Therefore, the nodes which fail to establish entangled link with the center node will remain separable from other nodes, while all the nodes that successfully share entangled link with the center node will be entangled together. Let $\mathcal{N}$ s.t. $\vert\mathcal{N}\vert=d$ be the index set of all sensor nodes, and $\mathcal{C}$ be the index set for the nodes which remain isolated, thus the set difference $\mathcal{N}\setminus\mathcal{C}$ denotes the index set of nodes that will be entangled. The objective parameter to estimate $\theta_1=\sum_{i\in\mathcal{N}}x_i/\sqrt{d}$ can be rewritten as
\begin{align}
    \theta_1 = \frac{1}{\sqrt{d}}\sum_{i\in\mathcal{C}}x_i + \sqrt{\frac{\vert\mathcal{N}\setminus\mathcal{C}\vert}{d}}\left(\frac{1}{\sqrt{\vert\mathcal{N}\setminus\mathcal{C}\vert}}\sum_{j\in\mathcal{N}\setminus\mathcal{C}}x_j\right) = \frac{1}{\sqrt{d}}\sum_{i\in\mathcal{C}}x_i + \sqrt{\frac{\vert\mathcal{N}\setminus\mathcal{C}\vert}{d}}\theta'_1,
\end{align}
where the second term is proportional to the average of local parameters on all the nodes that are entangled, $\theta'_1$. When not all sensor nodes are entangled, we consider that the sensor network will use the following hybrid strategy: The isolated nodes use local probe state to estimate the local parameters $x_i, i\in\mathcal{C}$ individually, while the entangled nodes use a globally entangled probe state to estimate $\theta'_1$. Then according to propagation of errors, we have the variance of estimating $\theta$ by such a hybrid strategy as
\begin{align}\label{Eq:var_hybrid}
    \mathrm{Var}_\mathrm{hybrid}(\hat{\theta}_{1,\mathcal{C}}) = \frac{1}{d}\sum_{i\in\mathcal{C}}\mathrm{Var}(\hat{x}_i) + \frac{\vert\mathcal{N}\setminus\mathcal{C}\vert}{d}\mathrm{Var}(\hat{\theta}'_1),
\end{align}
where the subscript $\mathcal{C}$ for estimator $\hat{\theta}_{1,\mathcal{C}}$ emphasizes that the estimator is uniquely determined by \textit{configuration} $\mathcal{C}$.

Due to the probabilistic nature of remote entanglement distribution, for each attempt of probe state generation the configuration $\mathcal{C}$ can be different. Nevertheless, we know exactly what configuration happens each time due to the heralded nature of quantum network protocols. Let $\mathfrak{C}$ denotes the collection of all possible configurations $\mathcal{C}$. For the scenario with $d$ sensor nodes, we have that $\vert\mathfrak{C}\vert = \sum_{n=0}^{d-2}\binom{d}{n}+1 = 2^d-d$ without over-counting. Suppose we repeat the quantum sensing cycle for $N$ times, and each configuration $\mathcal{C}$ occurs with probability $p_\mathcal{C}$. We may expect that there are $N_\mathcal{C}=p_\mathcal{C}N$ data points corresponding to configuration $\mathcal{C}$, then we have $\mathrm{Var}(\hat{\theta}_{1,\mathcal{C}})\propto 1/N_\mathcal{C}$. Assuming unbiased estimators, the normalized linear combination of $\hat{\theta}_{1,\mathcal{C}}$ still has the mean of $\theta_1$. Then our objective is to minimize the variance of the normalized linear combination $\hat{\theta}_1 = \sum_{\mathcal{C}\in\mathfrak{C}}w_\mathcal{C}\hat{\theta}_{1,\mathcal{C}}$, s.t. $\sum_{\mathcal{C}}w_\mathcal{C}=1$. We may further assume that $\hat{\theta}_{1,\mathcal{C}}$ are uncorrelated, which according to the Bienaym\'{e}'s identity gives $\mathrm{Var}(\hat{\theta}_1) =\sum_{\mathcal{C}\in\mathfrak{C}}w_\mathcal{C}^2\mathrm{Var}(\hat{\theta}_{1,\mathcal{C}})$.
Then using Lagrange multiplier the optimal weighting for the minimum variance is
\begin{align}
    w_\mathcal{C}^* = \frac{1}{\mathrm{Var}(\hat{\theta}_{1,\mathcal{C}})} \left[\sum_{\mathcal{C}\in\mathfrak{C}}\frac{1}{\mathrm{Var}(\hat{\theta}_{1,\mathcal{C}})}\right]^{-1},
\end{align}
which is the so-called inverse-variance weighting that gives the minimum variance
\begin{align}
    \mathrm{Var}_\mathrm{min}(\hat{\theta}_1) =\left[\sum_{\mathcal{C}\in\mathfrak{C}}\frac{1}{\mathrm{Var}(\hat{\theta}_{1,\mathcal{C}})}\right]^{-1}.
\end{align}

\section{Quantum network simulation}\label{app:qn_sim}

\begin{table}[t]
    \caption{\label{tab:sim_res}Simulation results for the three scenarios.}
    \begin{ruledtabular}
    \begin{tabular}{ccccc}
    & $p$ & $\eta$ & $\Tilde{\eta}$ & $F$ \\
    \colrule
    \textrm{Scenario 1} & $\approx0.02$ & $\approx0.95<1$ & $\approx0.02<1$ & 0.591\\
    \textrm{Scenario 2} & $\approx0.72$ & $\approx1.57>1$ & $\approx1.13>1$ & 0.732\\
    \textrm{Scenario 3} & $\approx0.81$ & $\approx2.19>1$ & $\approx1.77>1$ & 0.854\\
    \end{tabular}
    \end{ruledtabular}
\end{table}

In addition to the analytics, we simulate GHZ state distribution given realistic first-generation quantum repeater~\cite{briegel1998quantum,munro2015inside,muralidharan2016optimal,azuma2023quantum} protocol stacks with an open-source, customizable quantum network simulator, SeQUeNCe~\cite{wu2021sequence}. We consider a simple 3-node network with linear topology, where a center node directly connects the other two end nodes through optical fibers, with a Bell state measurement station in the middle of each fiber link. The network topology is visualized in Fig.~\ref{fig:schematics} (a). First, bipartite entanglement links are established between the center node and the other nodes. Then LOCC, such as gate teleportation~\cite{gottesman1999demonstrating,eisert2000optimal,jiang2007distributed} and graph state fusion~\cite{browne2005resource,pirker2018modular}, are performed to generate a GHZ state across all network nodes.

In our simulations we consider three scenarios, and thus three sets of system parameter values: $\{0.01\mathrm{s}, 0.1\mathrm{s}, 1\mathrm{s}\}$ for memory coherence times, $\{0.05, 0.1, 0.5\}$ for memory efficiency, and $\{0.8, 0.85, 0.9\}$ for raw entanglement fidelity. The first simulation scenario uses the first value in each of the above sets, and so on. Meanwhile, we fix other system parameters, especially: a 1s time interval allowed for entanglement distribution, a 10km distance between the center node and end nodes, 10 memories per end node and 20 memories on the center node so that entanglement purification~\cite{bennett1996purification,deutsch1996quantum} is possible. We emphasize that the parameter values are not chosen from a specific reported experiment, but are selected according to the general vision of the state of the art and the potential future development of various candidate platforms for quantum networking, including solid state systems~\cite{knaut2024entanglement,stolk2024metropolitan,ruskuc2024scalable}, atomic systems~\cite{liu2024creation,krutyanskiy2023entanglement}, and superconducting systems~\cite{sahu2023entangling}.

We characterize the performance of probe state distribution by three figures of merit, namely the success probability of distributing the 3-qubit GHZ state, $p$, the relative advantage, $\eta$, and the normalized relative advantage, $\Tilde{\eta} = p\cdot\eta$, which takes into account failure in entanglement distribution. For each scenario we repeat the simulation 1000 times, and $\eta$ is calculated based on the average of density matrices of successfully distributed GHZ states under ensemble interpretation. The results are collected in Table~\ref{tab:sim_res}, from which it is clear that the most modest parameter choice does not permit quantum advantage, while when hardware performance improves quantum advantage becomes possible without changing the realistic quantum network protocol stack. We also note the average fidelity $F$ of the successfully generated states for each scenario; as expected, the fidelity increases as the network parameters improve.

In the following we elaborate on the simulation details.

\begin{figure}[t]
    \centering
    \includegraphics[width=\linewidth]{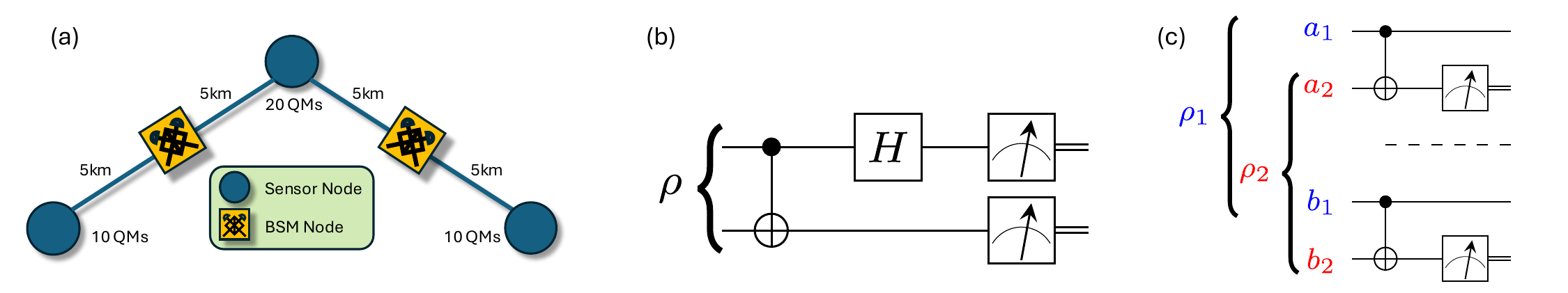}
    \caption{Visualization of network simulation schematics and protocols: (a) topology of the simulated 3-node quantum network, (b) standard Bell state measurement (BSM), and (c) entanglement purification protocol based on bilocal CNOT, specifically the BBPSSW/DEJMPS protocol.}
    \label{fig:schematics}
\end{figure}

\subsection{Bipartite entanglement (Bell pairs) distribution}
The distribution of bipartite entanglement involves remote entanglement generation between nearest network nodes that are directly connected by a physical implementation of quantum channels (e.g., optical fibers or free-space optics), entanglement swapping which extends the generated entangled states to more distant node pairs, and potentially entanglement purification which is expected to improve the quality (fidelity) of distributed entangled pairs.

For entanglement generation, in SeQUeNCe we implement a new abstract model of single-heralded entanglement generation protocol based on meet-in-the-middle photonic Bell state measurement. The underlying processes assumed for a single attempt of entanglement generation includes:
\begin{enumerate}
    \item Local memory-photon entanglement on a network node is generated, and the photon is transmitted to a middle interference center via a lossy optical fiber to perform heralded measurement (Bell state measurement/BSM).
    \item Two photons being transmitted from both nodes are expected to simultaneously arrive at the middle node at certain time determined by fiber lengths.
    \item However, due to the existence of losses in optical fibers, it is possible that one or both photons are not successfully transmitted. We assume that the heralded measurement can only be successful if both photons arrive.
    \item Moreover, we assume that the underlying implementation of heralded measurement consists of linear optics only, which fundamentally upper bounds the success probability of BSM to 1/2~\cite{lutkenhaus1999bell}.
\end{enumerate}
According to the above description, the overall success probability of an entanglement generation attempt is $p_{t,l}p_{t,r}p_m$ where $p_{t,l(r)}$ is the probability for left (right) photon to arrive at the middle station and $p_m$ is the success probability of heralded measurement when both photons arrive, which in simulation is a positive, tunable parameter below 1/2. Additionally, the successfully generated entangled state, is assumed to be a Bell diagonal state (BDS). 

Entanglement swapping involves three nodes, one middle node which will perform projective BSM, and two end nodes which will receive the BSM results from the middle node and perform local corrections accordingly. We consider the standard quantum circuit of BSM as shown in Fig.~\ref{fig:schematics}(b), where measurements are both in computational basis. We also consider that different physical implementation of entanglement swapping could lead to different success probability, and allow the success probability of each entanglement swapping to be tunable. Bipartite entanglement purification involves two nodes which share multiple entangled state. In SeQUeNCe, the default protocol is the BBPSSW/DEJMPS protocol based on bilocal CNOT, i.e. both nodes will perform CNOT on two qubits from two entangled states they hold, and perform single qubit computational basis measurement on one qubit per node, whose circuit is shown in Fig.~\ref{fig:schematics}(c). The measurement outcomes from both nodes are classically communicated to each other to determine if the purification attempt is successful.

To increase the practicality in our simulation, we consider imperfections in single-qubit measurements and two-qubit gates (especially CNOT). However, we do not consider errors in single-qubit gates, because experimentally multi-qubit gates are much more noisy than single-qubit ones, with infidelity about one order of magnitude higher. As we assume entangled states are in BDS form, we perform offline analytical derivation of imperfect entanglement swapping results in Sec.~\ref{sec:analytic_swap}, and imperfect entanglement purification results in Sec.~\ref{sec:analytic_purify}. These allow us to avoid explicit tracking of density matrix under quantum operations during simulation of entanglement distribution. We assume a homogeneous model of gate and measurement imperfections on each node: On one specific node, the error model of multi-qubit gate is unchanged when applying to different qubits, and so is the error model of single-qubit measurement. 

When multiple entangled states (quantum memories) are available, the space of strategies for executing entanglement protocols (e.g., generation, swapping and purification) is large. Currently in SeQUeNCe, the strategy is determined by the priorities of each protocol in the stack. For concreteness, our implementation gives the highest priority to entanglement purification (i.e., purification will be performed ASAP), whenever multiple entangled states are available on a single link (between two network nodes). Entanglement swapping has the second highest priority, and it will be performed immediately when there is one entangled state established between a center node and each of its left and right neighbors (not necessarily nearest neighbors).
We also emphasize that when taking into account quantum memory decoherence, and the non-universality of entanglement purification~\cite{zang2025no}, the optimization of quantum protocol strategies/policies~\cite{khatri2021policies} will be very hard due to the large policy space, so this is beyond the scope of this work where simulation is for the demonstration of principle for DQS probe state preparation.

\subsection{Time-dependent quantum memory decoherence}
Quantum memories inevitably undergo decoherence throughout the time after they are initialized. Therefore, besides quantum operation imperfections, we also implement time-dependent memory decoherence as analytically modeled by continuous time Pauli channels~\cite{zang2025entanglement} which is naturally compatible with the BDS assumption. 

Quantum states distributed over the network only change when quantum protocols are operated, such as entanglement swapping and entanglement purification. Therefore, memory idling decoherence effect only needs to be applied to the quantum state prior to the quantum protocols, according to the duration of idling which can be calculated as the difference between the current time and the last time point when the quantum state is updated. Then the quantum memory decoherence over time smoothly fits in the discrete event simulation framework. 

Due to the decoherence over time, one cannot use a certain quantum memory forever, and resetting is needed after a certain amount of time, which is the cutoff time for the storage of quantum states. In our simulation, we set the cutoff time to be proportional to quantum memory coherence time, while the coefficient is a tunable parameter. Specifically, for entangled states involving more than one quantum memories, the time for all involved memories to be reset is determined by the smallest reset time among the memories.

For entanglement generation, the quantum memories which store the successfully generated entangled states are initialized before the nodes receive the heralding signals. To account for the idling errors of quantum memories before the arrival of measurement results, we assume that the successfully generated entangled state starts from a certain BDS, specified by the entanglement generation protocol (parameters including the initial fidelity and relative strength of three Pauli error components). Then the entangled state will decohere under independent quantum memory decoherence channels which are dependent on idling time, and the decoherence time for each memory is determined by the time when the memory is initialized and the time when the heralding signal is received.

\subsection{Probe state preparation as network application}
Our objective probe state is a global GHZ state of all sensors distributed in different spatial locations. The generation of this GHZ state requires entangling operation which is realized with Bell pairs which are distributed by the quantum network. Conceptually, we may view the process of generating global GHZ state as an application of the quantum network, which consumes distributed bipartite entangled states from the service of quantum network. More specifically, the application can be divided into two parts:
\begin{enumerate}
    \item The involved $N$ sensor nodes in the network collectively establish this application, and decide one out of the $N$ nodes as the center node which corresponds to the center qubit of standard GHZ state preparation circuit. Then the rest $(N-1)$ nodes request bipartite entanglement, i.e. Bell pair, with the center node. 
    \item When bipartite entanglement has been established between the sensor nodes, it can be used to perform CNOT teleportation, or GHZ merging. 
\end{enumerate}

At the end of network simulation, the assembly process of all the distributed Bell pairs into a GHZ state is implemented with the help of functions from the open-source package QuTiP~\cite{johansson2012qutip}. Specifically, we provide two implementations of the process, one using GHZ merging and the other using CNOT teleportation, where we explicitly simulate noisy CNOT gates and noisy single-qubit measurement, together with the classical feedforward correction based on single-qubit measurement outcome(s). Specifically, in this work we implement the assemble-in-the-end protocol, where bipartite entanglement between sensor nodes are requested within a fixed period of time, before they are assembled into a multipartite state. The optimization of the assembly process, e.g. optimal time of assembly in analogy to optimal time of purification~\cite{zang2025entanglement}, can be interesting future work. It is then possible that more than one entangled link between two sensor nodes exist at the end the stage, and in such cases we perform a final round of purification before the assembly, so that between each pair of nodes there is at most one entangled link. The final purification is implemented in a fidelity-aware manner (assuming the capability of estimating entangled state fidelities based on information of system hardware and timing), due to the consideration of the non-universality of entanglement purification~\cite{zang2025no}: We always attempt to purify the two lowest-fidelity entanglement pairs in the entangled state ensemble, and we repeat this process until there are at most one entangled pair left. After this process, we utilize QuTiP functions to simulate the assembly of bipartite entangled states into a GHZ state.

The processes during bipartite entanglement distribution can be easily tracked with SeQUeNCe features, including noisy entanglement generation, swapping and purification, and time dynamics for noisy Bell state in BDS form under quantum memory idling decoherence. Moreover, the network simulation with SeQUeNCe is parallelizable~\cite{wu2024parallel}, so in principle we can simulate bipartite entanglement distribution at large scale, which is physically allowed because in such scenarios entangled states are independent bipartite states. However, when we start to create multipartite entanglement, the necessity of storing larger quantum states and even quantum dynamics simulation will eventually limit the problem scale, especially for DQS where multipartite entangled resource state is desired.

\subsection{Network simulation parameters}
\noindent\textbf{Protocol configuration---}

The probe state generation protocol mainly has one parameter, the cutoff time $T_\mathrm{pc}$: The GHZ state generation protocol consists of generation of individual bipartite entanglement, and the final assembly of the GHZ state. All pairs of sensor nodes will attempt Bell state generation for $T_\mathrm{pc}$, before the final assembly.

\noindent\textbf{Network configuration---}
\begin{enumerate}
    \item Network topology $\mathcal{G}$. A graph that determines how the nodes, including both sensor nodes and potential intermediate repeater nodes in the network core, are connected with physical channels, e.g. optical fibers. 
    \item Channel length $L$. The length of a specific optical fiber, which will determine the loss of photons.
    \item Classical communication time $T_\mathrm{cc}$. We consider that classical communication may require some higher level communication protocols, thus the time is not necessarily equal to length divided by the speed of light.
    \item Number of quantum memories per node available for bipartite entanglement generation attempts $M$.
\end{enumerate}

\noindent\textbf{Hardware configuration---}
\begin{enumerate}
    \item Quantum memory coherence time $\tau_m$.
    \item Quantum memory error pattern. The ratio between three Pauli errors in the continuous time idling channel.
    \item Quantum memory frequency $f_m$. The frequency at which the entanglement generation can be attempted.
    \item Quantum memory efficiency $\eta_m$. The success probability of establishing memory-photon entanglement in each entanglement generation attempt.
    \item Memory cutoff ratio $r_m$. After a certain quantum memory has been initialized for $r_m\tau_m$ time, it will be reset.
    \item Raw Bell state $\rho_0$. The initial state after entanglement generation, before memory decoherence.
    \item Two-qubit gate fidelity $\eta_\mathrm{g}$ and single-qubit measurement fidelity $\eta_\mathrm{m}$.
\end{enumerate}

\subsection{Analytical modeling}\label{sec:analytic}
Here we provide the analytical formulae which facilitate our simulation. They could also be of independent interest to other analytical studies. In the following, for both entanglement swapping and purification, we include imperfection in two-qubit gates and single-qubit measurements~\cite{briegel1998quantum}: $\Tilde{U}_{ij}\rho\Tilde{U}^\dagger_{ij} = pU_{ij}\rho U^\dagger_{ij} + \frac{1-p}{4}I_{ij}\otimes\mathrm{tr}_{ij}\rho$ and $\Tilde{P}_{i=0,1} = \eta|i\rangle\langle i| + (1-\eta)|1-i\rangle\langle 1-i|$, where $\mathrm{tr}_{ij}(\cdot)$ represents partial tracing over qubits $i,j$, $U_{ij}$ is an ideal two-qubit unitary, and $\Tilde{U}_{ij}$ is an imperfect implementation of $U_{ij}$, which has $p$ probability of perfect implementation and $(1-p)$ probability of resulting in depolarizing error. $\Tilde{P}_{i=0,1}$ is the POVM corresponding to imperfect implementation of single-qubit projective measurement $P_i=|i\rangle\langle i|$, which has $\eta$ probability of giving a correct measurement outcome. Some of the results derived below have been used under specific settings in~\cite{zang2023entanglement,chung2025cross} written by part of this work's authors, and the special case of identical measurement and gate error rates for two parties, i.e. $\eta_1=\eta_2=\eta$ and $p_1=p_2=p$, can be found in~\cite{hartmann2007role} without derivation. In the future we may consider other operation error models, which could make analytical derivation of results more complicated and thus explicit simulation of quantum operations~\cite{zang2022simulation} during network simulation might then be necessary.

\subsubsection{Entanglement swapping}\label{sec:analytic_swap}
Consider two Bell diagonal states, one between node A and node B and the other between node B and node C, where node B is the middle swapping station:
\begin{align}
    \rho_\mathrm{A,B_1}(\vec{\lambda})\otimes\rho_\mathrm{B_2,C}(\vec{\lambda'}) = \left(\lambda_1\Phi^+ + \lambda_2\Phi^- + \lambda_3\Psi^+ + \lambda_4\Psi^-\right)_\mathrm{A,B_1}\otimes\left(\lambda_1'\Phi^+ + \lambda_2'\Phi^- + \lambda_3'\Psi^+ + \lambda_4'\Psi^-\right)_\mathrm{B_2,C},
\end{align}
where $|\Phi^\pm\rangle=(|00\rangle\pm|11\rangle)/\sqrt{2}$, $|\Psi^\pm\rangle=(|01\rangle\pm|10\rangle)/\sqrt{2}$, and $\Phi^\pm,\Psi^\pm$ are the projectors onto the corresponding pure Bell states. Entanglement swapping requires performing BSM on the two qubits held by node B, i.e. $\mathrm{B_1,B_2}$. Conditioned on measurement outcome, single-qubit operation is further performed on qubit A (or C) to ensure that the resulting state is in a specific form of Bell state, and here we focus on $\Phi^+$ without loss of generality. Specifically, if the outcomes (in the order of $\mathrm{B_1,B_2}$) are 00 do nothing; if the outcomes are 01, apply $X$ gate; if the outcomes are 10, apply $Z$ gate; if the outcomes are 11, apply $Y$ gate.

By evaluation of CNOT on $\mathrm{B_1,B_2}$ + Hadamard on $\mathrm{B_1}$ for tensor products of two pure Bell states, and the four imperfect POVM's corresponding to four possible measurement outcomes, we have the final output state
\begin{align}
    \rho_{A,C}(\vec{\lambda},\vec{\lambda'}) =& 
    p[\eta_1\eta_2C_I + (1-\eta_1)\eta_2C_X + \eta_1(1-\eta_2)C_Z + (1-\eta_1)(1-\eta_2)C_Y]\Phi^+_\mathrm{A,C}\nonumber\\
    &+ 
    p[(1-\eta_1)\eta_2C_I + \eta_1\eta_2C_X + (1-\eta_1)(1-\eta_2)C_Z + \eta_1(1-\eta_2)C_Y]\Phi^-_\mathrm{A,C}\nonumber\\
    &+
    p[\eta_1(1-\eta_2)C_I + (1-\eta_1)(1-\eta_2)C_X + \eta_1\eta_2C_Z + (1-\eta_1)\eta_2C_Y]\Psi^+_\mathrm{A,C}\nonumber\\
    &+
    p[(1-\eta_1)(1-\eta_2)C_I + \eta_1(1-\eta_2)C_X + (1-\eta_1)\eta_2C_Z + \eta_1\eta_2C_Y]\Psi^-_\mathrm{A,C}\nonumber\\
    &+ 
    \frac{1-p}{4}(\Phi^+_\mathrm{A,C} + \Phi^-_\mathrm{A,C} + \Psi^+_\mathrm{A,C} + \Psi^-_\mathrm{A,C}),
\end{align}
where we have defined $C_I=\lambda_1\lambda_1' + \lambda_2\lambda_2' + \lambda_3\lambda_3' + \lambda_4\lambda_4',\ C_X=\lambda_1\lambda_2' + \lambda_2\lambda_1' + \lambda_3\lambda_4' + \lambda_4\lambda_3',\ C_Y=\lambda_1\lambda_4' + \lambda_4\lambda_1' + \lambda_2\lambda_3' + \lambda_3\lambda_2',\ C_Z=\lambda_1\lambda_3' + \lambda_3\lambda_1' + \lambda_2\lambda_4' + \lambda_4\lambda_2'$, and we assume CNOT error probability $(1-p)$ and single-qubit measurement error probabilities $(1-\eta_1)$ and $(1-\eta_2)$.

\subsubsection{Entanglement purification}\label{sec:analytic_purify}
Under these error models, consider that CNOT gates on both sides have different error probabilities $p_A$ and $p_B$, also measurements have different error probabilities $\eta_A$ and $\eta_B$. Then for two Bell diagonal states as input to the DEJMPS purification protocol, the (un-normalized) output state conditioned on success is:
\begin{align}
    \Tilde{\rho}_{A_1,B_1}(\vec{\lambda},\vec{\lambda'}) =& p_Ap_B\left[(1-\eta_A-\eta_B+2\eta_A\eta_B)(\lambda_1\lambda_1' + \lambda_2\lambda_2') + (\eta_A+\eta_B-2\eta_A\eta_B)(\lambda_1\lambda_3' + \lambda_2\lambda_4')\right]\Phi^+_\mathrm{A_1,B_1}\nonumber\\
    &+ p_Ap_B\left[(1-\eta_A-\eta_B+2\eta_A\eta_B)(\lambda_1\lambda_2' + \lambda_2\lambda_1') + (\eta_A+\eta_B-2\eta_A\eta_B)(\lambda_1\lambda_4' + \lambda_2\lambda_3')\right]\Phi^-_\mathrm{A_1,B_1}\nonumber\\
    &+ p_Ap_B\left[(1-\eta_A-\eta_B+2\eta_A\eta_B)(\lambda_3\lambda_3' + \lambda_4\lambda_4') + (\eta_A+\eta_B-2\eta_A\eta_B)(\lambda_3\lambda_1' + \lambda_4\lambda_2')\right]\Psi^+_\mathrm{A_1,B_1}\nonumber\\
    &+ p_Ap_B\left[(1-\eta_A-\eta_B+2\eta_A\eta_B)(\lambda_3\lambda_4' + \lambda_4\lambda_3') + (\eta_A+\eta_B-2\eta_A\eta_B)(\lambda_3\lambda_2' + \lambda_4\lambda_1')\right]\Psi^-_\mathrm{A_1,B_1}\nonumber\\
    &+ \frac{1 - p_Ap_B}{8}(\Phi^+_\mathrm{A_1,B_1} + \Phi^-_\mathrm{A_1,B_1} + \Psi^+_\mathrm{A_1,B_1} + \Psi^-_\mathrm{A_1,B_1}),
\end{align}
where we use $A_1,B_1$ to denote the two qubits of kept entangled pair and $A_2,B_2$ have been measured (traced out), while the BDS density matrix elements with prime correspond to the measured entangled pair. Then the success probability is just the trace of the above un-normalized state
\begin{align}
    p_s =& p_Ap_B[\eta_A\eta_B+(1-\eta_A)(1-\eta_B)](\lambda_1\lambda_1' + \lambda_2\lambda_2' + \lambda_1\lambda_2' + \lambda_2\lambda_1' + \lambda_3\lambda_3' + \lambda_4\lambda_4' + \lambda_3\lambda_4' + \lambda_4\lambda_3')\nonumber\\
    &+ p_Ap_B[\eta_A(1-\eta_B)+(1-\eta_A)\eta_B](\lambda_1\lambda_3' + \lambda_2\lambda_4' + \lambda_1\lambda_4' + \lambda_2\lambda_3' + \lambda_3\lambda_1' + \lambda_4\lambda_2' + \lambda_3\lambda_2' + \lambda_4\lambda_1') + \frac{1 - p_Ap_B}{2},\label{Eq:succ_prob_imperfect}
\end{align}
from which we can explicitly obtain normalized output state conditioned upon success as $\rho_{A_1,B_1}(\vec{\lambda},\vec{\lambda'}) = \Tilde{\rho}_{A_1,B_1}(\vec{\lambda},\vec{\lambda'})/p_s$. Specifically, the first diagonal element in Bell basis corresponding to $\Phi^+_\mathrm{A_1,B_1}$ is the fidelity
\begin{align}
    F_s = \frac{p_Ap_B\left[\frac{\eta_A\eta_B+(1-\eta_A)(1-\eta_B)}{2}(\lambda_1\lambda_1' + \lambda_2\lambda_2') + \frac{\eta_A(1-\eta_B)+(1-\eta_A)\eta_B}{2}(\lambda_1\lambda_3' + \lambda_2\lambda_4')\right] + \frac{1 - p_Ap_B}{16}}{p_s/2}.
\end{align}

Additionally, we can prove that no matter what input (Bell diagonal) states, CNOT infidelities, and measurement infidelities are, the probability of getting measurement outcome indicating success is always not lower than 1/2:
\begin{proposition}
    Every trial of CNOT based recurrence entanglement purification protocol which has Bell diagonal states as input, using imperfect CNOT and single qubit measurement whose error models are described above, will get measurement outcome indicating success with probability not lower than 1/2.
\end{proposition}
\begin{proof}
    We use Equation~\ref{Eq:succ_prob_imperfect} as the starting point. Notice that $\lambda_1+\lambda_2+\lambda_3+\lambda_4=\lambda_1'+\lambda_2'+\lambda_3'+\lambda_4'=1$ according to normalization and $1\geq\lambda_1,\lambda_1'\geq 1/2$ to ensure that the BDS's are entangled. After some reorganization we have
    \begin{align}
        p_s =& \frac{1}{2} + p_Ap_B[\eta_A(1-\eta_B)+(1-\eta_A)\eta_B]\nonumber\\
        &+ p_Ap_B[ab+(1-a)(1-b)][\eta_A\eta_B+(1-\eta_A)(1-\eta_B)-\eta_A(1-\eta_B)-(1-\eta_A)\eta_B] - \frac{1}{2}p_Ap_B,
    \end{align}
    where we have defined $a\coloneqq\lambda_1+\lambda_2$, $b\coloneqq\lambda_1'+\lambda_2'$, and thus naturally $1\geq a,b\geq 1/2$. Then we have
    \begin{align}
        p_s \geq& \frac{1}{2} + p_Ap_B\left([\eta_A(1-\eta_B)+(1-\eta_A)\eta_B] + \frac{1}{2}[\eta_A\eta_B+(1-\eta_A)(1-\eta_B)-\eta_A(1-\eta_B)-(1-\eta_A)\eta_B] - \frac{1}{2} \right)\nonumber\\
        =& \frac{1}{2} + p_Ap_B\left(\frac{1}{2}[\eta_A\eta_B+(1-\eta_A)(1-\eta_B)+\eta_A(1-\eta_B)+(1-\eta_A)\eta_B] - \frac{1}{2} \right)\nonumber\\
        =& \frac{1}{2} + p_Ap_B\left(\frac{1}{2} - \frac{1}{2} \right) = \frac{1}{2},
    \end{align}
    where for the first inequality we have used the fact that $ab+(1-a)(1-b)\geq 1/2$ for $1\geq a,b\geq 1/2$.
\end{proof}

\twocolumngrid

\bibliography{references}

@article{komar2014quantum,
  title={A quantum network of clocks},
  author={Komar, Peter and Kessler, Eric M and Bishof, Michael and Jiang, Liang and S{\o}rensen, Anders S and Ye, Jun and Lukin, Mikhail D},
  journal={Nature Physics},
  volume={10},
  number={8},
  pages={582--587},
  year={2014},
  publisher={Nature Publishing Group UK London}
}

@article{proctor2017networked,
  title={Networked quantum sensing},
  author={Proctor, TJ and Knott, PA and Dunningham, JA},
  journal={arXiv preprint arXiv:1702.04271},
  year={2017}
}

@article{proctor2018multiparameter,
  title={Multiparameter estimation in networked quantum sensors},
  author={Proctor, Timothy J and Knott, Paul A and Dunningham, Jacob A},
  journal={Physical Review Letters},
  volume={120},
  number={8},
  pages={080501},
  year={2018},
  publisher={APS}
}

@article{huelga1997improvement,
  title={Improvement of frequency standards with quantum entanglement},
  author={Huelga, Susanna F and Macchiavello, Chiara and Pellizzari, Thomas and Ekert, Artur K and Plenio, Martin B and Cirac, J Ignacio},
  journal={Physical Review Letters},
  volume={79},
  number={20},
  pages={3865},
  year={1997},
  publisher={APS}
}

@article{rondin2014magnetometry,
  title={Magnetometry with nitrogen-vacancy defects in diamond},
  author={Rondin, Lo{\"\i}c and Tetienne, Jean-Philippe and Hingant, Thomas and Roch, Jean-Fran{\c{c}}ois and Maletinsky, Patrick and Jacques, Vincent},
  journal={Reports on Progress in Physics},
  volume={77},
  number={5},
  pages={056503},
  year={2014},
  publisher={IOP Publishing}
}

@article{zhang2021distributed,
  title={Distributed quantum sensing},
  author={Zhang, Zheshen and Zhuang, Quntao},
  journal={Quantum Science and Technology},
  volume={6},
  number={4},
  pages={043001},
  year={2021},
  publisher={IOP Publishing}
}

@article{kimble2008quantum,
  title={The quantum internet},
  author={Kimble, H Jeff},
  journal={Nature},
  volume={453},
  number={7198},
  pages={1023--1030},
  year={2008},
  publisher={Nature Publishing Group UK London}
}

@article{wehner2018quantum,
  title={Quantum internet: {A} vision for the road ahead},
  author={Wehner, Stephanie and Elkouss, David and Hanson, Ronald},
  journal={Science},
  volume={362},
  number={6412},
  pages={eaam9288},
  year={2018},
  publisher={American Association for the Advancement of Science}
}

@article{briegel1998quantum,
  title={Quantum repeaters: {The} role of imperfect local operations in quantum communication},
  author={Briegel, H-J and D{\"u}r, Wolfgang and Cirac, Juan I and Zoller, Peter},
  journal={Physical Review Letters},
  volume={81},
  number={26},
  pages={5932},
  year={1998},
  publisher={APS}
}

@article{meignant2019distributing,
  title={Distributing graph states over arbitrary quantum networks},
  author={Meignant, Cl{\'e}ment and Markham, Damian and Grosshans, Fr{\'e}d{\'e}ric},
  journal={Physical Review A},
  volume={100},
  number={5},
  pages={052333},
  year={2019},
  publisher={APS}
}

@article{de2020protocols,
  title={Protocols for creating and distilling multipartite {GHZ} states with {Bell} pairs},
  author={de Bone, S{\'e}bastian and Ouyang, Runsheng and Goodenough, Kenneth and Elkouss, David},
  journal={IEEE Transactions on Quantum Engineering},
  volume={1},
  pages={1--10},
  year={2020},
  publisher={IEEE}
}

@article{shettell2020graph,
  title={Graph states as a resource for quantum metrology},
  author={Shettell, Nathan and Markham, Damian},
  journal={Physical Review Letters},
  volume={124},
  number={11},
  pages={110502},
  year={2020},
  publisher={APS}
}

@inproceedings{fischer2021distributing,
  title={Distributing graph states across quantum networks},
  author={Fischer, Alex and Towsley, Don},
  booktitle={2021 IEEE International Conference on Quantum Computing and Engineering (QCE)},
  pages={324--333},
  year={2021},
  organization={IEEE}
}

@article{avis2023analysis,
  title={Analysis of multipartite entanglement distribution using a central quantum-network node},
  author={Avis, Guus and Rozp{\k{e}}dek, Filip and Wehner, Stephanie},
  journal={Physical Review A},
  volume={107},
  number={1},
  pages={012609},
  year={2023},
  publisher={APS}
}

@article{bugalho2023distributing,
  title={Distributing multipartite entanglement over noisy quantum networks},
  author={Bugalho, Lu{\'\i}s and Coutinho, Bruno C and Monteiro, Francisco A and Omar, Yasser},
  journal={Quantum},
  volume={7},
  pages={920},
  year={2023},
  publisher={Verein zur F{\"o}rderung des Open Access Publizierens in den Quantenwissenschaften}
}

@article{jiang2007distributed,
  title={Distributed quantum computation based on small quantum registers},
  author={Jiang, Liang and Taylor, Jacob M and S{\o}rensen, Anders S and Lukin, Mikhail D},
  journal={Physical Review A},
  volume={76},
  number={6},
  pages={062323},
  year={2007},
  publisher={APS}
}

@article{eisert2000optimal,
  title={Optimal local implementation of nonlocal quantum gates},
  author={Eisert, Jens and Jacobs, Kurt and Papadopoulos, Polykarpos and Plenio, Martin B},
  journal={Physical Review A},
  volume={62},
  number={5},
  pages={052317},
  year={2000},
  publisher={APS}
}

@article{gottesman1999demonstrating,
  title={Demonstrating the viability of universal quantum computation using teleportation and single-qubit operations},
  author={Gottesman, Daniel and Chuang, Isaac L},
  journal={Nature},
  volume={402},
  number={6760},
  pages={390--393},
  year={1999},
  publisher={Nature Publishing Group UK London}
}

@article{browne2005resource,
  title={Resource-efficient linear optical quantum computation},
  author={Browne, Daniel E and Rudolph, Terry},
  journal={Physical Review Letters},
  volume={95},
  number={1},
  pages={010501},
  year={2005},
  publisher={APS}
}

@article{pirker2018modular,
  title={Modular architectures for quantum networks},
  author={Pirker, A and Walln{\"o}fer, J and D{\"u}r, W},
  journal={New Journal of Physics},
  volume={20},
  number={5},
  pages={053054},
  year={2018},
  publisher={IOP Publishing}
}

@article{johansson2012qutip,
  title={{QuTiP}: {An} open-source {Python} framework for the dynamics of open quantum systems},
  author={Johansson, J Robert and Nation, Paul D and Nori, Franco},
  journal={Computer Physics Communications},
  volume={183},
  number={8},
  pages={1760--1772},
  year={2012},
  publisher={Elsevier}
}

@article{wu2021sequence,
  title={{SeQUeNCe}: a customizable discrete-event simulator of quantum networks},
  author={Wu, Xiaoliang and Kolar, Alexander and Chung, Joaquin and Jin, Dong and Zhong, Tian and Kettimuthu, Rajkumar and Suchara, Martin},
  journal={Quantum Science and Technology},
  volume={6},
  number={4},
  pages={045027},
  year={2021},
  publisher={IOP Publishing}
}

@article{liu2021distributed,
  title={Distributed quantum phase estimation with entangled photons},
  author={Liu, Li-Zheng and Zhang, Yu-Zhe and Li, Zheng-Da and Zhang, Rui and Yin, Xu-Fei and Fei, Yue-Yang and Li, Li and Liu, Nai-Le and Xu, Feihu and Chen, Yu-Ao and others},
  journal={Nature Photonics},
  volume={15},
  number={2},
  pages={137--142},
  year={2021},
  publisher={Nature Publishing Group UK London}
}

@article{zhao2021field,
  title={Field demonstration of distributed quantum sensing without post-selection},
  author={Zhao, Si-Ran and Zhang, Yu-Zhe and Liu, Wen-Zhao and Guan, Jian-Yu and Zhang, Weijun and Li, Cheng-Long and Bai, Bing and Li, Ming-Han and Liu, Yang and You, Lixing and others},
  journal={Physical Review X},
  volume={11},
  number={3},
  pages={031009},
  year={2021},
  publisher={APS}
}

@INPROCEEDINGS{zang2023entanglement,
  author={Zang, Allen and Chen, Xinan and Kolar, Alexander and Chung, Joaquin and Suchara, Martin and Zhong, Tian and Kettimuthu, Rajkumar},
  booktitle={IEEE INFOCOM 2023 - IEEE Conference on Computer Communications Workshops}, 
  title={Entanglement Distribution in Quantum Repeater with Purification and Optimized Buffer Time}, 
  year={2023},
  volume={},
  number={},
  pages={1-6}
}

@article{hartmann2007role,
  title={Role of memory errors in quantum repeaters},
  author={Hartmann, Lorenz and Kraus, Barbara and Briegel, H-J and D{\"u}r, W},
  journal={Physical Review A},
  volume={75},
  number={3},
  pages={032310},
  year={2007},
  publisher={APS}
}

@article{saleem2023optimal,
  title={Optimal time for sensing in open quantum systems},
  author={Saleem, Zain H and Shaji, Anil and Gray, Stephen K},
  journal={Physical Review A},
  volume={108},
  number={2},
  pages={022413},
  year={2023},
  publisher={APS}
}

@article{zheng2022preparation,
  title={Preparation of metrological states in dipolar-interacting spin systems},
  author={Zheng, Tian-Xing and Li, Anran and Rosen, Jude and Zhou, Sisi and Koppenh{\"o}fer, Martin and Ma, Ziqi and Chong, Frederic T and Clerk, Aashish A and Jiang, Liang and Maurer, Peter C},
  journal={npj Quantum Information},
  volume={8},
  number={1},
  pages={150},
  year={2022},
  publisher={Nature Publishing Group UK London}
}

@article{deutsch1996quantum,
  title={Quantum privacy amplification and the security of quantum cryptography over noisy channels},
  author={Deutsch, David and Ekert, Artur and Jozsa, Richard and Macchiavello, Chiara and Popescu, Sandu and Sanpera, Anna},
  journal={Physical Review Letters},
  volume={77},
  number={13},
  pages={2818},
  year={1996},
  publisher={APS}
}

@article{liu2020quantum,
  title={Quantum {Fisher} information matrix and multiparameter estimation},
  author={Liu, Jing and Yuan, Haidong and Lu, Xiao-Ming and Wang, Xiaoguang},
  journal={Journal of Physics A: Mathematical and Theoretical},
  volume={53},
  number={2},
  pages={023001},
  year={2020},
  publisher={IOP Publishing}
}

@article{paris2009quantum,
  title={Quantum estimation for quantum technology},
  author={Paris, Matteo GA},
  journal={International Journal of Quantum Information},
  volume={7},
  number={supp01},
  pages={125--137},
  year={2009},
  publisher={World Scientific}
}

@incollection{petz2011introduction,
  title={Introduction to quantum {Fisher} information},
  author={Petz, D{\'e}nes and Ghinea, Catalin},
  booktitle={Quantum Probability and Related Topics},
  pages={261--281},
  year={2011},
  publisher={World Scientific}
}

@article{toth2014quantum,
  title={Quantum metrology from a quantum information science perspective},
  author={T{\'o}th, G{\'e}za and Apellaniz, Iagoba},
  journal={Journal of Physics A: Mathematical and Theoretical},
  volume={47},
  number={42},
  pages={424006},
  year={2014},
  publisher={IOP Publishing}
}

@article{helstrom1969quantum,
  title={Quantum detection and estimation theory},
  author={Helstrom, Carl W},
  journal={Journal of Statistical Physics},
  volume={1},
  pages={231--252},
  year={1969},
  publisher={Springer}
}

@article{monz201114,
  title={14-qubit entanglement: {Creation} and coherence},
  author={Monz, Thomas and Schindler, Philipp and Barreiro, Julio T and Chwalla, Michael and Nigg, Daniel and Coish, William A and Harlander, Maximilian and H{\"a}nsel, Wolfgang and Hennrich, Markus and Blatt, Rainer},
  journal={Physical Review Letters},
  volume={106},
  number={13},
  pages={130506},
  year={2011},
  publisher={APS}
}

@article{song2019generation,
  title={Generation of multicomponent atomic {Schr{\"o}dinger} cat states of up to 20 qubits},
  author={Song, Chao and Xu, Kai and Li, Hekang and Zhang, Yu-Ran and Zhang, Xu and Liu, Wuxin and Guo, Qiujiang and Wang, Zhen and Ren, Wenhui and Hao, Jie and others},
  journal={Science},
  volume={365},
  number={6453},
  pages={574--577},
  year={2019},
  publisher={American Association for the Advancement of Science}
}

@article{omran2019generation,
  title={Generation and manipulation of {Schr{\"o}dinger} cat states in Rydberg atom arrays},
  author={Omran, Ahmed and Levine, Harry and Keesling, Alexander and Semeghini, Giulia and Wang, Tout T and Ebadi, Sepehr and Bernien, Hannes and Zibrov, Alexander S and Pichler, Hannes and Choi, Soonwon and others},
  journal={Science},
  volume={365},
  number={6453},
  pages={570--574},
  year={2019},
  publisher={American Association for the Advancement of Science}
}

@article{ho2019ultrafast,
  title={Ultrafast variational simulation of nontrivial quantum states with long-range interactions},
  author={Ho, Wen Wei and Jonay, Cheryne and Hsieh, Timothy H},
  journal={Physical Review A},
  volume={99},
  number={5},
  pages={052332},
  year={2019},
  publisher={APS}
}

@article{pogorelov2021compact,
  title={Compact ion-trap quantum computing demonstrator},
  author={Pogorelov, Ivan and Feldker, Thomas and Marciniak, Ch D and Postler, Lukas and Jacob, Georg and Krieglsteiner, Oliver and Podlesnic, Verena and Meth, Michael and Negnevitsky, Vlad and Stadler, Martin and others},
  journal={PRX Quantum},
  volume={2},
  number={2},
  pages={020343},
  year={2021},
  publisher={APS}
}

@article{mooney2021generation,
  title={Generation and verification of 27-qubit {Greenberger-Horne-Zeilinger} states in a superconducting quantum computer},
  author={Mooney, Gary J and White, Gregory AL and Hill, Charles D and Hollenberg, Lloyd CL},
  journal={Journal of Physics Communications},
  volume={5},
  number={9},
  pages={095004},
  year={2021},
  publisher={IOP Publishing}
}

@article{zhao2021creation,
  title={Creation of {Greenberger-Horne-Zeilinger} states with thousands of atoms by entanglement amplification},
  author={Zhao, Yajuan and Zhang, Rui and Chen, Wenlan and Wang, Xiang-Bin and Hu, Jiazhong},
  journal={npj Quantum Information},
  volume={7},
  number={1},
  pages={24},
  year={2021},
  publisher={Nature Publishing Group UK London}
}

@article{comparin2022multipartite,
  title={Multipartite entangled states in dipolar quantum simulators},
  author={Comparin, Tommaso and Mezzacapo, Fabio and Roscilde, Tommaso},
  journal={Physical Review Letters},
  volume={129},
  number={15},
  pages={150503},
  year={2022},
  publisher={APS}
}

@article{zhang2024fast,
  title={Fast generation of {GHZ}-like states using collective-spin {XYZ} model},
  author={Zhang, Xuanchen and Hu, Zhiyao and Liu, Yong-Chun},
  journal={Physical Review Letters},
  volume={132},
  number={11},
  pages={113402},
  year={2024},
  publisher={APS}
}

@article{cao2024multi,
  title={Multi-qubit gates and '{Schr{\"o}dinger} cat' states in an optical clock},
  author={Cao, Alec and Eckner, William J and Yelin, Theodor Lukin and Young, Aaron W and Jandura, Sven and Yan, Lingfeng and Kim, Kyungtae and Pupillo, Guido and Ye, Jun and Oppong, Nelson Darkwah and others},
  journal={arXiv preprint arXiv:2402.16289},
  year={2024}
}

@article{yin2024fast,
  title={Fast and Accurate {GHZ} Encoding Using All-to-all Interactions},
  author={Yin, Chao},
  journal={arXiv preprint arXiv:2406.10336},
  year={2024}
}

@article{bollinger1996optimal,
  title={Optimal frequency measurements with maximally correlated states},
  author={Bollinger, John J and Itano, Wayne M and Wineland, David J and Heinzen, Daniel J},
  journal={Physical Review A},
  volume={54},
  number={6},
  pages={R4649},
  year={1996},
  publisher={APS}
}

@article{dur2000classification,
  title={Classification of multiqubit mixed states: {Separability} and distillability properties},
  author={D{\"u}r, Wolfgang and Cirac, J Ignacio},
  journal={Physical Review A},
  volume={61},
  number={4},
  pages={042314},
  year={2000},
  publisher={APS}
}

@article{guhne2010separability,
  title={Separability criteria for genuine multiparticle entanglement},
  author={G{\"u}hne, Otfried and Seevinck, Michael},
  journal={New Journal of Physics},
  volume={12},
  number={5},
  pages={053002},
  year={2010},
  publisher={IOP Publishing}
}

@inproceedings{cao2023distributed,
  title={Distributed quantum sensing network with geographically constrained measurement strategies},
  author={Cao, Yingkang and Wu, Xiaodi},
  booktitle={ICASSP 2023-2023 IEEE International Conference on Acoustics, Speech and Signal Processing (ICASSP)},
  pages={1--5},
  year={2023},
  organization={IEEE}
}

@article{wallman2016noise,
  title={Noise tailoring for scalable quantum computation via randomized compiling},
  author={Wallman, Joel J and Emerson, Joseph},
  journal={Physical Review A},
  volume={94},
  number={5},
  pages={052325},
  year={2016},
  publisher={APS}
}

@article{hashim2021randomized,
  title={Randomized Compiling for Scalable Quantum Computing on a Noisy Superconducting Quantum Processor},
  author={Hashim, Akel and Naik, Ravi K and Morvan, Alexis and Ville, Jean-Loup and Mitchell, Bradley and Kreikebaum, John Mark and Davis, Marc and Smith, Ethan and Iancu, Costin and O’Brien, Kevin P and others},
  journal={Physical Review X},
  volume={11},
  number={4},
  pages={041039},
  year={2021},
  publisher={APS}
}

@inproceedings{zang2022simulation,
  title={Simulation of Entanglement Generation between Absorptive Quantum Memories},
  author={Zang, Allen and Kolar, Alexander and Chung, Joaquin and Suchara, Martin and Zhong, Tian and Kettimuthu, Rajkumar},
  booktitle={2022 IEEE International Conference on Quantum Computing and Engineering (QCE)},
  pages={617--623},
  year={2022},
  organization={IEEE}
}

@article{chaves2013noisy,
  title={Noisy metrology beyond the standard quantum limit},
  author={Chaves, R and Brask, JB and Markiewicz, Marcin and Ko{\l}ody{\'n}ski, J and Ac{\'\i}n, A},
  journal={Physical Review Letters},
  volume={111},
  number={12},
  pages={120401},
  year={2013},
  publisher={APS}
}

@article{chakraborty2019distributed,
  title={Distributed routing in a quantum internet},
  author={Chakraborty, Kaushik and others},
  journal={arXiv preprint arXiv:1907.11630},
  year={2019}
}

@inproceedings{kolar2022adaptive,
  title={Adaptive, Continuous Entanglement Generation for Quantum Networks},
  author={Kolar, Alexander and Zang, Allen and Chung, Joaquin and Suchara, Martin and Kettimuthu, Rajkumar},
  booktitle={IEEE INFOCOM 2022-IEEE Conference on Computer Communications Workshops},
  pages={1--6},
  year={2022},
  organization={IEEE}
}

@article{inesta2023performance,
  title={Performance metrics for the continuous distribution of entanglement in multiuser quantum networks},
  author={I{\~n}esta, {\'A}lvaro G and Wehner, Stephanie},
  journal={Physical Review A},
  volume={108},
  number={5},
  pages={052615},
  year={2023},
  publisher={APS}
}

@article{lutkenhaus1999bell,
  title={Bell measurements for teleportation},
  author={L{\"u}tkenhaus, Norbert and Calsamiglia, John and Suominen, K-A},
  journal={Physical Review A},
  volume={59},
  number={5},
  pages={3295},
  year={1999},
  publisher={APS}
}

@article{dur2005standard,
  title={Standard forms of noisy quantum operations via depolarization},
  author={D{\"u}r, Wolfgang and Hein, Marc and Cirac, J Ignacio and Briegel, H-J},
  journal={Physical Review A},
  volume={72},
  number={5},
  pages={052326},
  year={2005},
  publisher={APS}
}

@article{emerson2007symmetrized,
  title={Symmetrized characterization of noisy quantum processes},
  author={Emerson, Joseph and Silva, Marcus and Moussa, Osama and Ryan, Colm and Laforest, Martin and Baugh, Jonathan and Cory, David G and Laflamme, Raymond},
  journal={Science},
  volume={317},
  number={5846},
  pages={1893--1896},
  year={2007},
  publisher={American Association for the Advancement of Science}
}

@article{dankert2009exact,
  title={Exact and approximate unitary 2-designs and their application to fidelity estimation},
  author={Dankert, Christoph and Cleve, Richard and Emerson, Joseph and Livine, Etera},
  journal={Physical Review A},
  volume={80},
  number={1},
  pages={012304},
  year={2009},
  publisher={APS}
}

@article{zang2025entanglement,
  title={Entanglement Purification in Quantum Networks: Guaranteed Improvement and Optimal Time},
  author={Zang, Allen and Chen, Xin-An and Chitambar, Eric and Suchara, Martin and Zhong, Tian},
  journal={arXiv preprint arXiv:2505.02286},
  year={2025}
}

@article{bennett1996purification,
  title={Purification of noisy entanglement and faithful teleportation via noisy channels},
  author={Bennett, Charles H and Brassard, Gilles and Popescu, Sandu and Schumacher, Benjamin and Smolin, John A and Wootters, William K},
  journal={Physical Review Letters},
  volume={76},
  number={5},
  pages={722},
  year={1996},
  publisher={APS}
}

@article{zang2025no,
  title={No-go theorems for universal entanglement purification},
  author={Zang, Allen and Chen, Xinan and Chitambar, Eric and Suchara, Martin and Zhong, Tian},
  journal={Physical Review Letters},
  volume={134},
  number={19},
  pages={190803},
  year={2025},
  publisher={APS}
}

@article{khatri2021policies,
  title={Policies for elementary links in a quantum network},
  author={Khatri, Sumeet},
  journal={Quantum},
  volume={5},
  pages={537},
  year={2021},
  publisher={Verein zur F{\"o}rderung des Open Access Publizierens in den Quantenwissenschaften}
}

@misc{sequence-github,
    author={},
    title={{SeQUeNCe: Simulator of QUantum Network Communication}},
    url={https://github.com/sequence-toolbox/SeQUeNCe},
    howpublished={\url{https://github.com/sequence-toolbox/SeQUeNCe}},
    year={2023}
}

@article{wu2024parallel,
  title={Parallel simulation of quantum networks with distributed quantum state management},
  author={Wu, Xiaoliang and Kolar, Alexander and Chung, Joaquin and Jin, Dong and Suchara, Martin and Kettimuthu, Rajkumar},
  journal={ACM Transactions on Modeling and Computer Simulation},
  volume={34},
  number={2},
  pages={1--28},
  year={2024},
  publisher={ACM New York, NY}
}

@article{yang2024quantum,
  title={Quantum-enhanced metrology with network states},
  author={Yang, Yuxiang and Yadin, Benjamin and Xu, Zhen-Peng},
  journal={Physical Review Letters},
  volume={132},
  number={21},
  pages={210801},
  year={2024},
  publisher={APS}
}

@article{sekatski2020optimal,
  title={Optimal distributed sensing in noisy environments},
  author={Sekatski, Pavel and W{\"o}lk, Sabine and D{\"u}r, Wolfgang},
  journal={Physical Review Research},
  volume={2},
  number={2},
  pages={023052},
  year={2020},
  publisher={APS}
}

@article{fan2024optimized,
  title={Optimized Distribution of Entanglement Graph States in Quantum Networks},
  author={Fan, Xiaojie and Zhan, Caitao and Gupta, Himanshu and Ramakrishnan, CR},
  journal={arXiv preprint arXiv:2405.00222},
  year={2024}
}

@article{van2024utilizing,
  title={Utilizing probabilistic entanglement between sensors in quantum networks},
  author={Van Milligen, Emily A and Gagatsos, Christos N and Kaur, Eneet and Towsley, Don and Guha, Saikat},
  journal={arXiv preprint arXiv:2407.15652},
  year={2024}
}

@article{shimizu2024simple,
  title={Simple loss-tolerant protocol for {GHZ}-state distribution in a quantum network},
  author={Shimizu, Hikaru and Roga, Wojciech and Elkouss, David and Takeoka, Masahiro},
  journal={arXiv preprint arXiv:2404.19458},
  year={2024}
}

@article{gross2020one,
  title={One from many: {Estimating} a function of many parameters},
  author={Gross, Jonathan A and Caves, Carlton M},
  journal={Journal of Physics A: Mathematical and Theoretical},
  volume={54},
  number={1},
  pages={014001},
  year={2020},
  publisher={IOP Publishing}
}

@article{fadel2023multiparameter,
  title={Multiparameter quantum metrology and mode entanglement with spatially split nonclassical spin ensembles},
  author={Fadel, Matteo and Yadin, Benjamin and Mao, Yuping and Byrnes, Tim and Gessner, Manuel},
  journal={New Journal of Physics},
  volume={25},
  number={7},
  pages={073006},
  year={2023},
  publisher={IOP Publishing}
}

@article{dur2014improved,
  title={Improved quantum metrology using quantum error correction},
  author={D{\"u}r, Wolfgang and Skotiniotis, Michalis and Froewis, Florian and Kraus, Barbara},
  journal={Physical Review Letters},
  volume={112},
  number={8},
  pages={080801},
  year={2014},
  publisher={APS}
}

@article{kessler2014quantum,
  title={Quantum error correction for metrology},
  author={Kessler, Eric M and Lovchinsky, Igor and Sushkov, Alexander O and Lukin, Mikhail D},
  journal={Physical Review Letters},
  volume={112},
  number={15},
  pages={150802},
  year={2014},
  publisher={APS}
}

@article{arrad2014increasing,
  title={Increasing sensing resolution with error correction},
  author={Arrad, Gilad and Vinkler, Yuval and Aharonov, Dorit and Retzker, Alex},
  journal={Physical Review Letters},
  volume={112},
  number={15},
  pages={150801},
  year={2014},
  publisher={APS}
}

@article{zhou2018achieving,
  title={Achieving the {Heisenberg} limit in quantum metrology using quantum error correction},
  author={Zhou, Sisi and Zhang, Mengzhen and Preskill, John and Jiang, Liang},
  journal={Nature Communications},
  volume={9},
  number={1},
  pages={78},
  year={2018},
  publisher={Nature Publishing Group UK London}
}

@article{wolk2020noisy,
  title={Noisy distributed sensing in the {Bayesian} regime},
  author={W{\"o}lk, S and Sekatski, Pavel and D{\"u}r, Wolfgang},
  journal={Quantum Science and Technology},
  volume={5},
  number={4},
  pages={045003},
  year={2020},
  publisher={IOP Publishing}
}

@article{hamann2022approximate,
  title={Approximate decoherence free subspaces for distributed sensing},
  author={Hamann, Arne and Sekatski, Pavel and D{\"u}r, Wolfgang},
  journal={Quantum Science and Technology},
  volume={7},
  number={2},
  pages={025003},
  year={2022},
  publisher={IOP Publishing}
}

@article{hamann2024optimal,
  title={Optimal distributed multi-parameter estimation in noisy environments},
  author={Hamann, Arne and Sekatski, Pavel and D{\"u}r, Wolfgang},
  journal={Quantum Science and Technology},
  volume={9},
  number={3},
  pages={035005},
  year={2024},
  publisher={IOP Publishing}
}

@article{rubio2020quantum,
  title={Quantum sensing networks for the estimation of linear functions},
  author={Rubio, Jes{\'u}s and Knott, Paul A and Proctor, Timothy J and Dunningham, Jacob A},
  journal={Journal of Physics A: Mathematical and Theoretical},
  volume={53},
  number={34},
  pages={344001},
  year={2020},
  publisher={IOP Publishing}
}

@article{zhou2020saturating,
  title={Saturating the quantum {Cram{\'e}r--Rao} bound using {LOCC}},
  author={Zhou, Sisi and Zou, Chang-Ling and Jiang, Liang},
  journal={Quantum Science and Technology},
  volume={5},
  number={2},
  pages={025005},
  year={2020},
  publisher={IOP Publishing}
}

@article{negrin2024efficient,
  title={Efficient Multiparty Entanglement Distribution with {DODAG-X} Protocol},
  author={Negrin, Roberto and Dirnegger, Nicolas and Munizzi, William and Talukdar, Jugal and Narang, Prineha},
  journal={arXiv preprint arXiv:2408.07118},
  year={2024}
}

@article{awschalom2021development,
  title={Development of quantum interconnects (quics) for next-generation information technologies},
  author={Awschalom, David and Berggren, Karl K and Bernien, Hannes and Bhave, Sunil and Carr, Lincoln D and Davids, Paul and Economou, Sophia E and Englund, Dirk and Faraon, Andrei and Fejer, Martin and others},
  journal={PRX Quantum},
  volume={2},
  number={1},
  pages={017002},
  year={2021},
  publisher={APS}
}

@article{barontini2022measuring,
  title={Measuring the stability of fundamental constants with a network of clocks},
  author={Barontini, G and Blackburn, Laura and Boyer, V and Butuc-Mayer, F and Calmet, Xavier and L{\'o}pez-Urrutia, JR Crespo and Curtis, EA and Darquie, B and Dunningham, Jacob and Fitch, NJ and others},
  journal={EPJ Quantum Technology},
  volume={9},
  number={1},
  pages={12},
  year={2022},
  publisher={Springer Berlin Heidelberg}
}

@article{krutyanskiy2023entanglement,
  title={Entanglement of trapped-ion qubits separated by 230 meters},
  author={Krutyanskiy, V and Galli, M and Krcmarsky, V and Baier, S and Fioretto, DA and Pu, Y and Mazloom, A and Sekatski, P and Canteri, M and Teller, M and others},
  journal={Physical Review Letters},
  volume={130},
  number={5},
  pages={050803},
  year={2023},
  publisher={APS}
}

@article{brady2022entangled,
  title={Entangled sensor-networks for dark-matter searches},
  author={Brady, Anthony J and Gao, Christina and Harnik, Roni and Liu, Zhen and Zhang, Zheshen and Zhuang, Quntao},
  journal={PRX Quantum},
  volume={3},
  number={3},
  pages={030333},
  year={2022},
  publisher={APS}
}

@article{knaut2024entanglement,
  title={Entanglement of nanophotonic quantum memory nodes in a telecom network},
  author={Knaut, CM and others},
  journal={Nature},
  volume={629},
  number={8012},
  pages={573--578},
  year={2024},
  publisher={Nature Publishing Group UK London}
}

@article{liu2024creation,
  title={Creation of memory--memory entanglement in a metropolitan quantum network},
  author={Liu, Jian-Long and others},
  journal={Nature},
  volume={629},
  number={8012},
  pages={579--585},
  year={2024},
  publisher={Nature Publishing Group UK London}
}

@article{stolk2024metropolitan,
  title={Metropolitan-scale heralded entanglement of solid-state qubits},
  author={Stolk, Arian J and van der Enden, Kian L and Slater, Marie-Christine and te Raa-Derckx, Ingmar and Botma, Pieter and van Rantwijk, Joris and Biemond, JJ Benjamin and Hagen, Ronald AJ and Herfst, Rodolf W and Koek, Wouter D and others},
  journal={Science Advances},
  volume={10},
  number={44},
  pages={eadp6442},
  year={2024},
  publisher={American Association for the Advancement of Science}
}

@article{gessner2018sensitivity,
  title={Sensitivity bounds for multiparameter quantum metrology},
  author={Gessner, Manuel and Pezz{\`e}, Luca and Smerzi, Augusto},
  journal={Physical Review Letters},
  volume={121},
  number={13},
  pages={130503},
  year={2018},
  publisher={APS}
}

@article{albarelli2019evaluating,
  title={Evaluating the Holevo Cram{\'e}r-Rao bound for multiparameter quantum metrology},
  author={Albarelli, Francesco and Friel, Jamie F and Datta, Animesh},
  journal={Physical Review Letters},
  volume={123},
  number={20},
  pages={200503},
  year={2019},
  publisher={APS}
}

@article{roy2019fundamental,
  title={Fundamental noisy multiparameter quantum bounds},
  author={Roy, Shibdas},
  journal={Scientific Reports},
  volume={9},
  number={1},
  pages={1038},
  year={2019},
  publisher={Nature Publishing Group UK London}
}

@article{qian2019heisenberg,
  title={Heisenberg-scaling measurement protocol for analytic functions with quantum sensor networks},
  author={Qian, Kevin and Eldredge, Zachary and Ge, Wenchao and Pagano, Guido and Monroe, Christopher and Porto, James V and Gorshkov, Alexey V},
  journal={Physical Review A},
  volume={100},
  number={4},
  pages={042304},
  year={2019},
  publisher={APS}
}

@article{qian2021optimal,
  title={Optimal measurement of field properties with quantum sensor networks},
  author={Qian, Timothy and Bringewatt, Jacob and Boettcher, Igor and Bienias, Przemyslaw and Gorshkov, Alexey V},
  journal={Physical Review A},
  volume={103},
  number={3},
  pages={L030601},
  year={2021},
  publisher={APS}
}

@article{bringewatt2021protocols,
  title={Protocols for estimating multiple functions with quantum sensor networks: Geometry and performance},
  author={Bringewatt, Jacob and Boettcher, Igor and Niroula, Pradeep and Bienias, Przemyslaw and Gorshkov, Alexey V},
  journal={Physical Review Research},
  volume={3},
  number={3},
  pages={033011},
  year={2021},
  publisher={APS}
}

@article{ehrenberg2023minimum,
  title={Minimum-entanglement protocols for function estimation},
  author={Ehrenberg, Adam and Bringewatt, Jacob and Gorshkov, Alexey V},
  journal={Physical Review Research},
  volume={5},
  number={3},
  pages={033228},
  year={2023},
  publisher={APS}
}

@article{steinert2010high,
  title={High sensitivity magnetic imaging using an array of spins in diamond},
  author={Steinert, Steffen and Dolde, Florian and Neumann, Philipp and Aird, Andrew and Naydenov, Boris and Balasubramanian, Gopalakrishnan and Jelezko, Fedor and Wrachtrup, Joerg},
  journal={Review of Scientific Instruments},
  volume={81},
  number={4},
  year={2010},
  publisher={AIP Publishing}
}

@article{hall2012high,
  title={High spatial and temporal resolution wide-field imaging of neuron activity using quantum NV-diamond},
  author={Hall, LT and Beart, GCG and Thomas, EA and Simpson, DA and McGuinness, LP and Cole, JH and Manton, JH and Scholten, RE and Jelezko, Fedor and Wrachtrup, J{\"o}rg and others},
  journal={Scientific Reports},
  volume={2},
  number={1},
  pages={401},
  year={2012},
  publisher={Nature Publishing Group UK London}
}

@article{pham2011magnetic,
  title={Magnetic field imaging with nitrogen-vacancy ensembles},
  author={Pham, Linh My and Le Sage, David and Stanwix, Paul L and Yeung, Tsun Kwan and Glenn, D and Trifonov, Alexei and Cappellaro, Paola and Hemmer, Philip R and Lukin, Mikhail D and Park, Hongkun and others},
  journal={New Journal of Physics},
  volume={13},
  number={4},
  pages={045021},
  year={2011},
  publisher={IOP Publishing}
}

@article{humphreys2013quantum,
  title={Quantum enhanced multiple phase estimation},
  author={Humphreys, Peter C and Barbieri, Marco and Datta, Animesh and Walmsley, Ian A},
  journal={Physical Review Letters},
  volume={111},
  number={7},
  pages={070403},
  year={2013},
  publisher={APS}
}

@article{xia2020demonstration,
  title={Demonstration of a reconfigurable entangled radio-frequency photonic sensor network},
  author={Xia, Yi and Li, Wei and Clark, William and Hart, Darlene and Zhuang, Quntao and Zhang, Zheshen},
  journal={Physical Review Letters},
  volume={124},
  number={15},
  pages={150502},
  year={2020},
  publisher={APS}
}

@article{guo2020distributed,
  title={Distributed quantum sensing in a continuous-variable entangled network},
  author={Guo, Xueshi and Breum, Casper R and Borregaard, Johannes and Izumi, Shuro and Larsen, Mikkel V and Gehring, Tobias and Christandl, Matthias and Neergaard-Nielsen, Jonas S and Andersen, Ulrik L},
  journal={Nature Physics},
  volume={16},
  number={3},
  pages={281--284},
  year={2020},
  publisher={Nature Publishing Group UK London}
}

@article{kim2024distributed,
  title={Distributed quantum sensing of multiple phases with fewer photons},
  author={Kim, Dong-Hyun and Hong, Seongjin and Kim, Yong-Su and Kim, Yosep and Lee, Seung-Woo and Pooser, Raphael C and Oh, Kyunghwan and Lee, Su-Yong and Lee, Changhyoup and Lim, Hyang-Tag},
  journal={Nature Communications},
  volume={15},
  number={1},
  pages={266},
  year={2024},
  publisher={Nature Publishing Group UK London}
}

@article{malia2022distributed,
  title={Distributed quantum sensing with mode-entangled spin-squeezed atomic states},
  author={Malia, Benjamin K and Wu, Yunfan and Mart{\'\i}nez-Rinc{\'o}n, Juli{\'a}n and Kasevich, Mark A},
  journal={Nature},
  volume={612},
  number={7941},
  pages={661--665},
  year={2022},
  publisher={Nature Publishing Group UK London}
}

@article{ye2024essay,
  title={Essay: Quantum Sensing with Atomic, Molecular, and Optical Platforms for Fundamental Physics},
  author={Ye, Jun and Zoller, Peter},
  journal={Physical Review Letters},
  volume={132},
  number={19},
  pages={190001},
  year={2024},
  publisher={APS}
}

@article{yue2014quantum,
  title={Quantum-enhanced metrology for multiple phase estimation with noise},
  author={Yue, Jie-Dong and Zhang, Yu-Ran and Fan, Heng},
  journal={Scientific Reports},
  volume={4},
  number={1},
  pages={5933},
  year={2014},
  publisher={Nature Publishing Group UK London}
}

@article{albarelli2022probe,
  title={Probe incompatibility in multiparameter noisy quantum metrology},
  author={Albarelli, Francesco and Demkowicz-Dobrza{\'n}ski, Rafa{\l}},
  journal={Physical Review X},
  volume={12},
  number={1},
  pages={011039},
  year={2022},
  publisher={APS}
}

@article{roberts2020search,
  title={Search for transient variations of the fine structure constant and dark matter using fiber-linked optical atomic clocks},
  author={Roberts, Benjamin M and Delva, Pacome and Al-Masoudi, Ali and Amy-Klein, Anne and Baerentsen, Christian and Baynham, CFA and Benkler, Erik and Bilicki, Slawomir and Bize, Sebastien and Bowden, William and others},
  journal={New Journal of Physics},
  volume={22},
  number={9},
  pages={093010},
  year={2020},
  publisher={IOP Publishing}
}

@article{crawford2021quantum,
  title={Quantum sensing for energy applications: Review and perspective},
  author={Crawford, Scott E and Shugayev, Roman A and Paudel, Hari P and Lu, Ping and Syamlal, Madhava and Ohodnicki, Paul R and Chorpening, Benjamin and Gentry, Randall and Duan, Yuhua},
  journal={Advanced Quantum Technologies},
  volume={4},
  number={8},
  pages={2100049},
  year={2021},
  publisher={Wiley Online Library}
}

@article{toth2007efficient,
  title={Efficient algorithm for multiqudit twirling for ensemble quantum computation},
  author={T{\'o}th, G{\'e}za and Garc{\'\i}a-Ripoll, Juan Jos{\'e}},
  journal={Physical Review A},
  volume={75},
  number={4},
  pages={042311},
  year={2007},
  publisher={APS}
}

@techreport{awschalom2022roadmap,
  title={A roadmap for quantum interconnects},
  author={Awschalom, David D and Bernien, Hannes and Brown, Rex and Clerk, Aashish and Chitambar, Eric and Dibos, Alan and Dionne, Jennifer and Eriksson, Mark and Fefferman, Bill and Fuchs, Greg David and others},
  year={2022},
  institution={Argonne National Laboratory (ANL), Argonne, IL (United States)}
}

@article{muralidharan2016optimal,
  title={Optimal architectures for long distance quantum communication},
  author={Muralidharan, Sreraman and Li, Linshu and Kim, Jungsang and L{\"u}tkenhaus, Norbert and Lukin, Mikhail D and Jiang, Liang},
  journal={Scientific Reports},
  volume={6},
  number={1},
  pages={20463},
  year={2016},
  publisher={Nature Publishing Group UK London}
}

@article{azuma2023quantum,
  title={Quantum repeaters: From quantum networks to the quantum internet},
  author={Azuma, Koji and Economou, Sophia E and Elkouss, David and Hilaire, Paul and Jiang, Liang and Lo, Hoi-Kwong and Tzitrin, Ilan},
  journal={Reviews of Modern Physics},
  volume={95},
  number={4},
  pages={045006},
  year={2023},
  publisher={APS}
}

@article{sahu2023entangling,
  title={Entangling microwaves with light},
  author={Sahu, Rishabh and Qiu, Liu and Hease, William and Arnold, Georg and Minoguchi, Yuri and Rabl, Peter and Fink, Johannes M},
  journal={Science},
  volume={380},
  number={6646},
  pages={718--721},
  year={2023},
  publisher={American Association for the Advancement of Science}
}

@article{ruskuc2024scalable,
  title={Scalable multipartite entanglement of remote rare-earth ion qubits},
  author={Ruskuc, Andrei and Wu, Chun-Ju and Green, Emanuel and Hermans, Sophie LN and Choi, Joonhee and Faraon, Andrei},
  journal={arXiv preprint arXiv:2402.16224},
  year={2024}
}

@article{chitambar2014everything,
  title={Everything you always wanted to know about LOCC (but were afraid to ask)},
  author={Chitambar, Eric and Leung, Debbie and Man{\v{c}}inska, Laura and Ozols, Maris and Winter, Andreas},
  journal={Communications in Mathematical Physics},
  volume={328},
  pages={303--326},
  year={2014},
  publisher={Springer}
}

@inproceedings{ghaderibaneh2023generation,
  title={Generation and distribution of GHZ states in quantum networks},
  author={Ghaderibaneh, Mohammad and Gupta, Himanshu and Ramakrishnan, CR},
  booktitle={2023 IEEE International Conference on Quantum Computing and Engineering (QCE)},
  volume={1},
  pages={1120--1131},
  year={2023},
  organization={IEEE}
}

@article{shettell2022private,
  title={Private network parameter estimation with quantum sensors},
  author={Shettell, Nathan and Hassani, Majid and Markham, Damian},
  journal={arXiv preprint arXiv:2207.14450},
  year={2022}
}

@article{bugalho2024private,
  title={Private and Robust States for Distributed Quantum Sensing},
  author={Bugalho, Lu{\'\i}s and Hassani, Majid and Omar, Yasser and Markham, Damian},
  journal={arXiv preprint arXiv:2407.21701},
  year={2024}
}

@article{hassani2024privacy,
  title={Privacy in networks of quantum sensors},
  author={Hassani, Majid and Scheiner, Santiago and Paris, Matteo GA and Markham, Damian},
  journal={arXiv preprint arXiv:2408.01711},
  year={2024}
}

@article{berry2015quantum,
  title={Quantum Bell-Ziv-Zakai bounds and Heisenberg limits for waveform estimation},
  author={Berry, Dominic W and Tsang, Mankei and Hall, Michael JW and Wiseman, Howard M},
  journal={Physical Review X},
  volume={5},
  number={3},
  pages={031018},
  year={2015},
  publisher={APS}
}

@article{pezze2017optimal,
  title={Optimal measurements for simultaneous quantum estimation of multiple phases},
  author={Pezz{\`e}, Luca and Ciampini, Mario A and Spagnolo, Nicol{\`o} and Humphreys, Peter C and Datta, Animesh and Walmsley, Ian A and Barbieri, Marco and Sciarrino, Fabio and Smerzi, Augusto},
  journal={Physical Review Letters},
  volume={119},
  number={13},
  pages={130504},
  year={2017},
  publisher={APS}
}

@article{ge2018distributed,
  title={Distributed quantum metrology with linear networks and separable inputs},
  author={Ge, Wenchao and Jacobs, Kurt and Eldredge, Zachary and Gorshkov, Alexey V and Foss-Feig, Michael},
  journal={Physical Review Letters},
  volume={121},
  number={4},
  pages={043604},
  year={2018},
  publisher={APS}
}

@article{gorecki2022multiple,
  title={Multiple-phase quantum interferometry: real and apparent gains of measuring all the phases simultaneously},
  author={G{\'o}recki, Wojciech and Demkowicz-Dobrza{\'n}ski, Rafa{\l}},
  journal={Physical Review Letters},
  volume={128},
  number={4},
  pages={040504},
  year={2022},
  publisher={APS}
}

@article{zhang2014quantum,
  title={Quantum metrological bounds for vector parameters},
  author={Zhang, Yu-Ran and Fan, Heng},
  journal={Physical Review A},
  volume={90},
  number={4},
  pages={043818},
  year={2014},
  publisher={APS}
}

@article{munro2015inside,
  title={Inside quantum repeaters},
  author={Munro, William J and Azuma, Koji and Tamaki, Kiyoshi and Nemoto, Kae},
  journal={IEEE Journal of Selected Topics in Quantum Electronics},
  volume={21},
  number={3},
  pages={78--90},
  year={2015},
  publisher={IEEE}
}

@article{guhne2009entanglement,
  title={Entanglement detection},
  author={G{\"u}hne, Otfried and T{\'o}th, G{\'e}za},
  journal={Physics Reports},
  volume={474},
  number={1-6},
  pages={1--75},
  year={2009},
  publisher={Elsevier}
}

@article{seevinck2008partial,
  title={Partial separability and entanglement criteria for multiqubit quantum states},
  author={Seevinck, Michael and Uffink, Jos},
  journal={Physical Review A},
  volume={78},
  number={3},
  pages={032101},
  year={2008},
  publisher={APS}
}

@article{horodecki2009quantum,
  title={Quantum entanglement},
  author={Horodecki, Ryszard and Horodecki, Pawe{\l} and Horodecki, Micha{\l} and Horodecki, Karol},
  journal={Reviews of Modern Physics},
  volume={81},
  number={2},
  pages={865--942},
  year={2009},
  publisher={APS}
}

@article{laskowski2005detection,
  title={Detection of N-particle entanglement with generalized Bell inequalities},
  author={Laskowski, Wies{\l}aw and {\.Z}ukowski, Marek},
  journal={Physical Review A},
  volume={72},
  number={6},
  pages={062112},
  year={2005},
  publisher={APS}
}

@article{nagata2002configuration,
  title={Configuration of Separability and Tests for Multipartite Entanglement in Bell-Type Experiments},
  author={Nagata, Koji and Koashi, Masato and Imoto, Nobuyuki},
  journal={Physical Review Letters},
  volume={89},
  number={26},
  pages={260401},
  year={2002},
  publisher={APS}
}

@article{mermin1990extreme,
  title={Extreme quantum entanglement in a superposition of macroscopically distinct states},
  author={Mermin, N David},
  journal={Physical Review Letters},
  volume={65},
  number={15},
  pages={1838},
  year={1990},
  publisher={APS}
}

@article{peres1996separability,
  title={Separability criterion for density matrices},
  author={Peres, Asher},
  journal={Physical Review Letters},
  volume={77},
  number={8},
  pages={1413},
  year={1996},
  publisher={APS}
}

@article{horodecki2001separability,
  title={Separability of n-particle mixed states: necessary and sufficient conditions in terms of linear maps},
  author={Horodecki, Micha{\l} and Horodecki, Pawe{\l} and Horodecki, Ryszard},
  journal={Physics Letters A},
  volume={283},
  number={1-2},
  pages={1--7},
  year={2001},
  publisher={Elsevier}
}

@inproceedings{ghaderibaneh2022pre,
  title={Pre-distribution of entanglements in quantum networks},
  author={Ghaderibaneh, Mohammad and Gupta, Himanshu and Ramakrishnan, CR and Luo, Ertai},
  booktitle={2022 IEEE International Conference on Quantum Computing and Engineering (QCE)},
  pages={426--436},
  year={2022},
  organization={IEEE}
}

@inproceedings{zhan2025design,
  title={Design and simulation of the adaptive continuous entanglement generation protocol},
  author={Zhan, Caitao and Chung, Joaquin and Zang, Allen and Kolar, Alexander and Kettimuthu, Rajkumar},
  booktitle={2025 International Conference on Quantum Communications, Networking, and Computing (QCNC)},
  pages={127--134},
  year={2025},
  organization={IEEE}
}

@article{huang2025peer,
  title={Peer-to-Peer Distribution of Graph States Across Spacetime Quantum Networks of Arbitrary Topology},
  author={Huang, Yuexun and Ren, Xiangyu and Li, Bikun and Wong, Yat and Liang, Zhiding and Jiang, Liang},
  journal={Proceedings of the ACM on Measurement and Analysis of Computing Systems},
  volume={9},
  number={2},
  pages={1--36},
  year={2025},
  publisher={ACM New York, NY, USA}
}

@inproceedings{zang2024analytical,
  title={Analytical Performance Estimations for Quantum Repeater Network Scenarios},
  author={Zang, Allen and Chung, Joaquin and Kettimuthu, Rajkumar and Suchara, Martin and Zhong, Tian},
  booktitle={2024 IEEE International Conference on Quantum Computing and Engineering (QCE)},
  volume={1},
  pages={1960--1966},
  year={2024},
  organization={IEEE}
}

@article{czekaj2015quantum,
  title={Quantum metrology: Heisenberg limit with bound entanglement},
  author={Czekaj, {\L}ukasz and Przysi{\k{e}}{\.z}na, Anna and Horodecki, Micha{\l} and Horodecki, Pawe{\l}},
  journal={Physical review A},
  volume={92},
  number={6},
  pages={062303},
  year={2015},
  publisher={APS}
}

@article{datta2011quantum,
  title={Quantum metrology with imperfect states and detectors},
  author={Datta, Animesh and Zhang, Lijian and Thomas-Peter, Nicholas and Dorner, Uwe and Smith, Brian J and Walmsley, Ian A},
  journal={Physical Review A—Atomic, Molecular, and Optical Physics},
  volume={83},
  number={6},
  pages={063836},
  year={2011},
  publisher={APS}
}

@article{ouyang2021robust,
  title={Robust quantum metrology with explicit symmetric states},
  author={Ouyang, Yingkai and Shettell, Nathan and Markham, Damian},
  journal={IEEE Transactions on Information Theory},
  volume={68},
  number={3},
  pages={1809--1821},
  year={2021},
  publisher={IEEE}
}

@article{zang2025enhancing,
  title={Enhancing Noisy Quantum Sensing by GHZ State Partitioning},
  author={Zang, Allen and Zheng, Tian-Xing and Maurer, Peter C and Chong, Frederic T and Suchara, Martin and Zhong, Tian},
  journal={arXiv preprint arXiv:2507.02829},
  year={2025}
}

@article{sekatski2017quantum,
  title={Quantum metrology with full and fast quantum control},
  author={Sekatski, Pavel and Skotiniotis, Michalis and Ko{\l}ody{\'n}ski, Janek and D{\"u}r, Wolfgang},
  journal={Quantum},
  volume={1},
  pages={27},
  year={2017},
  publisher={Verein zur F{\"o}rderung des Open Access Publizierens in den Quantenwissenschaften}
}

@article{demkowicz2017adaptive,
  title={Adaptive quantum metrology under general markovian noise},
  author={Demkowicz-Dobrza{\'n}ski, Rafa{\l} and Czajkowski, Jan and Sekatski, Pavel},
  journal={Physical Review X},
  volume={7},
  number={4},
  pages={041009},
  year={2017},
  publisher={APS}
}

@article{huang2024vacuum,
  title={Vacuum beam guide for large scale quantum networks},
  author={Huang, Yuexun and Salces--Carcoba, Francisco and Adhikari, Rana X and Safavi-Naeini, Amir H and Jiang, Liang},
  journal={Physical Review Letters},
  volume={133},
  number={2},
  pages={020801},
  year={2024},
  publisher={APS}
}

@article{pezze2025advances,
  title={Advances in multiparameter quantum sensing and metrology},
  author={Pezz{\`e}, Luca and Smerzi, Augusto},
  journal={arXiv preprint arXiv:2502.17396},
  year={2025}
}

@article{gessner2020multiparameter,
  title={Multiparameter squeezing for optimal quantum enhancements in sensor networks},
  author={Gessner, Manuel and Smerzi, Augusto and Pezz{\`e}, Luca},
  journal={Nature communications},
  volume={11},
  number={1},
  pages={3817},
  year={2020},
  publisher={Nature Publishing Group UK London}
}

@inproceedings{chung2025cross,
  title={Cross-validating quantum network simulators},
  author={Chung, Joaquin and Hajdu{\v{s}}ek, Michal and Benchasattabuse, Naphan and Kolar, Alexander and Singal, Ansh and Soon, Kento Samuel and Teramoto, Kentaro and Zang, Allen and Kettimuthu, Raj and Van Meter, Rodney},
  booktitle={IEEE INFOCOM 2025-IEEE Conference on Computer Communications Workshops (INFOCOM WKSHPS)},
  pages={1--6},
  year={2025},
  organization={IEEE}
}

@article{hu2025optimal,
  title={Optimal scheme for distributed quantum metrology},
  author={Hu, Zhiyao and Zang, Allen and Wang, Jianwei and Zhong, Tian and Yuan, Haidong and Jiang, Liang and Saleem, Zain H},
  journal={arXiv preprint arXiv:2509.18334},
  year={2025}
}

@inproceedings{wu2023qucomm,
  title={Qucomm: Optimizing collective communication for distributed quantum computing},
  author={Wu, Anbang and Ding, Yufei and Li, Ang},
  booktitle={Proceedings of the 56th Annual IEEE/ACM International Symposium on Microarchitecture},
  pages={479--493},
  year={2023}
}

@INPROCEEDINGS{liu2025hardware,
  author={Liu, Ji and Zang, Allen and Suchara, Martin and Zhong, Tian and Hovland, Paul D},
  booktitle={2025 62nd ACM/IEEE Design Automation Conference (DAC)}, 
  title={Hardware-Software Co-design for Distributed Quantum Computing}, 
  year={2025},
  articleno = {66},
  pages={1-6}
}

\end{document}